\documentclass[aps,a4paper,showkeys,nofootinbib,longbibliography,notitlepage,10pt
]{revtex4-2}
\usepackage[utf8]{inputenc}
\usepackage{filecontents}
\usepackage{natbib}
\usepackage{amsmath,amssymb,bm,amsthm}
\usepackage{subfigure}
\usepackage{graphicx}
\usepackage{dcolumn}
\usepackage{bm}
\usepackage[mathlines]{lineno}
\usepackage{booktabs}
\usepackage{color}
\newcounter{one}
\usepackage{url}
\usepackage{ragged2e}%

\usepackage{textcase}

\usepackage{colortbl}
\usepackage{tabularx}
\usepackage{verbatim}
\usepackage{multirow}

\usepackage[T1]{fontenc} 
\usepackage{lmodern}
\usepackage{bbm}
\usepackage[utf8]{inputenc}
\usepackage{amsfonts}
\usepackage{array}

\usepackage{bbm}
\usepackage{enumitem}
\usepackage{umoline}

\usepackage[usenames,svgnames]{xcolor}
\usepackage{natbib}
\usepackage[hyperindex,breaklinks]{hyperref}
\hypersetup{
     colorlinks=true,       		
     linkcolor=Navy,          	
     citecolor=Navy,            
     filecolor=Navy,      		
     urlcolor=Navy,           	
    runcolor=cyan,
 }
\usepackage{graphicx}
\usepackage{amsfonts}
\usepackage{amssymb}
\usepackage{amsmath}
\usepackage{bbm}
\usepackage{enumitem}

\newcommand{\bra}[1]{\langle #1 |}
\newcommand{\ket}[1]{| #1 \rangle}

\newcommand{\tr}[0]{ {\rm tr}}

\newcommand{\ceil}[1]{ \left \lceil #1 \right \rceil}
\newcommand{\half}[1]{{ \rm h}}
\newcommand{\Oorderof}{\mathcal{O}}
\newcommand{\orderof}[1]{\Oorderof(#1)} 

\newcommand{\for}[0]{\quad \textrm{for} \quad}

\newcommand{\dist}{d}

\newcommand{\co}{{\rm c}}

\newcommand{\diam}{{\rm diam}}

\newcommand{\poly}{{\rm poly}}

\newcommand{\ad}{{\rm ad}}

\def\beq{\begin{equation}}
\def\eeq{\end{equation}}
\def\nbeq{\begin{equation*}}
\def\neeq{\end{equation*}}
\def\<{\langle}
\def\>{\rangle}

\def\tr{{\rm tr}}

\newcommand{\Pro}{\phi}
\newcommand{\Mt}{{\rm M}}

\newcommand{\mD}{\mathcal{D}}

\newcommand{\mN}{{\mathcal{N}}}
\newcommand{\mJ}{{\mathcal{J}}}

\newcommand{\mS}{{\mathcal{S}}}
\newcommand{\mQ}{{\mathcal{Q}}}

\newtheorem{theorem}{Theorem}

\newtheorem{lemma}{Lemma}
\newtheorem{corol}[lemma]{Corollary}
\newtheorem{assump}[lemma]{Assumption} 
\newtheorem{definition}{Definition}  
\newtheorem{prop}[lemma]{Proposition} 
\newtheorem{conj}[theorem]{Conjecture} 

\newcommand{\br}[1]{\left( #1 \right)}
\newcommand{\abs}[1]{\left | #1 \right|}
\newcommand{\brr}[1]{\left[ #1 \right]}
\newcommand{\brrr}[1]{\left\{ #1 \right\}}
 \newcommand{\norm}[1]{\left \|  #1 \right \|}

\usepackage{mathtools}
\def\multiset#1#2{\ensuremath{\left(\kern-.3em\left(\genfrac{}{}{0pt}{}{#1}{#2}\right)\kern-.3em\right)}}

\begin{document}


\title{Spectral Small-Incremental-Entangling: Breaking Quasi-Polynomial Complexity Barriers in Long-Range Interacting Systems}


%

\author{Donghoon Kim$^1$}

\author{Yusuke Kimura$^1$}

\author{Hugo Mackay$^{1,2}$}

\author{Yosuke Mitsuhashi$^1$}

\author{Hideaki Nishikawa$^{1,3}$}

\author{Carla Rubiliani$^{1,4}$}

\author{Cheng Shang$^{1}$}

\author{Ayumi Ukai$^{1,5}$}

\author{Tomotaka Kuwahara$^{1,6}$}

\affiliation{$^{1}$
Analytical Quantum Complexity RIKEN Hakubi Research Team, RIKEN Center for Quantum Computing (RQC), Wako, Saitama 351-0198, Japan
}

\affiliation{$^{2}$
Department of Physics, Harvard University, Cambridge, MA 02138, USA
}

\affiliation{$^{3}$
Department of Physics, Kyoto University, Kyoto 606-8502, Japan
}

\affiliation{$^{4}$
Department of Mathematics, University of T\"ubingen, 72076 T\"ubingen, Germany
}

\affiliation{$^{5}$
Research Institute for Mathematical Sciences, Kyoto University, Kyoto 606-8502, Japan
}

\affiliation{$^{6}$
RIKEN Pioneering Research Institute (PRI), Wako, Saitama 351-0198, Japan
}

\begin{abstract}

How the detailed structure of quantum complexity emerges from quantum dynamics remains a fundamental challenge highlighted by advances in quantum simulators and information processing. 
As a celebrated milestone, the Small-Incremental-Entangling (SIE) theorem provides a universal constraint on the rate of entanglement generation. While the SIE theorem limits the total amount of entanglement, it leaves a major open problem in fully characterizing the fine entanglement structure.
Here we introduce the concept of Spectral-Entangling strength, which captures the structural entangling power of an operator, and use it to establish a new spectral SIE theorem: we derive a universal speed limit for R\'enyi entanglement growth at $\alpha \ge 1/2$, revealing a robust $1/s^2$ decay threshold in the entanglement spectrum. 
Remarkably, our bound at $\alpha=1/2$ is both qualitatively and quantitatively optimal, establishing the universal threshold beyond which entanglement growth becomes unbounded. This result exposes the detailed structure of Schmidt coefficients and, in turn, enables rigorous truncation-based error control, providing a quantitative link between entanglement structure and computational complexity. 
Building on these results, our framework establishes a generalized entanglement area law under the adiabatic-path condition, thus extending a central principle of quantum many-body physics to general interactions. As a more practical application, we show that one-dimensional long-range interacting systems admit polynomial bond-dimension approximations for ground states, time-evolved states, and thermal states. This closes the long-standing quasi-polynomial gap and demonstrates that such systems can be simulated with polynomial complexity comparable to short-range models. In particular, by controlling R\'enyi entanglement, we obtain a rigorous, a priori error guarantee for the time-dependent density-matrix renormalization-group algorithm.
Overall, our results extend the scope of the SIE theorem and establish a unified framework that reveals the detailed structure of quantum complexity.
\end{abstract}

\maketitle


\tableofcontents

\section{Introduction}

Understanding how complex structures emerge in the dynamics of quantum systems over time is a fundamental problem in quantum many-body physics. With the rapid progress of quantum simulation platforms and quantum information processing, unveiling such dynamical complexity has become increasingly relevant both theoretically and experimentally. This problem is deeply connected to several essential topics, including the foundation of thermalization in quantum many-body systems~\cite{doi:10.1142/S0217979222300079,PRXQuantum.4.040201}, the characterization of quantum phases of matter~\cite{PhysRevB.82.155138,PhysRevB.84.235128,202400196}, the efficiency guarantee of quantum many-body algorithms~\cite{0034-4885-75-2-022001,HHKL2021,PhysRevX.11.011047,PRXQuantum.2.040331}, and the theoretical framework of entanglement itself, as well as the role of entanglement as a fundamental resource in quantum information science~\cite{RevModPhys.80.517,RevModPhys.82.277,Frerot_2023}.
Among the known universal principles governing quantum dynamics, the Lieb--Robinson bound stands out as a cornerstone result~\cite{ref:LR-bound72,PhysRevLett.97.050401,Chen_2023}. It provides an effective light-cone structure in non-relativistic quantum systems, setting a fundamental limit on how quickly information and entanglement can propagate through local interactions.
However, the Lieb--Robinson bound inherently relies on geometric locality, and its applicability becomes significantly limited when dealing with systems with long-range interactions or those defined on infinite-dimensional graphs~\cite{Kuwahara_2016_njp,chen2019finite,PhysRevX.10.031010,PhysRevLett.127.160401,PhysRevX.11.031016}. In such cases, the standard notions of locality and causal structure no longer suffice to capture the nontrivial entanglement dynamics.

Another well-known information-theoretic constraint that does not rely on geometric locality is the \textit{Small-Incremental-Entangling} (SIE) theorem~\cite{PhysRevA.76.052319}. This theorem considers a bipartition of a quantum system into subsystems $A$ and $B$, and asserts that the \textit{rate} at which entanglement entropy is generated between them can be upper bounded solely in terms of the interaction strength and the Hilbert space dimensions of the directly coupled subsystems $A_0 \subset A$ and $B_0 \subset B$.
In systems with spatial locality, the SIE bound can be derived from the Lieb--Robinson bound~\cite{PhysRevLett.97.050401}. However, its general validity across arbitrary quantum systems was first conjectured by Kitaev and Bravyi, and later fully proven in 2013~\cite{PhysRevLett.111.170501,10.1063/1.4901039,PhysRevA.92.022311,Marien2016}.
Because the rate of entanglement growth depends only on the interaction across the boundary between $A$ and $B$, the SIE theorem is often referred to as the \textit{Dynamical Area Law}. This term emphasizes the fact that entanglement generation in time is governed by the surface—rather than the volume—of the interacting regions.
The SIE theorem has served as a key methodological tool in a wide range of fundamental problems, such as supporting the area law conjecture for noncritical ground states in higher-dimensional systems~\cite{PhysRevLett.111.170501,michalakis2012}, analyzing complexity growth in quantum circuits~\cite{PhysRevB.98.035118,PhysRevLett.127.020501}, and characterizing measurement-induced phase transitions~\cite{PhysRevLett.128.010603}. These applications demonstrate the power and universality of the SIE framework. Nevertheless, as we will discuss in the next section, the entanglement generation rate alone is often insufficient to fully capture the complexity of quantum many-body dynamics.

Despite the theoretical strength of existing approaches such as the Lieb--Robinson bound and the SIE theorem, most of these frameworks focus primarily on the \textit{amount} of information---for example, how much entanglement is generated---rather than accessing the \textit{structure} of that information.
The Lieb--Robinson bound provides no insight into the internal structure within the light cone, and the SIE theorem says nothing about the detailed features of the entanglement being generated. This shift from quantifying the amount to characterizing the structure represents a major open challenge in modern quantum many-body physics.
In particular, it is known that entanglement entropy alone is insufficient to meaningfully describe the quantitative complexity of quantum systems~\cite{PhysRevLett.100.030504,PRXQuantum.1.010304}. Therefore, if we wish to utilize the SIE theorem as a foundational tool for understanding quantum many-body complexity, and to apply it to more intricate scenarios such as quantum algorithms or non-equilibrium phenomena, it becomes essential to introduce new formalisms that can capture the structure of entanglement itself.

In this work, we move beyond these limitations by establishing a spectral formulation of entanglement growth. Our framework not only provides rigorous dynamical constraints on the entanglement spectrum itself, but also achieves both qualitative and quantitative optimality: at the critical threshold, we establish a universal upper bound corresponding to a $1/s^2$ decay of squared Schmidt coefficients, which is sharp and cannot be improved.
This perspective sharpens the dynamical area law into a more fine-grained structural principle and extends its reach to broader consequences, including generalized area laws of gapped ground states, efficient matrix-product-state approximations for long-range systems, and rigorous a priori performance guarantees for numerical algorithms such as t-DMRG (time-dependent density-matrix renormalization group). 
Conceptually, the ability to address the entanglement spectrum directly opens the door to a more refined understanding of quantum complexity, enabling applications that go well beyond entropy-based diagnostics.

\begin{figure*}[ttt]
  \centering
  \includegraphics[width=1\textwidth]{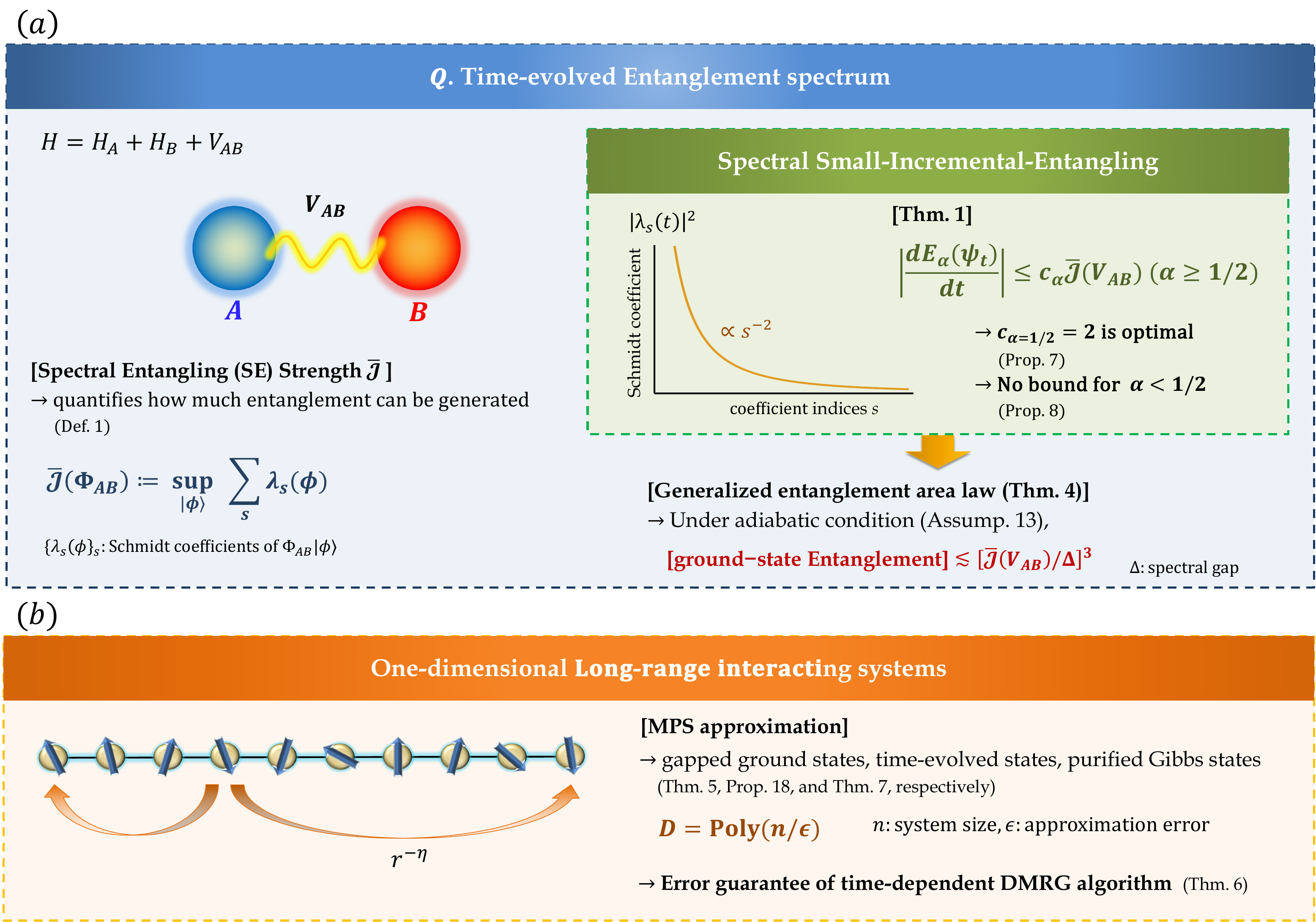}
\caption{Overview of the main results. (a) \emph{Spectral Small-Incremental-Entangling (SIE):} 
Introducing the spectral-entangling strength $\bar J$, we establish that R\'enyi entanglement entropies satisfy 
$\lvert dE_\alpha/dt \rvert \le c_\alpha \bar J(V_{AB})$ for $\alpha \ge 1/2$, with the constant $c_{1/2}=2$ being sharp, while for $\alpha < 1/2$ the growth rate can diverge. 
The optimal $\alpha=1/2$ case corresponds to an entanglement spectrum with Schmidt coefficients decaying as $1/s^{2}$, as illustrated in the panel. 
(b) \emph{Applications to 1D long-range interacting systems:} 
When interactions decay faster than $r^{-2}$, ground states, real-time dynamics, and Gibbs states admit polynomial-bond-dimension MPS/MPO approximations. 
In addition, our results yield rigorous error guarantees for t-DMRG simulations.}
  \label{fig:Overview}
\end{figure*}

\section{Overview of the main results}

\subsection{Spectral-Entangling (SE) strength}

We consider a general bipartite quantum system composed of subsystems $A$ and $B$, where the Hamiltonian takes the form $H = H_A + H_B + V_{AB}$. The time evolution generated by $e^{-iHt}$ induces nontrivial changes in the entanglement spectrum, namely the structure of the Schmidt coefficients, and our goal is to understand how this spectrum evolves. To address this, we introduce a quantity called the spectral-entangling (SE) strength, defined for a general operator $\Phi_{AB}$. This quantity plays a central role in the present work, as it quantifies the entangling power of an operator. 
\vspace{4mm} \\  
{\bf [SE strength (Def.~\ref{Def:Interaction_strength}, informal)]}
For any product state $\ket{\phi}$, let us denote the Schmidt coefficients of $\Phi_{AB}\ket{\phi}$ by $\{\lambda_s(\phi)\}_s$. 
Then, the operator $\Phi_{AB}$ has the SE strength $\bar{\mJ}(\Phi_{AB})$ which is defined as supremum over product states: $\sup_{\phi} \sum_s \lambda_s(\phi)$.  
\vspace{4mm} \\  
Intuitively, this quantity measures the maximal ability of $\Phi_{AB}$ to generate spectral weight across the bipartition.
A simple characterization can be obtained when $\Phi_{AB}$ admits a decomposition $\Phi_{AB} = \sum_j J_j \Phi_{A,j} \otimes \Phi_{B,j}$ with $\|\Phi_{A,j}\| = \|\Phi_{B,j}\| = 1$, in which case one finds $\bar{\mJ}(\Phi_{AB}) \sim \sum_j |J_j|$.

\subsection{Spectral Small-Incremental-Entangling (SIE)}

This framework allows us to investigate how the entanglement spectrum $\{\abs{\lambda_s(t)}^2\}$ of an initially product state $\ket{\phi}$ evolves under time evolution into $\ket{\phi_t}$. 
Informally, one can state that the structure of the entanglement spectrum follows the scaling law 
$$\abs{\lambda_s(t)}^2 \sim e^{\mathcal{O}\brr{\bar{\mJ}(V_{AB})t }}/s^2.$$ 
In general, this decay law cannot be further improved without imposing additional geometric assumptions such as finite-dimensionality or short-range interactions~\cite{PhysRevLett.97.157202,PhysRevX.11.011047,PRXQuantum.2.040331}. The result can be derived from an upper bound on the growth rate of the $\alpha$-R\'enyi entanglement $E_\alpha$, as defined precisely in Eq.~\eqref{Def_Renyi_ent}. In order to emphasize the spectral nature of the bound, we refer to this extension as the Spectral Small-Incremental-Entangling (Spectral SIE) theorem. 
\vspace{4mm} \\  
{\bf [Theorem~\ref{thm:Renyi_SIE} (informal)]} For $\alpha \ge 1/2$, the generation rate of the $\alpha$-R\'enyi entanglement is bounded from above by $c_\alpha \bar{\mJ}(V_{AB})$, where $c_\alpha$ is an $O(1)$ constant satisfying $c_{1/2}=2$, $c_{1}=4/e$ and $c_\infty=2$.
\vspace{4mm} \\  
An important observation is that $\alpha = 1/2$ constitutes a universal threshold. For $\alpha < 1/2$, the entanglement generation rate becomes unbounded, as proven in Proposition~\ref{prop:optimality_2}. Moreover, for $\alpha = 1/2$ the constant $c_{\alpha=1/2}$ is not only universal but also quantitatively optimal, as shown in Proposition~\ref{prop:optimality_1}. A significant advantage of considering the R\'enyi entanglement with $\alpha < 1$ is that it enables rigorous error analysis in terms of Schmidt rank truncation~\cite{PhysRevLett.100.030504,PRXQuantum.1.010304}, 
thereby providing a quantitative estimate of how entanglement structure affects the computational complexity of simulating quantum dynamics.

While previous attempts have been made to extend the SIE theorem, these were restricted to the case $\alpha > 1$~\cite{10.1063/1.5037802,PhysRevA.109.042404}. In contrast, the present establishment of the spectral SIE theorem crucially relies on the use of the spectral-entangling strength $\bar{\mJ}(V_{AB})$ rather than the operator norm $\|V_{AB}\|$. This distinction is essential, and its role in the optimality of the spectral SIE bounds is further discussed in Sec.~\ref{Sec:optimatlity_spec:SIE}.

\subsection{Low-Schmidt-rank approximation of general operators}

As discussed above, it has been established that time-evolved states admit efficient low-rank approximations. A natural question is whether a similar low-rank approximation exists at the level of the time-evolution operator itself. 
In conclusion, approximating an operator with respect to the operator norm is significantly more challenging.
To address this, we demonstrate that there exists a complexity separation between Schmidt-rank approximability of time-evolved quantum states and that of the corresponding time-evolution unitary.
\vspace{4mm} \\  
{\bf [Proposition~\ref{No_go_theorem_approx} (informal)]} There exist Hamiltonians $H_{AB}$ with $\bar{\mJ}(V_{AB}) = 1$ for which the unitary $e^{-iH_{AB}t}$ does not admit any low-rank approximation in a certain (but constant) time regime. 
\vspace{4mm} \\  
Note that $\bar{\mJ}(V_{AB}) = 1$ means that any time-evolved quantum state is well-approximated by a low-Schmidt-rank state (see also Proposition~\ref{Prop:MPO_approximation}), while low-rank operators cannot approximate the time-evolution operator.  
This negative result is derived from classical lower bounds on Kolmogorov width~\cite{Kas74,Gluskin1986,Foucart2013}.

Although a general relationship between entanglement generation and operator-level low-rank approximability remains elusive, one possible direction is proposed in Conjecture~\ref{conj:operator_approx}, which is formulated using a generalized version of the SE strength (see Def.~\ref{Def:Interaction_strength_renyi}). Furthermore, we provide sufficient conditions for the existence of operator low-rank approximations, which are formalized in Theorem~\ref{thm:Schmidt_rank_truncation}. These results highlight the subtle but fundamental difference between approximability at the state level and at the operator level, and suggest that the spectral-entangling framework may offer a path toward resolving this gap.


\subsection{Generalized entanglement area law} 

The entanglement area law refers to the property that, for a ground state of a finite-dimensional lattice system partitioned into two regions $A$ and $B$, the entanglement entropy scales with the size of the boundary between the two regions~\cite{RevModPhys.82.277,Hastings_2007}. In particular, the area law conjecture argues that the area law always holds for noncritical ground states with a spectral gap, and it remains one of the central open problems in quantum many-body physics~\cite{PhysRevB.85.195145,arad2013area,Arad2017,brandao2013area,10.1145/3519935.3519962,Kuwahara2020arealaw,kim2024Boson,10.1145/3618260.3649612,10.1063/5.0167353,10.1063/5.0239332}. The notion of an area law can be extended to more general bipartite interaction settings~\cite{aharonov2014local}. Specifically, when a system is divided into $A$ and $B$, the Schmidt rank of the interaction operator between $A$ and $B$ can be interpreted as an effective boundary size, thereby allowing a generalized formulation of the area law. For example, in a four-partite setup $(A_1,A_0,B_0,B_1)$ where the interaction exists only between $A_0$ and $B_0$, the relevant boundary size is determined by the dimensions of $A_0$ and $B_0$. Typical examples include generalized area laws for local Hamiltonians defined on graphs, which have been widely discussed in the literature. According to the conventional SIE theorem~\cite{PhysRevLett.111.170501,10.1063/1.4901039}, the generalized area law can be shown to hold for entanglement dynamics even in such general settings. It is also widely known that generalized area laws for mutual information hold for quantum Gibbs states at finite temperature~\cite{PhysRevLett.100.070502,kim2024thermal}. However, whether the same statement applies to gapped ground states has long been debated, and this question was resolved in the negative by Aharonov et al. in 2014~\cite{aharonov2014local}. At present, aside from special classes~\cite{kim2025gen}, it remains a fundamental open problem to determine the precise conditions under which a generalized area law holds.

In this work, we establish that the generalized area law holds under the assumption that an adiabatic path exists~\cite{PhysRevB.72.045141}. Given a parameter-dependent Hamiltonian $H(s)$, we define an adiabatic path as a continuous trajectory in parameter space that connects a reference Hamiltonian with a known area-law-satisfying ground state to a target Hamiltonian, while maintaining a finite spectral gap throughout. This condition is standard in proofs of the area law for spatially local systems~\cite{michalakis2012,PhysRevLett.111.170501,PhysRevLett.113.197204}, where it is combined with the Lieb--Robinson bound, and where the SIE theorem provides a key estimate on entanglement growth along the adiabatic evolution. In more general settings, such as systems defined on infinite-dimensional graphs or general two-body systems, the Lieb--Robinson bound is unavailable, and consequently, the adiabatic continuation operator lacks a simple geometric structure. In these regimes, the conventional SIE theorem ceases to be effective. To overcome this difficulty, we combine the adiabatic theorem~\cite{10.1063/1.2798382} with two key tools: the Approximate-Ground-State-Projection formalism~\cite{arad2013area} and the spectral SIE framework developed in this work. This combination enables us to prove a generalized area law without imposing any assumption of locality.

Our main result can be summarized as follows. 
\vspace{4mm} \\  
{\bf [Theorem~\ref{thm:generalized_area_law} (Informal)]} If the target ground state is connected to a trivial ground state that satisfies the generalized area law via an adiabatic path (Assumption~\ref{assump:Boundary-adiabatic path}), then the entanglement entropy between $A$ and $B$ is bounded from above by $[\bar{\mJ}(V_{AB})/\Delta]^3$, where $\Delta$ denotes the spectral gap.\vspace{4mm} \\  
 Furthermore, analogous statements can be made regarding low-Schmidt-rank approximations of the ground state for the bipartition. To the best of our knowledge, this constitutes the first proof of a generalized area law for ground states in the most general settings without assuming geometric locality.

\subsection{Polynomial complexity of the long-range interacting systems}

Up to this point, our results have addressed the most general settings, including physical systems without any underlying geometric structure. Nevertheless, the framework we have developed also applies in scenarios where geometric structure is present. Among such cases, a particularly important and nontrivial application concerns the computational complexity of simulating quantum systems with long-range interactions. These systems are typically characterized by interactions that decay with spatial distance $r$ as a power law $r^{-\eta}$ with $\eta>0$, yet a comprehensive understanding of their entanglement properties and simulation complexity remains elusive.

One of the central open questions in this domain is whether the simulation complexity of general long-range interacting systems can be improved from quasi-polynomial time $e^{\mathrm{polylog}(n)}$ to polynomial time $e^{\log(n)}$, where $n$ denotes the system size. Achieving such an improvement would suggest that long-range interacting systems and short-range systems belong to the same computational complexity class, a possibility that is both highly counterintuitive and profoundly nontrivial.

The difficulty of this problem can be illustrated by considering the approximation of operators such as $e^{-iHt}$ or $e^{-\beta H}$ by polynomial functions of the Hamiltonian $H$. When $H^m$ is represented as a matrix product operator (MPO), the bond dimension generally grows with $m$ as $n^{\orderof{m}}$ for a long-range interacting Hamiltonian~\cite{PhysRevLett.134.190404}. If $m$ depends even weakly on the system size---for instance, if $m \propto \log\log(n)$---the required bond dimension quickly exceeds polynomial bounds. This stands in sharp contrast to the case of short-range Hamiltonians, for which an MPO representation with constant bond dimension $\orderof{1}$ suffices. Consequently, physical quantities that cannot be captured by low-degree polynomial approximations in $H$ inevitably lead to quasi-polynomial simulation complexity in the presence of long-range interactions.

Building on our framework, we demonstrate that one-dimensional long-range interacting systems can, in fact, be represented as matrix product operators (MPOs) with polynomial bond dimension. The essential observation is that the SE strength admits an $\orderof{1}$ upper bound in such systems provided the interaction decays faster than $r^{-2}$ (Lemma~\ref{lem:Long-range_SE_Strength}):
\vspace{4mm} \\  
{\bf [Lemma~\ref{lem:Long-range_SE_Strength} (Informal)]}
Let $\tilde{J}$ be an $\orderof{1}$ upper bound for the SE strength of the boundary interaction for any bipartition of the 1D chain. Then, $\tilde{J}=\orderof{1}$ for $\eta>2$. 
\vspace{4mm} \\  
 This bound ensures that the entanglement structure generated by long-range interactions is sufficiently constrained to allow polynomially efficient tensor-network representations.
In this work, we consider three fundamental classes of quantum states: ground states, time-evolved states, and thermal equilibrium states. For these states, only quasi-polynomial simulation complexity had previously been established~\cite{Kuwahara2020arealaw,PhysRevLett.134.190404,HighT_Alhambra}. Our results thus improve the known bounds by showing that, under the above decay condition, these physically relevant states can be captured within polynomial complexity, thereby closing the gap between long-range and short-range interacting systems in one dimension.

\subsubsection{Ground states}

We begin by considering matrix product state (MPS) approximations of gapped ground states in one-dimensional long-range interacting systems. For such ground states, the entanglement area law has already been established~\cite{Kuwahara2020arealaw,liu2025ent}, and it is known that the required bond dimension scales as $e^{\log^{5/2}(n)}$, thereby imposing a quasi-polynomial overhead. This limitation arises because most existing approaches rely on constructing an Approximate Ground State Projector (AGSP) using the Chebyshev polynomials of the Hamiltonian~\cite{arad2013area}. As discussed earlier, polynomial approximations of long-range Hamiltonians typically induce quasi-polynomial complexity, rendering high-precision approximations prohibitively costly. 
In this work, we overcome this bottleneck by employing time-evolution operators together with the spectral SIE framework. In particular, we estimate the SE strength of a Gaussian-filter AGSP and demonstrate that such filters yield asymptotically better approximations than conventional polynomial-based AGSPs. As a result, we prove that the ground state can be approximated by an MPS with polynomial bond dimension. 
\vspace{4mm} \\  
{\bf [Theorem~\ref{thnm:Gs_approx} (Informal)]}
Given a ground state $\ket{\Omega}$ with spectral gap $\Delta$, one can approximate it by an MPS up to an error $\epsilon$, by choosing the bond dimension $D$ such that $D=(n/\epsilon)^{\orderof{\tilde{J}/\Delta}}$. 
\vspace{4mm} \\  
This result has far-reaching implications for computational complexity theory. In particular, it provides a rigorous proof that the Local Hamiltonian Problem for one-dimensional long-range interacting systems belongs to the complexity class \textsf{NP}.

\subsubsection{Quantum dynamics}

In a similar manner, the spectral SIE formalism enables polynomially efficient approximation of time-evolved states by matrix product states (MPS). 
\vspace{4mm} \\  
{\bf [Proposition~\ref{Prop:MPO_approximation} (Informal)]} 
For any product state $\ket{\phi}$, the time-evolved state can be approximated by an MPS up to an error $\epsilon$, where the bond dimension $D$ scales as $e^{\orderof{\tilde{J}t}} (n/\epsilon)^2$. 
\vspace{4mm} \\  
While this establishes efficient approximability of time-evolved states, the time complexity of simulating quantum dynamics on a classical computer remains considerably more challenging. Previous approaches have estimated this complexity by explicitly constructing matrix product operator (MPO) representations of the time-evolution operator $e^{-iHt}$~\cite{PhysRevLett.97.157202,PRXQuantum.2.040331,PhysRevX.11.011047,PhysRevLett.134.190404}. At present, the best known methods achieve only quasi-polynomial simulation cost~\cite{PhysRevLett.134.190404} for long-range interacting systems. 

In our approach, instead of considering the explicit MPO construction of the time-evolution operator itself, we certify the accuracy of time-evolution simulation using the time-dependent Density Matrix Renormalization Group (t-DMRG) algorithm~\cite{PhysRevLett.93.040502,PhysRevLett.93.076401}. This strategy overcomes the expensive task of explicitly constructing the full MPO representation. 
A key difficulty, previously highlighted by Osborne~\cite{PhysRevLett.97.157202}, is that naive error analysis of t-DMRG suffers from exponential error amplification at each time step, undermining reliability (Sec.~\ref{sec:Err_gua_challenge}). 
In particular, while truncation-based diagnostics (e.g., the discarded weight) provide \emph{a posteriori} accuracy checks during the simulation~\cite{PhysRevLett.93.040502}, a general \emph{a priori} upper bound on the total error derived solely from the dynamics of a given local Hamiltonian has not been established for standard t-DMRG.

In this work, we address this problem by monitoring the $(1/2)$-R\'enyi entanglement proxy of the approximated state at each step, thereby guaranteeing controlled error propagation throughout the simulation:
\vspace{4mm} \\  
{\bf [Theorem~\ref{Thm:error_efficiency_t-DMRG} (Informal)]}  
For the time-dependent DMRG algorithm with bond dimension $D$, the total simulation error in the t-DMRG algorithm is upper-bounded by $\epsilon$ using the bond dimension of order of $e^{\orderof{\tilde{J} t}} n^5/\epsilon^4$, where the number of time steps is chosen appropriately.
\vspace{4mm} \\  
We believe that the same analytical method may serve as a general paradigm for certifying the accuracy of other DMRG-based algorithms in the study of complex quantum systems.

\subsubsection{Quantum Gibbs states}

A natural question is whether similar techniques apply to imaginary time evolution, which is fundamental for studying thermal equilibrium states. This setting, however, presents intrinsic challenges. The key difficulty is that imaginary time evolution does not preserve the norm of the quantum state. Because imaginary-time evolution is non-unitary, even a product-form propagator (e.g., $e^{-\tau H_A}\otimes \hat{1}_B$) can, after normalization, significantly modify the entanglement spectrum across $A|B$.
This illustrates in a striking way why complex-time evolution is qualitatively more difficult than real-time evolution~\cite{PhysRevLett.93.207204,PhysRevX.11.011047}. Consequently, a direct application of the spectral SIE framework allows efficient MPO (or purified MPS) approximations only in high-temperature regimes, where the thermal state remains close to the identity and entanglement is naturally limited. 

To overcome this limitation, instead of relying on spectral SIE, we employ the operator low-rank approximation result established in Theorem~\ref{thm:Schmidt_rank_truncation} and adapt it to quantum Gibbs states. In this approach, unlike in the case of real-time evolution, only the \textit{existence} of an efficient MPS approximation can be guaranteed. 
\vspace{4mm} \\  
{\bf [Theorem~\ref{poly_approx:MPO_gibbs} (Informal)]}
A purified one-dimensional Gibbs state can be approximated to error $\epsilon$ by an MPS whose bond dimension scales as $(n/\epsilon)^{\orderof{\beta}}$. 
\vspace{4mm} \\ 
This result qualitatively improves upon the previously known quasi-polynomial complexity bounds for quantum Gibbs states~\cite{PhysRevLett.134.190404,HighT_Alhambra}. 
By contrast, designing an explicit and efficient algorithm to generate such an MPO remains open.

\section{Setup}

\subsection{General framework}

We consider a general many-body qudit system, with the total set of qudits denoted by $\Lambda$. Each qudit is assumed to have a finite Hilbert space dimension. We bipartition the system into two subsystems $A$ and $B$ such that $\Lambda = A \sqcup B$. 
For notational simplicity, $A \sqcup B$ is frequently
abbreviated as $AB$, particularly when it appears in a subscript.
Additionally, we often use $A_0, A_1$ and $B_0, B_1$ to denote
decompositions of $A$ and $B$ into two subsets, such as $A=A_0\sqcup A_1$ and $B=B_0\sqcup B_1$.
For any subset $X \subseteq \Lambda$, we denote the Hilbert space dimension by $\mathcal{D}_X$.
A product state $\ket{\psi_{A}} \otimes \ket{\psi_{B}}$ is often abbreviated
as $\ket{\psi_{A},\psi_{B}}$.

We consider an arbitrary Hamiltonian supported on $A$ and $B$, which takes the form
\begin{align}
\label{Ham_H_V_AB}
H = H_A + H_B + V_{AB},
\end{align}
where $H_A$ and $H_B$ act only on $A$ and $B$, respectively, and $V_{AB}$ represents the interaction term between $A$ and $B$. At this stage, we do not impose any specific locality constraints on the Hamiltonian.
In studying the entanglement generation, we typically assume that $H_A$ and $H_B$ may take arbitrary form (e.g., their norms are unbounded), while only the boundary interaction $V_{AB}$ is restricted (see Def.~\ref{Def:Interaction_strength} below).

For a given bipartition $\Lambda = A \sqcup B$ and a quantum state $\ket{\psi}$, we consider its Schmidt decomposition:
\begin{align}
\ket{\psi} = \sum_{s=1}^{D_\Lambda} \lambda_s \ket{\psi_{A,s}} \otimes \ket{\psi_{B,s}}.
\end{align}
The R\'enyi entanglement entropy of order $\alpha$ is then defined as
\begin{align}
\label{Def_Renyi_ent}
E_\alpha(\psi) = \frac{1}{1 - \alpha} \log \left( \sum_s \lambda_s^{2\alpha} \right) =  \frac{1}{1 - \alpha} \log\brr{\tr_A \br{\rho^\alpha_A}},
\end{align}
where $\rho_A = \tr_B\left( \ket{\psi}\bra{\psi} \right)$ is the reduced density matrix of $\ket{\psi}$ on subsystem $A$, and $\tr_B$ denotes the partial trace over $B$.
Furthermore, for an arbitrary operator $O$, we define the Schmidt rank ${\rm SR} (O)$ as the minimum integer such that
 \begin{align}
O = \sum_{m=1}^{{\rm SR} (O)} O_{A,m} \otimes O_{B,m},
\end{align}
where $O_{A,m}$ and $O_{B,m}$ are supported on the subsystems $A$ and $B$, respectively.

As a measure of operator norm, we frequently employ the Schatten $p$-norm defined by
\begin{align}
\|O\|_p := \left[ \tr(|O|^p) \right]^{1/p},
\end{align}
where $O$ is an arbitrary operator and $|O| := \sqrt{O^\dagger O}$. In particular, $\|O\|_1$ corresponds to the trace norm, and $\|O\|_\infty$ is the operator norm (i.e., the maximum singular value), which we simply denote by $\|O\|$.

We use $\ket{\psi_t}$ to denote the time-evolved quantum state, defined as
\begin{align}
\ket{\psi_t} := e^{-iHt} \ket{\psi},
\end{align}
while $\ket{\Omega}$ denotes the ground state of the Hamiltonian $H$,
and $\Delta$ denotes the spectral gap between the ground state and the first excited state. Furthermore, we study the quantum Gibbs state at inverse temperature $\beta$:
\begin{align}
\rho_\beta := \frac{1}{Z_\beta} e^{-\beta H}, \quad Z_\beta := \tr\left( e^{-\beta H} \right).
\end{align}

\subsection{One-dimensional systems with power-law decaying interactions}

In several applications, we consider a one-dimensional spatial geometry. Specifically, we use $\Lambda$ to denote a chain of $n$ sites, i.e., $
\Lambda := \{1, 2, \dots, n\}$. We focus on $k$-local Hamiltonians of the form
\begin{align}
H = \sum_{|Z| \leq k} h_Z, \quad \max_{i \in \Lambda} \sum_{Z \ni i} \|h_Z\| \leq g, \label{eq:Hdef}
\end{align}
where $Z \subseteq \Lambda$ denotes the support of the interaction term $h_Z$, and $|Z|$ is its cardinality.

We characterize the decay of interactions using a function $J(r)$ defined by
\begin{align}
\sum_{Z \ni \{i, i'\}} \|h_Z\| \leq J(r), \quad \forall i, i' \in \Lambda, 
\label{def_interaction_decay}
\end{align}
where $r:=|i-i'|$ is a distance between the site $i$ and $i'$. 
In contrast, a system is said to have finite-range interactions if there exists a
positive integer $l_H$ such that
\begin{align}
J(r) = 0 \quad \text{for } r > l_H,
\label{def_short_range}
\end{align}
for some finite integer $l_H > 0$ (see Sec.~\ref{sec:System with short-range interactions}). Conversely, a system exhibits long-range (power-law decaying) interactions if
\begin{align}
\label{def__long_range}
J(r) = J_0 r^{-\eta},\quad \eta>2 ,
\end{align}
where the decay exponent $\eta$ determines the strength of long-range couplings (see Sec.~\ref{Sec:1D_long}).

For a given subset $A \subseteq \Lambda$, the local Hamiltonian $H_A$ includes all terms fully supported within $A$:
\begin{align}
H_A := \sum_{Z \subseteq A} h_Z.
\end{align}
The boundary interaction term between $A$ and its complement $B := \Lambda \setminus A$ is then
\begin{align}
V_{AB} := \sum_{\substack{Z : Z \cap A \neq \emptyset,\\ Z \cap B \neq \emptyset}} h_Z,
\end{align}
such that the full Hamiltonian is again written as $H = H_A + H_B + V_{AB}$.

We aim to approximate states such as $\ket{\psi(t)}$, $\ket{\Omega}$, and $\rho_\beta$ using Matrix Product States (MPSs) or Matrix Product Operators (MPOs), which take the following forms:
\begin{align}
\ket{M_D} &= \sum_{s_1, \dots, s_n} \tr \left( M_1^{[s_1]} M_2^{[s_2]} \cdots M_n^{[s_n]} \right) \ket{s_1, \dots, s_n}, 
\end{align}
where $\{M_j^{[\cdot]}\}_{j=1}^n$ are $D \times D$ matrices, and $D$ is the bond dimension.
We note that any local observables for $\ket{M_D}$ can be computed in at most $\orderof{nD^3}$ computational time.


\subsection{Spectral-entangling strength}

A central problem in this work is efficiently characterizing the capacity of an operator to generate entanglement spectra across subsystems. However, within existing methodologies, there has not been an established quantity that efficiently and directly quantifies the ability of a general operator to generate the entanglement spectrum. To address this gap, we introduce a new quantity, which we refer to as the \textit{Spectral-Entangling (SE) Strength}.

The SE strength introduced here has several desirable features:
\begin{itemize}
    \item Its upper bound can be easily calculated, and in many cases, the quantity itself can be evaluated exactly.
    \item It provides an \textit{optimal upper bound} on the efficiency of generating the entanglement spectrum.
\end{itemize}
In this sense, the SE strength serves as an effective and analytically tractable measure of the spectral-entangling capability of general operators. The formal definition is given below (see also Sec.~\ref{Sec:Conj_alpha} for an extension of this definition)

\begin{definition}[Spectral-Entangling (SE) Strength]
\label{Def:Interaction_strength}
Let $\Phi_{AB}$ be an arbitrary operator acting across subsystems $A$ and $B$. For any product state $\ket{\phi} = \ket{\phi_{AA'}} \otimes \ket{\phi_{BB'}}$ with ancillas $A'$ and $B'$, we let the Schmidt decomposition of $\Phi_{AB} \ket{\phi}$ be
\begin{align}
\Phi_{AB} \ket{\phi} = \sum_s \lambda_s(\phi) \ket{\phi_{AA',s}} \otimes \ket{\phi_{BB',s}}.
\end{align}
Then, the SE strength of $\Phi_{AB}$ is defined as
\begin{align}
\bar{\mathcal{J}}(\Phi_{AB}) := \sup_{\ket{\phi}} \sum_s \lambda_s(\phi),
\label{def_eq_bar_mJ_phi}
\end{align}
where the supremum is taken over all product states $\ket{\phi} = \ket{\phi_{AA'}} \otimes \ket{\phi_{BB'}}$.
\end{definition}

\noindent
{\bf Remark.}
Due to the concave-roof optimization (or supremum over product states), computing $\bar{\mathcal{J}}(\Phi_{AB})$ rigorously is not straightforward.
In the subsequent sections, we will consider several cases in Eqs.~\eqref{V_A_0B_0_definition} and \eqref{H_AB_Ising_int}, where Eqs.~\eqref{bar_J_exact_cal} and \eqref{H_AB_Ising_int_exact_value} give the analytical solutions, respectively.

On the other hand, one can easily calculate an upper bound for $\bar{\mathcal{J}}(\Phi_{AB})$.
For example, if $\Phi_{AB}$ admits a decomposition of the form
\begin{align}
\Phi_{AB} = \sum_j J_j \Phi_{A,j} \otimes \Phi_{B,j}, \quad \text{with } \quad \|\Phi_{A,j} \otimes \Phi_{B,j}\| = 1,
\end{align}
then the SE strength is trivially bounded by
\begin{align}
\bar{\mathcal{J}}(\Phi_{AB}) \leq \sum_j |J_j|.
\label{Trivial_Ineq_interaction_strength}
\end{align}

We emphasize that the SE strength is a structural property of general quantum operators, not limited to Hamiltonians (see Lemma~\ref{lemm:Sum_bar_J} below). It is well-defined for a wide range of operators, including time-evolution unitaries, the approximate-ground-state-projection (AGSP), and general quantum channels. This versatility allows the SE strength to systematically quantify their influence on the entanglement spectrum of quantum states, thereby establishing a unifying framework for entanglement dynamics and structure across diverse physical systems and computational settings.

Finally, we remark on the presence of ancilla systems, which can significantly change the value of the SE strength. 
For example, the two-qubit swap operator $S_{AB}$ maps any product state 
$\ket{\phi_A}\otimes\ket{\phi_B}$ to another product state, hence 
$\bar{\mathcal{J}}(S_{AB})=1$ if no ancilla is allowed. 
However, when additional qubits $A'$ and $B'$ are attached, consider the product state
\begin{align}
\label{Ansilla_existence}
\ket{0_A0_{A'}}\otimes \frac{1}{\sqrt{2}}
  \br{\ket{0_B0_{B'}}+\ket{1_B1_{B'}}}.
\end{align}
After applying $S_{AB}$, the resulting state across the $AA'|BB'$ cut has 
two equal Schmidt coefficients $1/\sqrt{2}$, implying 
$\bar{\mathcal{J}}(S_{AB})\ge \sqrt{2}$. 
Thus, the existence of ancillas reveals a nontrivial entangling capability that would otherwise remain hidden.


\section{Fundamental lemmas}

In this section, we present a collection of elementary but fundamental lemmas that will serve as the basis for our subsequent analysis. While the results themselves are derived via straightforward arguments, they capture essential features underpinning our approach.

First, Lemma~\ref{key_corollary_spectral_SIE} and its Corollary~\ref{corol_key_corollary_spectral_SIE} provide a basic but powerful statement on the decay of the Schmidt coefficients. These results, when combined with the definition of spectral-entangling (SE) strength, allow us to rigorously discuss the truncation error of the Schmidt rank for an arbitrary operator $\Phi_{AB}$ acting on bipartite product states, as detailed in Corollary~\ref{corol:Interaction_strength_S_Coeff}.

Next, Lemma~\ref{lemm:Sum_bar_J} addresses a fundamental property of the SE strength, namely, its subadditivity under certain linear combinations of operators. This property will play a key role in establishing upper bounds for more general operators appearing in our framework.

Finally, Lemma~\ref{lemm:Renyi_Schmidt} elucidates a general relationship between the decay of the Schmidt coefficients and the R\'enyi entanglement. This lemma highlights how the entanglement spectrum constrains the scaling of the largest Schmidt coefficients, thereby linking spectral properties to entanglement measures in a quantitative manner.

Each of these results will play a key role in our analysis, providing the technical foundation for our main theorems.

\subsection{Decay of the Schmidt Coefficients}

We begin by establishing a general upper bound on the sum of overlaps between a non-orthogonal product state expansion and an orthonormal basis, which underlies our analysis of the entanglement spectrum.

\begin{lemma} \label{key_corollary_spectral_SIE}
Let $\ket{\Psi}$ be an arbitrary unnormalized quantum state of the form
\begin{align}
\ket{\Psi} = \sum_{j=1}^{\infty} g_j \ket{A_j} \otimes \ket{B_j},
\label{psi_express_1}
\end{align}
where $\{\ket{A_j}\}_j$ and $\{\ket{B_j}\}_j$ are normalized, but not necessarily orthogonal, states. That is,
\begin{align}
\langle A_{j'} | A_j \rangle \neq 0, \quad \langle B_{j'} | B_j \rangle \neq 0 \for j\neq j' .
\end{align}
Let $\{\ket{a_s}\}_s$ and $\{\ket{b_s}\}_s$ be arbitrary orthonormal bases on subsystems $A$ and $B$, respectively. Then the following inequality holds:
\begin{align}
\sum_s \left| \langle a_s, b_s | \Psi \rangle \right| \le \sum_{j=1}^{\infty} |g_j| =: \mathfrak{g}.
\label{key_corollary_spectral_SIE/main_ineq}
\end{align}
\end{lemma}

\subsubsection{Proof of Lemma~\ref{key_corollary_spectral_SIE}}

The proof follows immediately from the inequality as follows:
\begin{align}
\sum_s \left| \langle a_s, b_s | \phi_A, \phi_B \rangle \right| \le 1,
\label{general_upper_sum_s}
\end{align}
which holds for any product state $\ket{\phi_A} \otimes \ket{\phi_B}$. Indeed, using linearity and the above inequality, we obtain
\begin{align}
\sum_s \left| \langle a_s, b_s | \Psi \rangle \right| 
&= \sum_s \left| \sum_j g_j \langle a_s, b_s | A_j, B_j \rangle \right| \notag \\
&\le \sum_j |g_j| \sum_s \left| \langle a_s, b_s | A_j, B_j \rangle \right| \notag \\
&\le \sum_j |g_j|.
\end{align}

An application of the Cauchy--Schwarz inequality yields inequality~\eqref{general_upper_sum_s}:
\begin{align}
\sum_s \left| \langle a_s | \phi_A \rangle \right| \cdot \left| \langle b_s | \phi_B \rangle \right| 
&\le \left( \sum_s |\langle a_s | \phi_A \rangle|^2 \right)^{1/2}
    \left( \sum_s |\langle b_s | \phi_B \rangle|^2 \right)^{1/2} = 1.
\end{align}
This completes the proof.  $\square$

\vspace{1em}

The lemma implies the following corollary regarding the Schmidt coefficients.

\begin{corol} \label{corol_key_corollary_spectral_SIE}
Let $\ket{\Psi}$ be as defined in Eq.~\eqref{psi_express_1}, and let its Schmidt decomposition be given by
\begin{align}
\ket{\Psi} = \sum_{s=1}^{\mD_\Lambda} \lambda_s \ket{\phi_{A,s}} \otimes \ket{\phi_{B,s}},
\label{psi_express_Schmidt_1}
\end{align}
where the Schmidt coefficients satisfy $\lambda_1 \ge \lambda_2 \ge \cdots \ge 0$. Then we have
\begin{align}
\sum_{s=1}^{\mD_\Lambda} \lambda_s \le \mathfrak{g},
\end{align}
where $\mathfrak{g}$ is as defined in Eq.~\eqref{key_corollary_spectral_SIE/main_ineq}.
\end{corol}

{\bf Remark.}
From the ordering $\lambda_1 \ge \cdots \ge \lambda_{\mD_\Lambda}$, it follows that for any $s_0 \in \mathbb{N}$,
\begin{align}
\label{math_frak_g_upper_bound}
\mathfrak{g} \ge \sum_{s=1}^{s_0} \lambda_s \ge s_0 \lambda_{s_0} \quad \Rightarrow \quad \lambda_{s_0} \le \frac{\mathfrak{g}}{s_0}.
\end{align}
Thus, for all $s \ge 1$,
\begin{align}
\label{lambda_s_bound}
\lambda_s \le \frac{\mathfrak{g}}{s}.
\end{align}

We now use this to evaluate the approximation error due to Schmidt rank truncation.

\begin{corol} \label{corol:Interaction_strength_S_Coeff}
Let $\Phi_{AB}$ be an arbitrary operator acting between subsystems $A$ and $B$. For any product state $\ket{\phi} = \ket{\phi_A} \otimes \ket{\phi_B}$, let $\Phi_{AB} \ket{\phi}$ have Schmidt decomposition
\begin{align}
\Phi_{AB} \ket{\phi} = \sum_s \lambda_s(\phi) \ket{\phi_{A,s}} \otimes \ket{\phi_{B,s}},
\end{align}
with $\lambda_s(\phi)$ in descending order. Define the truncated state
\begin{align}
\ket{\phi_D} := \sum_{s=1}^D \lambda_s(\phi) \ket{\phi_{A,s}} \otimes \ket{\phi_{B,s}}.
\end{align}
Then the approximation error satisfies
\begin{align}
 \label{corol:Interaction_strength_S_Coeff_main}
\left\| \Phi_{AB} \ket{\phi} - \ket{\phi_D} \right\| \le \frac{\bar{\mathcal{J}}(\Phi_{AB})}{\sqrt{D}}.
\end{align}
\end{corol}

\textit{Proof of Corollary~\ref{corol:Interaction_strength_S_Coeff}}
Using Eq.~\eqref{lambda_s_bound} for each $\lambda_s(\phi)$, and applying the definition of the SE strength, we deduce the inequality
\begin{align}
\label{upper_bound_error_approx_D}
\left\| \Phi_{AB} \ket{\phi} - \ket{\phi_D} \right\|^2 
&= \sum_{s>D} \lambda_s(\phi)^2 \le \sum_{s>D} \left( \frac{\bar{\mathcal{J}}(\Phi_{AB})}{s} \right)^2 
\le \bar{\mathcal{J}}(\Phi_{AB})^2 \sum_{s>D} \frac{1}{s^2} \le \frac{\bar{\mathcal{J}}(\Phi_{AB})^2}{D}.
\end{align}
Taking the square root on both sides yields the desired bound. $\square$

\subsection{Upper bound on the SE strength for general operators}

We consider an arbitrary operator of the form
\begin{align}
\label{Def:O_in_Phi(x)}
O = \int_{-\infty}^\infty f(x)\Phi_{AB}(x) dx ,
\end{align} 
where $\Phi(x)$ is an arbitrary operator that depends on $x \in \mathbb{R}$.  
We then analyze how the entanglement spectrum varies with the choice of the operator $O$.
We prove the following proposition:
\begin{lemma} [Sub-additivity of the SE strength]\label{lemm:Sum_bar_J}
For any operator of the form given by Eq.~\eqref{Def:O_in_Phi(x)}, the inequality
\begin{align} 
\bar{\mJ}(O) \le  \int_{-\infty}^\infty |f(x)| \bar{\mJ} [\Phi_{AB}(x)] dx
\end{align}
is satisfied.
\end{lemma}

\textit{Proof of Lemma~\ref{lemm:Sum_bar_J}.}
From Definition~\ref{Def:Interaction_strength}, we define $\ket{\tilde{\phi}}$ such that 
\begin{align} 
\ket{\tilde{\phi}}= \arg\sup_{\ket{\phi}} \br{\sum_s \tilde{\lambda}_s(\phi)} .
\end{align}
For the quantum state $\ket{\tilde{\phi}}$, we consider the Schmidt decomposition of  
\begin{align} 
\label{Schmidt_decomp_O}
O \ket{\tilde{\phi}}= \sum_{s} \tilde{\lambda}_s(\phi)\ket{\tilde{\phi}_{AA',s}} \otimes \ket{\tilde{\phi}_{BB',s}} . 
\end{align}
Then, our task is to estimate $\sum_s \tilde{\lambda}_s(\tilde{\phi})=\bar{\mJ}(O)$. 

For an arbitrary $\Phi_{AB}(x)$ and the state $\ket{\tilde{\phi}}$, we define the Schmidt decomposition as follows:
\begin{align} 
\Phi_{AB}(x) \ket{\tilde{\phi}}= \sum_{s} \tilde{\lambda}_s(\tilde{\phi},x)\ket{\tilde{\phi}_{AA',s}(x)} \otimes \ket{\tilde{\phi}_{BB',s}(x)}  ,
\end{align}
which also yields 
\begin{align} 
O \ket{\tilde{\phi}}= \int_{-\infty}^\infty dx f(x)  \sum_{s} \tilde{\lambda}_s(\tilde{\phi},x)\ket{\tilde{\phi}_{AA',s}(x)} \otimes \ket{\tilde{\phi}_{BB',s}(x)}  ,
\end{align}
where we use Eq.~\eqref{Def:O_in_Phi(x)}.
From the Schmidt decomposition~\eqref{Schmidt_decomp_O}, we obtain an upper bound on $\bar{\mJ}(O)$ as
\begin{align} 
\label{uppp_bar_mJ_O}
\bar{\mJ}(O)
&=  \sum_{s_0}\bra{\tilde{\phi}_{AA',s_0}, \tilde{\phi}_{BB',s_0} }  O \ket{\tilde{\phi}}  \notag \\
&= \sum_{s_0}\int_{-\infty}^\infty dx f(x)  \sum_{s}  \tilde{\lambda}_s(\tilde{\phi},x) \bra{\tilde{\phi}_{AA',s_0}, \tilde{\phi}_{BB',s_0} } \tilde{\phi}_{AA',s}(x),\tilde{\phi}_{BB',s}(x)\rangle    \notag \\
&\le \int_{-\infty}^\infty dx |f(x)|  \sum_{s}  \tilde{\lambda}_s(\tilde{\phi},x) \sum_{s_0} \abs{ \bra{\tilde{\phi}_{AA',s_0}, \tilde{\phi}_{BB',s_0} } \tilde{\phi}_{AA',s}(x),\tilde{\phi}_{BB',s}(x)\rangle }.
\end{align}

By using the inequality~\eqref{key_corollary_spectral_SIE/main_ineq} in Lemma~\ref{key_corollary_spectral_SIE} with $g_1=1$ and $g_j=0$ ($j\ge 2$), we have 
\begin{align} 
\label{s_0sum_s_x}
\sum_{s_0} \abs{ \bra{\tilde{\phi}_{AA',s_0}, \tilde{\phi}_{BB',s_0} } \tilde{\phi}_{AA',s}(x),\tilde{\phi}_{BB',s}(x)\rangle } \le 1 
\end{align}
for an arbitrary $x$ and $s$. 
Furthermore, applying the equation~\eqref{def_eq_bar_mJ_phi} to $\bar{\mJ} [\Phi_{AB}(x)]$ implies 
\begin{align} 
\label{upp_mJ_Phi_AB}
\bar{\mJ} [\Phi_{AB}(x)]  \ge \sum_{s} \tilde{\lambda}_s(\tilde{\phi},x)  . 
\end{align}
By combining the inequalities~\eqref{s_0sum_s_x} and \eqref{upp_mJ_Phi_AB} with~\eqref{uppp_bar_mJ_O}, we arrive at the desired inequality of 
\begin{align} 
\label{uppp_bar_mJ_O_fin}
\bar{\mJ}(O)
&\le \int_{-\infty}^\infty dx |f(x)|  \sum_{s}  \tilde{\lambda}_s(\tilde{\phi},x)  \le  \int_{-\infty}^\infty dx |f(x)| \bar{\mJ} [\Phi_{AB}(x)] .
\end{align}
This completes the proof. $\square$

\subsection{Schmidt coefficients vs. R\'enyi entanglement}

\begin{lemma} \label{lemm:Renyi_Schmidt}
Let $\ket{\psi}$ be an arbitrary quantum state with the Schmidt decomposition as follows:
\begin{align}
\ket{\psi} = \sum_{s=1}^{D_\Lambda} \lambda_s \ket{\psi_{A,s}} \otimes \ket{\psi_{B,s}},
\end{align}
where the descending order ($\lambda_1\ge \lambda_2\ge \cdots$) is assumed.  
Then, for any natural number $s_0$ and $\alpha \in (0,1)$, $\lambda_{s_0}$ satisfies the following inequality:
\begin{align}
\label{lemm:Renyi_Schmidt_main_ineq}
| \lambda_{s_0}| \le \br{ \frac{e^{(1-\alpha)E_\alpha(\psi)}}{s_0}}^{1/(2\alpha)} ,
\end{align}
where $E_\alpha(\psi)$ is the R\'enyi entanglement~\eqref{Def_Renyi_ent}. 
\end{lemma}

\textit{Proof of Lemma~\ref{lemm:Renyi_Schmidt}.}
From the definition~\eqref{Def_Renyi_ent}, we obtain
\begin{align}
\label{proof_lemm:Renyi_Schmidt_1}
 \sum_{s=1}^{D_\Lambda} \lambda^{2\alpha}_s = e^{(1-\alpha) E_\alpha(\psi)}  .
\end{align}
Then, we have 
\begin{align}
 \sum_{s=1}^{D_\Lambda} \lambda^{2\alpha}_s \ge  \sum_{s=1}^{s_0} \lambda^{2\alpha}_s  \ge s_0  \lambda^{2\alpha}_{s_0} , 
\end{align}
where we use $\lambda_s \le \lambda_{s_0}$ for $s\le s_0$. 
Therefore, by combining the above two relations, we get 
\begin{align}
s_0  \lambda^{2\alpha}_{s_0} \le e^{(1-\alpha) E_\alpha(\psi)}  , 
\end{align}
which yields the main inequality~\eqref{lemm:Renyi_Schmidt_main_ineq}.
This completes the proof. $\square$

\section{Spectral Small-Incremental-Entangling (SIE)}
\subsection{Entanglement rate for R\'enyi entanglement ($\alpha \ge 1/2$)}

\begin{theorem} \label{thm:Renyi_SIE}
For an arbitrary time-evolved quantum state $\ket{\psi_t}= e^{-iHt} \ket{\psi}$ and the $\alpha$-R\'enyi entanglement $E_\alpha(\psi_t)$ in Eq.~\eqref{Def_Renyi_ent}, i.e., 
\begin{align}
\label{Def_Renyi_ent_re}
E_\alpha(\psi_t) =\frac{1}{1-\alpha} \log\brr{ \tr_A \br{\rho_{t,A}^\alpha}}  , \quad \rho_{t,A}= \tr_B\br{ \ket{\psi_t}\bra{\psi_t}} ,
\end{align}
we obtain the upper bound on the entanglement rate as 
\begin{align}
\label{lemm:Renyi_SIE_main_ineq}
\abs{\frac{d E_\alpha(\psi_t)}{dt}}\le  c_{\alpha}  \bar{\mJ}(V_{AB})  
\end{align}
for $\alpha\ge 1/2$, where $c_{\alpha} $ is defined as follows (see also Fig.~\ref{fig:c_alpha}):  
\begin{align}
\label{Def_c_alpgha}
c_{\alpha}:= \frac{2\alpha}{1-\alpha}\brr{ \br{2\alpha-1}^{(2\alpha-1)/(2-2\alpha)} - \br{2\alpha-1}^{1/(2-2\alpha)}} .
\end{align}
In particular, from $c_{1/2}=2$, the inequality~\eqref{lemm:Renyi_SIE_main_ineq} for $\alpha=1/2$ can be rewritten as
\begin{align}
\label{lemm:Renyi_SIE_main_ineq_1/2}
\abs{\frac{d E_{\alpha=1/2} (\psi_t)}{dt}}\le  2 \bar{\mJ}(V_{AB})  .
\end{align}
The statement remains valid even when the system is extended by attaching arbitrary ancillas to $A$ and $B$.
\end{theorem}

\begin{figure}[ttt]
  \centering
  \includegraphics[width=0.45\textwidth]{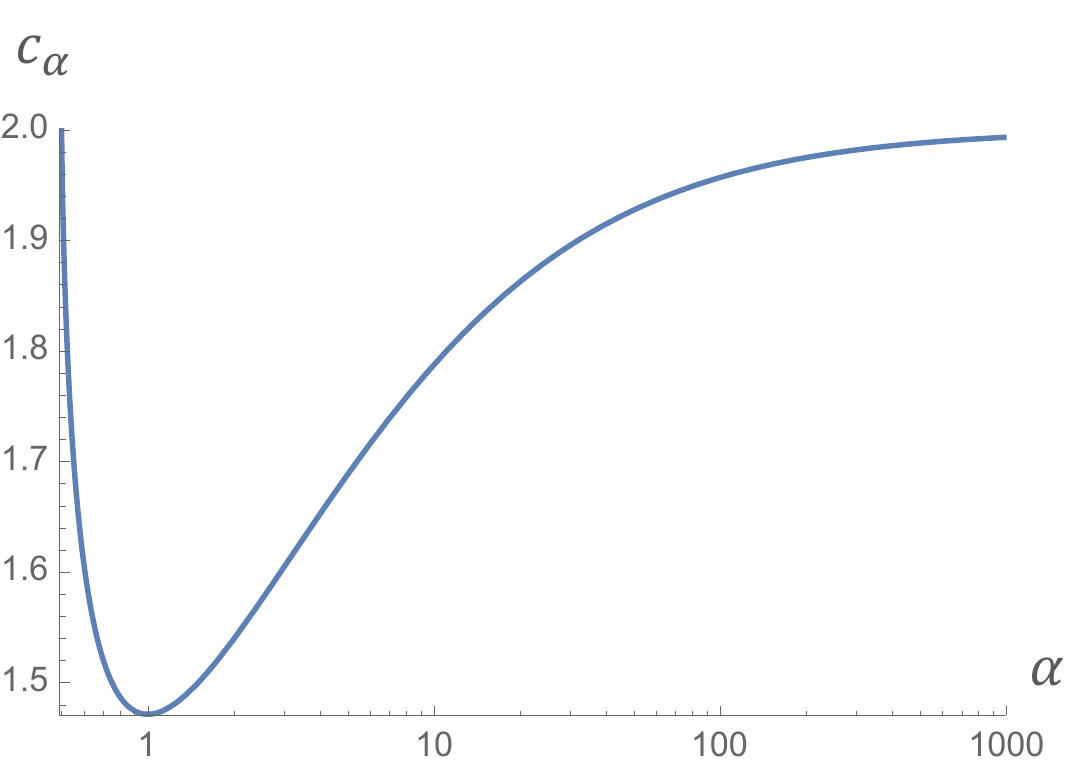}
  \caption{
    Plot of $c_\alpha$ with respect to $\alpha$. We have $c_{\alpha=1/2}=2$, $c_{\alpha=1}=4/e$, and $c_{\alpha=\infty}=2$. 
  }
  \label{fig:c_alpha}
\end{figure}

{\bf Remark.}
As a specific example, let us consider the case where $V_{AB}$ is given by
\begin{align}
V_{AB} = \sum_{j} J_{j,AB} h_{j,A} \otimes h_{j,B}   , \quad \norm{h_{j,A}}= \norm{h_{j,B}}=1 .
\end{align}
Then, the upper bound~\eqref{lemm:Renyi_SIE_main_ineq} can be simplified as 
\begin{align}
\abs{\frac{d E_\alpha(\psi_t)}{dt}}\le c_{\alpha} \sum_{j} |J_{j,AB}|   , 
\end{align}
where we use the inequality~\eqref{Trivial_Ineq_interaction_strength} to obtain $ \bar{\mJ}(V_{AB}) \le \sum_j |J_{j,AB}|$.

\subsection{Optimality of the entanglement generation}\label{Sec:optimatlity_spec:SIE}

Before showing the proof, we emphasize the optimality of the entanglement rate bound for R\'enyi entanglement. The main points are listed below:
\begin{itemize}
    \item \textbf{Tightness at $\alpha=1/2$:} In the special case $\alpha=1/2$, the coefficient $c_{1/2}=2$ is tight. There exists an explicit quantum dynamics for which the entanglement generation rate saturates the upper bound over a finite-time interval, demonstrating that the result cannot be improved in general (see Sec.~\ref{sec_Optimality of Theorem_thm_Renyi_SIE}).
    \item \textbf{Threshold of $\alpha=1/2$:} The value $\alpha=1/2$ serves as a sharp threshold for the existence of a meaningful upper bound on the generation rate of R\'enyi entanglement. For $\alpha < 1/2$, no universal bound exists; that is, there are examples where the entanglement generation in an $\orderof{1}$ time can become arbitrarily large depending on the Hilbert space dimension. Conversely, for $\alpha \ge 1/2$, the bound in Theorem~\ref{thm:Renyi_SIE} provides a dimension-independent constraint (see Sec.~\ref{Optimality of the threshold: Unbounded entanglement}).
\end{itemize}
We also make two remarks on optimality:
\begin{itemize}
    \item For $\alpha > 1/2$, the optimality of the entanglement rate bound remains an open problem. In particular, for the case $\alpha=1$ (the von Neumann entropy), the qualitative behavior of the entanglement rate differs from the case $\alpha=1/2$ (see below).
    \item Even if we consider systems with spatial structure, such as those with short-range interactions, the unbounded entanglement rate below $\alpha=1/2$ can appear instantaneously. However, in such cases, qualitatively different behavior can emerge for the entanglement generation over a finite time. We will discuss this point in detail in Sec.~\ref{sec:System with short-range interactions}. 
\end{itemize}
Regarding the first additional point, the standard SIE theorem~\cite{PhysRevA.76.052319,PhysRevLett.111.170501,10.1063/1.4901039,Marien2016} for $\alpha=1$ depends on $\norm{V_{AB}}$, 
while our generalized bound depends on the SE strength $\bar{\mJ}(V_{AB})$.
Here, we can see a qualitative difference between $\bar{\mJ}(V_{AB})$ and $\norm{V_{AB}}$.

%
%
As a concrete example, let us consider the case where $V_{AB}$ is supported on $A_0B_0$ ($A_0\subset A$, $B_0\subset B$) and $\mathcal{D}_{A_0}=\mathcal{D}_{B_0}=M+1$. 
We then consider an interaction of the form
\begin{align}
V_{A_0B_0} = \sum_{j=1}^{M} J\br{ \ket{j_{A_0},j_{B_0}} \bra{0_{A_0},0_{B_0}} + {\rm h.c.} },
\end{align}
where $\{\ket{j_{A_0}}\}_{j=0}^M$ and $\{\ket{j_{B_0}}\}_{j=0}^M$ are arbitrary operator bases on $A_0$ and $B_0$, respectively.
For this specific example,
we obtain precisely [see also Eq.~\eqref{bar_J_exact_cal}]
\begin{align}
\bar{\mJ}(V_{A_0B_0}) =  J M ,\quad \norm{V_{A_0B_0}}=J \sqrt{M} .
\end{align}
From the above estimation, we get the entanglement rate for the $\alpha=1$ case~\cite{PhysRevLett.111.170501} :
\begin{align}
\abs{\frac{d E_{\alpha=1} (\psi_t)}{dt}}\le 18 \norm{V_{A_0B_0}} \log(\mathcal{D}_{A_0})=18 J \sqrt{M} \log (M+1) ,
\end{align}
whereas we can explicitly find a dynamics such that  
\begin{align}
\abs{\frac{d E_{\alpha=1/2} (\psi_t)}{dt}}\ge 2\bar{\mJ}(V_{AB}) - \varepsilon =  2J M - \varepsilon  , 
\end{align}
where $\varepsilon$ approaches zero in the limit as the Hilbert space dimension becomes large (see Proposition~\ref{prop:optimality_1}). 
Therefore, the entanglement rates for $\alpha=1$ and $\alpha=1/2$ show a qualitative difference in the limit of $M\to \infty$.

It is an open problem to unify the current spectral SIE and the standard SIE theorems.

\subsection{Proof of Theorem~\ref{thm:Renyi_SIE}}

Without loss of generality, we may set $H_A=H_B=0$. 
First, we consider 
\begin{align}
\frac{d}{dt} E_\alpha(\psi_t) =\frac{1}{1-\alpha}\frac{d}{dt} \log\brr{ \tr_A \br{\rho_{t,A}^\alpha}}  
= \frac{\alpha}{(1-\alpha) \tr \br{\rho_{t,A}^\alpha}} \tr_{AB} \br{\rho_{t,A}^{\alpha-1}  [-iH, \rho_{t}]},
\end{align}
where $\rho_{t}=\ket{\psi_t}\bra{\psi_t}$. 
Precisely speaking, the trace should be restricted to the support of $\rho_A$, i.e., the subspace corresponding to strictly positive eigenvalues. In particular, contributions from eigenstates with zero eigenvalues are excluded so that the operator $\rho_A^{\alpha-1}$ is well-defined.
Then, we obtain 
\begin{align}
\label{d/dt_E_alpha_psi_t}
\abs{\frac{d}{dt} E_\alpha(\psi_t) } \le \frac{\alpha}{\abs{1-\alpha}\tr \br{\rho_{t,A}^\alpha}} \abs{ \tr_{AB} \br{\rho_{t,A}^{\alpha-1}  [H, \rho_{t}]}} .
\end{align}
In the following, for an arbitrary quantum state $\rho=\ket{\phi}\bra{\phi}$, we generally consider the upper bound of 
$\abs{ \tr_{AB} \br{\rho_{A}^{\alpha-1}  [H, \rho]} } $, where $\rho_A=\tr_B\br{\rho}$.

For this purpose, we use the Schmidt decomposition of 
\begin{align}
\label{Schmidt_decompo_psi_ket}
\ket{\phi} = \sum_{s} \lambda_s \ket{\phi_{s,A}} \otimes \ket{\phi_{s,B}}  ,\quad \rho_{A}=  \sum_{s} \lambda_s^2  \ket{\phi_{s,A}}\bra{\phi_{s,A}}. 
\end{align}
From the expression, we obtain
\begin{align}
\tr_{AB} \br{\rho_{A}^{\alpha-1}  [H, \rho]} 
&= \sum_{s,s''}  \lambda_s^{2\alpha-2} \bra{\phi_{s,A}, \phi_{s'',B}} [V_{AB} , \rho] \ket{\phi_{s,A}, \phi_{s'',B}} \notag \\
&=  \sum_{s,s''}  \lambda_s^{2\alpha-2}  \br{\bra{\phi_{s,A}, \phi_{s'',B}} V_{AB}  \ket{\psi} \langle \psi \ket{\phi_{s,A}, \phi_{s'',B}}- {\rm c.c.} } \notag \\
&=  \sum_{s,s'}  \lambda_s^{2\alpha-2}  \br{\bra{\phi_{s,A}, \phi_{s,B}} V_{AB}  \lambda_s \lambda_{s'}  \ket{\phi_{s',A}, \phi_{s',B}}- {\rm c.c.} }  \notag \\
&= \sum_{s,s'}  \lambda_s^{2\alpha-1} \lambda_{s'}   \br{\bra{\phi_{s,A}, \phi_{s,B}} V_{AB}   \ket{\phi_{s',A}, \phi_{s',B}}- {\rm c.c.} }  \notag \\
&= \sum_{s,s'}  \br{ \lambda_s^{2\alpha-1} \lambda_{s'} -\lambda_{s'}^{2\alpha-1} \lambda_{s}  }\bra{\phi_{s,A}, \phi_{s,B}} V_{AB}   \ket{\phi_{s',A}, \phi_{s',B}} ,
\label{tr_AB_rho_A_alpha-1_H,rho}
\end{align} 
where we use Eq.~\eqref{Schmidt_decompo_psi_ket} in the step from the second to the third line.

Next, we consider the relations
\begin{align}
\abs{\lambda_s^{2\alpha-1} \lambda_{s'} -\lambda_{s'}^{2\alpha-1} \lambda_{s}} = \lambda_s^{2\alpha} \abs {x - x^{2\alpha-1}}  \le \tilde{c}_{\alpha} \lambda_s^{2\alpha}
\end{align} 
for $\lambda_s > \lambda_{s'}$ (or $s< s'$), where $x=\lambda_{s'}/\lambda_s \in [0,1]$ and the constant $\tilde{c}_{\alpha}$ is defined by
\begin{align}
\tilde{c}_{\alpha} :=\abs{ \br{2\alpha-1}^{1/(2-2\alpha)} - \br{2\alpha-1}^{(2\alpha-1)/(2-2\alpha)} },
\end{align} 
where we set $\tilde{c}_{1/2}=1$. 
Note that $\abs {x - x^{2\alpha-1}}$ attains its maximum when $x=(2\alpha-1)^{1/(2-2\alpha)}$.

Introducing $\zeta_{s,s'}$ as 
 \begin{align}
\zeta_{s,s'} := \abs{ \bra{\phi_{s,A}, \phi_{s,B}} V_{AB}   \ket{\phi_{s',A}, \phi_{s',B}}  } ,
\end{align} 
the term on the extreme left-hand side of Eq.~\eqref{tr_AB_rho_A_alpha-1_H,rho} admits an upper bound as follows:
\begin{align}
\abs{\tr_{AB} \br{\rho_{A}^{\alpha-1}  [V_{AB} , \rho]} }
&\le \tilde{c}_{\alpha} \br{\sum_{s<s'}   \lambda_s^{2\alpha} \zeta_{s,s'} + \sum_{s>s'}   \lambda_{s'}^{2\alpha} \zeta_{s,s'}}  \notag \\
&=2 \tilde{c}_{\alpha}\sum_{s<s'}  \lambda_s^{2\alpha} \zeta_{s,s'}  .
\label{tr_AB_rho_A_alpha-1_H,rho_2}
\end{align} 

From Definition~\ref{Def:Interaction_strength} for the SE strength, $ V_{AB}   \ket{\phi_{s,A}, \phi_{s,B}}$ can be rewritten as
 \begin{align}
 V_{AB}   \ket{\phi_{s,A}, \phi_{s,B}}=  \sum_{s'} \lambda^{(s)}_{s'} \ket{\tilde{\phi}_{s',A}, \tilde{\phi}_{s',B}} , \quad   \sum_{s'} \lambda^{(s)}_{s'}\le \bar{\mJ}(V_{AB}) . 
\end{align} 
Therefore, applying the inequality~\eqref{key_corollary_spectral_SIE/main_ineq} in Lemma~\ref{key_corollary_spectral_SIE}, we obtain
 \begin{align}
\sum_{s' : s' >s} \zeta_{s,s'} \le  \sum_{s'} \abs{ \bra{\phi_{s',A}, \phi_{s',B}} V_{AB}   \ket{\phi_{s,A}, \phi_{s,B}}} \le  \bar{\mJ}(V_{AB}) 
\end{align} 
for any $s$. 
The above upper bound reduces the inequality~\eqref{tr_AB_rho_A_alpha-1_H,rho_2} to 
\begin{align}
\abs{\tr_{AB} \br{\rho_{A}^{\alpha-1}  [V_{AB} , \rho]} }
&\le 2\tilde{c}_\alpha \bar{\mJ}(V_{AB})  \sum_{s} \lambda_s^{2\alpha}  =   2\tilde{c}_\alpha\bar{\mJ}(V_{AB}) \tr_A \br{\rho_A^\alpha} .
\end{align} 
By applying the above inequality to~\eqref{d/dt_E_alpha_psi_t}, we obtain
\begin{align}
\label{d/dt_E_alpha_psi_t_upp}
\abs{\frac{d}{dt} E_\alpha(\psi_t) } \le \frac{2 \tilde{c}_{\alpha} \alpha}{\abs{1-\alpha}} \bar{\mJ}(V_{AB}) =c_{\alpha}  \bar{\mJ}(V_{AB})  .
\end{align}
This completes the proof. $\square$

{~}

\hrulefill{\bf [ End of Proof of Theorem~\ref{thm:Renyi_SIE}]}

{~}

From Theorem~\ref{thm:Renyi_SIE}, we can deduce the following corollary, which will be useful in subsequent discussions:
\begin{corol}\label{corol:Renyi_MJ}
For a Hamiltonian $H$ of the form given in Eq.~\eqref{Ham_H_V_AB}, the SE strength $\bar{\mJ}(e^{-iHt})$ is upper-bounded by
\begin{align}
\label{corol:Renyi_MJ_main_ineq}
\bar{\mJ}(e^{-iHt}) \le e^{\bar{\mJ} (V_{AB}) t} .
\end{align}
We note that the upper bound is generalized to arbitrary time-dependent Hamiltonians. 
\end{corol}

\textit{Proof of Corollary~\ref{corol:Renyi_MJ}.}
Let us consider an arbitrary product state $\ket{\phi}=\ket{\phi_A} \otimes \ket{\phi_B}$.
Then, the (1/2)-R\'enyi entanglement of the time-evolved state $\ket{\phi_t} = e^{-iHt} \ket{\phi}$ satisfies the upper bound
\begin{align}
E_{1/2}(\phi_t) = \int_0^t \abs{\frac{d}{dt_1} E_{1/2} (\phi_{t_1}) } dt_1 
\le 2\bar{\mJ}(V_{AB}) t ,
\end{align}
where we use $c_{1/2}=2$, i.e., $|dE_{1/2} (\phi_{t})/dt | \le 2 \bar{\mJ}(V_{AB})$. 
We then define the Schmidt decomposition of $\ket{\phi_t}$ as follows:
\begin{align}
\ket{\phi_t} = \sum_s \lambda_s(t)  \ket{\phi_{A,s}^{(t)}} \otimes \ket{\phi_{B,s}^{(t)}} ,
\end{align}
from which we obtain $E_{1/2}(\phi_t)$ as 
\begin{align}
E_{1/2}(\phi_t) = 2 \log \br{\sum_{s}  \lambda_s(t)} .
\end{align}
Therefore, we obtain
\begin{align}
\sum_{s}  \lambda_s(t) = e^{E_{1/2}(\phi_t)/2} \le e^{\bar{\mJ}(V_{AB}) t} . 
\end{align}
Since this holds for any initial state $\ket{\phi}$, the main inequality~\eqref{corol:Renyi_MJ_main_ineq} follows from the definition~\eqref{def_eq_bar_mJ_phi}, i.e., $\bar{\mJ}(\Phi_{AB}) := \sup_{\phi} \sum_{s} \lambda_s (\phi)$. 
This completes the proof. $\square$

\subsection{Tightness at $\alpha=1/2$} \label{sec_Optimality of Theorem_thm_Renyi_SIE}

We show here that the coefficient $2$ in Theorem~\ref{thm:Renyi_SIE} at $\alpha=1/2$ is tight.  
We prove the following proposition:
\begin{prop} \label{prop:optimality_1}
Let  $A$ and $B$ each consist of $n+1$ qudits of dimension $M+1$.
Then, in the limit of $n\to \infty$, there exists quantum dynamics such that the average entanglement rate saturates the spectral SIE bound~\eqref{lemm:Renyi_SIE_main_ineq_1/2} over a finite-time interval: 
\begin{align}
\label{lower_bound_E_rate_1/2_result1}
\abs{\frac{E_{1/2}(\psi_t) - E_{1/2}(\psi_0) }{t} } \ge 2 \bar{\mJ} (V_{A_0B_0}) - \frac{2[\bar{\mJ} (V_{A_0B_0})]^2 t}{n} 
\end{align}
for 
\begin{align}
\label{lower_bound_E_rate_1/2_conditi}
t\le \frac{n \sqrt{M}}{2\bar{\mJ} (V_{A_0B_0})} .
\end{align}
In the inequality~\eqref{lower_bound_E_rate_1/2_result1},  only the first term is dominant for $n\gg [ \bar{\mJ} (V_{A_0B_0})]^2t$, or $t \ll n/ [ \bar{\mJ} (V_{A_0B_0})]^2$.  
\end{prop}

\subsubsection{Proof of Proposition~\ref{prop:optimality_1}.}

From the given assumptions, the subsystems $A$ and $B$ are composed of $n+1$ qudits with a dimension $M+1$, respectively:
 \begin{align}
A= A_0 \sqcup A_1 \sqcup \cdots \sqcup A_{n} ,\quad  B= B_0 \sqcup B_1 \sqcup \cdots \sqcup B_{n}  
\end{align} 
with 
 \begin{align}
\mD_{A_0}= \mD_{A_1}=\cdots = \mD_{A_{n}}=M+1 ,  \quad  \mD_{B_0}= \mD_{B_1}=\cdots = \mD_{B_{n}}=M+1 .
\end{align} 

Here, we consider the interaction $V_{A_0B_0}$ of the form
\begin{align}
\label{V_A_0B_0_definition}
V_{A_0B_0} = \sum_{j=1}^{M}J \br{ \ket{j_{A_0},j_{B_0}} \bra{0_{A_0},0_{B_0}} + {\rm h.c.} }.
\end{align}
Since this operator is diagonalizable, we obtain
\begin{align}
\label{V_A_0B_0_acting}
e^{- i V_{A_0B_0} t} \ket{0_{A_0},0_{B_0}} = \cos (\sqrt{M} J t) \ket{0_{A_0},0_{B_0}} -i  \sum_{j=1}^{M} \frac{\sin (\sqrt{M} J t) }{\sqrt{M}} \ket{j_{A_0},j_{B_0}}
\end{align}
for any $t$.
Also, we consider the Hamiltonian $H_{A,s}$ and $H_{B,s}$ that realizes a swap operation between $A_0$ (resp. $B_0$) and $A_s$ (resp. $B_s$).
Because the norm of $H_{A,s}$ and $H_{B,s}$ can be arbitrarily large, we can let 
\begin{align}
\label{H_A_sB_s_acting}
e^{-iH_{A,s} \varepsilon} = \sum_{j,j'=0}^M \ket{j'_{A_0},j_{A_s}} \bra{j_{A_0},j'_{A_s}} ,\quad 
e^{-iH_{B,s} \varepsilon} = \sum_{j,j'=0}^M \ket{j'_{B_0},j_{B_s}} \bra{j_{B_0},j'_{B_s}} 
\end{align}
for an infinitesimally small $\varepsilon$.
We recall that there are no constraints on the form of $H_A$ and $H_B$ in the Hamiltonian~\eqref{Ham_H_V_AB}.

Then, we construct the quantum dynamics using the time-dependent Hamiltonian as follows:
 \begin{align}
U_{0\to t'} := \prod_{s=1}^n e^{-i (H_{A,s}+H_{B,s}) \varepsilon} e^{- i V_{A_0B_0} t/n} ,
\end{align}
where $t'= t+\mD_0 \varepsilon$, which can be made arbitrarily close to $t$ by taking the limit $\varepsilon\to +0$. 
We choose the initial state $\ket{\psi_0}$ as 
 \begin{align}
\ket{\psi_0} = \bigotimes_{s=0}^{n} \ket{0_{A_s},0_{B_s}} . 
\end{align}
By applying the unitary time evolution $U_{0\to t'} $, we have 
 \begin{align}
&\ket{\psi_t}= U_{0\to t'} \ket{\psi_0}= \ket{0_{A_0},0_{B_0}} \otimes \prod_{s=1}^n \ket{\psi_{A_s,B_s}(t/n)}, \notag \\
&\ket{\psi_{A_s,B_s}(x)}=   \cos (\sqrt{M} J x) \ket{0_{A_0},0_{B_0}} -i  \sum_{j=1}^{M} \frac{\sin (\sqrt{M} J x) }{\sqrt{M}} \ket{j_{A_0},j_{B_0}}  ,
\end{align}
where we use Eqs.~\eqref{V_A_0B_0_acting} and \eqref{H_A_sB_s_acting}. 
Each of the quantum states $\{ \ket{\psi_{A_s,B_s}(t/n)}\}_{s=1}^n$ is given in the form of the Schmidt decomposition. 

Then, in the limit of $\varepsilon\to +0$, the $(1/2)$-R\'enyi entanglement entropy for $\ket{\psi_t}$ is given by 
\begin{align}
E_{1/2}(\psi_t) = \sum_{s=1}^n E_{1/2}(\psi_{A_s,B_s}(x)) 
&=2 n \log \br{\cos\br{\sqrt{M} J t/n}  + M \frac{\abs{\sin \br{\sqrt{M} J t/n}} }{\sqrt{M}}  } \notag \\
&=2 n \log \br{ \cos(z)  + \sqrt{M} \sin(z)   } ,
\end{align}
where we set $z=\sqrt{M} J t/n$ which satisfies $0\le z \le 1/2$ from the condition \eqref{lower_bound_E_rate_1/2_conditi}.
Note that the conditions $\cos(z),\sin(z) \ge 0$ for $0\le z \le 1/2$ were used in the above equation. 
By differentiating with respect to $z$, we obtain
\begin{align}
\frac{1}{dz} \log \br{ \cos(z)  + \sqrt{M} \sin(z)   } = \frac{\sqrt{M} \cos (z) - \sin(z)}{\cos (z) +\sqrt{M} \sin(z)} 
\ge 0 ,
\end{align}
where we use $M\ge 1$ and $z\le1/2$.
Hence, by using 
\begin{align}
\frac{\sqrt{M} \cos (z) - \sin(z)}{\cos (z) +\sqrt{M} \sin(z)} \ge  \sqrt{M} - (M+1) z  \for 0\le z \le 1/2, \quad M\ge1,  
\end{align}
we obtain 
\begin{align}
\label{lower_bound_E_rate_1/2_pre}
E_{1/2}(\psi_t) = 2 n \log \br{ \cos(z)  + \sqrt{M} \sin(z)   } 
&\ge2n\br{ \sqrt{M} z - \frac{M+1}{2} z^2} \notag \\
&= 2M Jt - \frac{2M^2 J^2 t^2}{n} .
\end{align}

Finally, we prove the equation
\begin{align}
\label{bar_J_exact_cal}
\bar{\mJ} (V_{A_0B_0} ) =MJ ,
\end{align}
which reduces the inequality~\eqref{lower_bound_E_rate_1/2_pre} to the desired form~\eqref{lower_bound_E_rate_1/2_result1}. 
For this purpose, we generally consider a product state in the form of 
\begin{align}
\ket{\phi} = \br{\sum_{j=0}^{M} a_j \ket{j_{A_0}} \otimes \ket{\phi_{j,A_{1:n}}} }   \otimes  \br{\sum_{j=0}^{M} b_j \ket{j_{B_0}} \otimes \ket{\phi_{j,B_{1:n}}} },
\end{align}
where we introduced the abbreviations $A_{1:n}:= \bigcup_{s=1}^n A_s$ and $B_{1:n}:= \bigcup_{s=1}^n B_s$, and $\{\ket{\phi_{j,A_{1:n}}}\}_{j=0}^{M}$, $\{\ket{\phi_{j,B_{1:n}}}\}_{j=0}^M$ are arbitrary quantum states supported on $A_{1:n}$ and $B_{1:n}$, respectively. 

Then, by applying $V_{A_0B_0}$ of Eq.~\eqref{V_A_0B_0_definition}, we have 
\begin{align}
V_{A_0B_0} \ket{\phi} = J \ket{0_{A_0},0_{B_0}}\sum_{j=1}^{M} a_j b_j  \ket{\phi_{j,A_{1:n}}, \phi_{j,B_{1:n}}} 
+  J \sum_{j=1}^{M} a_0 b_0\ket{j_{A_0},j_{B_0}}  \ket{\phi_{0,A_{1:n}}, \phi_{0,B_{1:n}}} ,
\end{align}
which is given by the form of the Schmidt decomposition. 
Applying the Cauchy--Schwarz inequality yields
\begin{align}
\label{bar_J_exact_cal_prove}
J \sum_{j=1}^{M} |a_j b_j| +  J \sum_{j=1}^{M} |a_0 b_0|
& \le 
J \br{ \sum_{j=1}^{M} |a_j|^2 + \sum_{j=1}^{M} |a_0|^2}^{1/2} \br{ \sum_{j=1}^{M}  |b_j|^2+ \sum_{j=1}^{M} |b_0|^2}^{1/2} \notag \\
&= J \brr{ 1+(M-1) |a_0|^2}^{1/2} \brr{ 1+(M-1) |b_0|^2}^{1/2} \le M J  ,
\end{align}
with equality when $|a_0|=|b_0|=1$.
Therefore, the definition~\eqref{def_eq_bar_mJ_phi} yields $\bar{\mJ} (V_{A_0B_0})=MJ$, which reduces the inequality~\eqref{lower_bound_E_rate_1/2_pre} to the main lower bound~\eqref{lower_bound_E_rate_1/2_result1}, where the condition $z=\sqrt{M} J t/n\le 1/2$ reduces to $t\le n \sqrt{M}/[2\bar{\mJ} (V_{A_0B_0} )]$.

\subsection{Optimality of the threshold: unbounded entanglement rate for $\alpha<1/2$} \label{Optimality of the threshold: Unbounded entanglement}

We next show that the threshold $\alpha=1/2$ cannot, in principle, be removed. 
As shown below, we explicitly demonstrate that the average entanglement rate is unbounded for $\alpha<1/2$  for a finite-time interval:
\begin{prop} \label{prop:optimality_2}
Let  $A$ and $B$ each consist of $2$ qudits $\{A_0,A_1\}$ and $\{B_0,B_1\}$, respectively, where $\mD_{A_0}=\mD_{B_0}=3$ and $\mD_{A_1}=\mD_{B_1}=\mD_0$.
Then, there exists a quantum dynamics such that 
\begin{align}
\label{lower_bound_E_rate_1/2}
\abs{E_{\alpha}(\psi_t) - E_{\alpha}(\psi_{t=0}) } \ge  \frac{1-2\alpha}{1-\alpha}\log \br{ \mD_0} +  \frac{1}{1-\alpha}\log \brr{ \alpha \br{\frac{\bar{\mJ} (V_{A_0B_0}) t}{2}}^{2\alpha} }  ,
\end{align}
under an appropriate initial state $\ket{\psi_{t=0}}$ [see Eq.~\eqref{initial_state_t=0} below], where we assume
\begin{align}
\label{condition_t_optimality_alpha<1/2}
t  \le \frac{1}{\bar{\mJ} (V_{A_0B_0})} .
\end{align}
\end{prop}

{\bf Remark.} By setting $\bar{\mJ} (V_{A_0B_0})=1$ and $t=1$, the inequality~\eqref{lower_bound_E_rate_1/2} yields 
 \begin{align}
\frac{\abs{E_{\alpha}(\psi_t) - E_{\alpha}(\psi_{t=0}) }}{t}  \ge   \frac{1-2\alpha}{1-\alpha}\log \br{ \mD_0} +  \frac{1}{1-\alpha}\log \br{ \alpha 2^{-2\alpha} } .  
\end{align}
From the above bound, as long as $\alpha<1/2$, the average entanglement rate can be made arbitrarily large regardless of $\bar{\mJ} (V_{A_0B_0})=1$. 
Thus, information about the boundary interaction alone does not suffice to upper-bound the R\'enyi entanglement for $\alpha<1/2$. 
This point is compared to the geometrically local interacting systems, which will be discussed in Sec.~\ref{sec:System with short-range interactions}. 

\subsubsection{Proof of Proposition~\ref{prop:optimality_2}}

We define the interaction operator $V_{A_0B_0}$ between $A_0$ and $B_0$ as follows: 
\begin{align}
\label{Choice_V_AB_J_11}
V_{A_0B_0} := J \br{  \ket{1_{A_0}, 1_{B_0} } \bra{0_{A_0}, 0_{B_0} } + {\rm h.c.}} . 
\end{align} 
In this setup, we prove that $\bar{\mJ} (V_{A_0B_0})$ in the definition Eq.~\eqref{lemm:Renyi_SIE_main_ineq} is equal to $J$:
\begin{align}
\label{mJ_V_AB_1J}
\bar{\mJ} (V_{A_0B_0})= J  ,
\end{align} 
where the proof proceeds in a manner analogous to that of Eq.~\eqref{bar_J_exact_cal}. 
From the condition~\eqref{condition_t_optimality_alpha<1/2}, we get 
\begin{align}
\label{condition_J_t_le1}
J t \le 1 .
\end{align}

We introduce $\mD_0$ sets of Hamiltonians $\{H_{s,A} + H_{s,B}\}_{s=1}^{\mD_0}$, where $H_{s,A}$ and $H_{s,B}$ are defined as
\begin{align}
H_{s,A} =- J_A \br{  \ket{s_{A_1} , 2_{A_0}} \bra{0_{A_1}, 1_{A_0} }  + {\rm h.c.}} ,\quad 
H_{s,B} =- J_B \br{  \ket{s_{B_1}, 2_{B_0}  } \bra{0_{B_1} , 1_{B_0} }  + {\rm h.c.}} .
\end{align} 
Here, $J_A$ and $J_B$ can be chosen arbitrarily large, and hence we set them such that 
\begin{align}
&e^{-iH_{s,A}\varepsilon} \ket{0_{A_1}, 1_{A_0} }=  i \ket{s_{A_1} , 2_{A_0}} , \notag \\
& e^{-iH_{s,B}\varepsilon} \ket{0_{B_1} , 1_{B_0} }= i \ket{s_{B_1}, 2_{B_0} } ,
\end{align} 
where we choose $\varepsilon J_A =\varepsilon  J_B=\pi/2$.
Also, from Eq.~\eqref{Choice_V_AB_J_11}, we get  
\begin{align}
e^{-iV_{A_0B_0}\Delta t }  \ket{0_{A_0}, 0_{B_0}}  =  \cos \br{J\Delta t} \ket{0_{A_0}, 0_{B_0}} -i   \sin \br{J\Delta t} \ket{1_{A_0}, 1_{B_0}} .
\end{align} 

We subdivide the time interval into $\mD_0$ segments as $t/\mD_0= \Delta t$ and construct the unitary time evolution by
\begin{align}
U_{0\to t'}:= \prod_{s=1}^{\mD_0} e^{-i(H_{s,A}+H_{s,B})\varepsilon}  e^{-iV_{A_0B_0}\Delta t} ,
\end{align} 
where $t'= t+\mD_0 \varepsilon$. By construction, $t'$ can be made arbitrarily close to $t$. 
We begin with the quantum state of 
\begin{align}
\label{initial_state_t=0}
\ket{\psi_0}=\ket{0_{A_1}, 0_{A_0}, 0_{B_0} ,0_{B_1}}. 
\end{align} 

As the first process, we have 
\begin{align}
\ket{\psi_1} &:=e^{-i(H_{1,A}+H_{1,B})\varepsilon}  e^{-iV_{A_0B_0}\Delta t} \ket{\psi_0} \notag \\
&=e^{-i(H_{1,A}+H_{1,B})\varepsilon} \br{ \cos \br{J\Delta t} \ket{0_{A_1}, 0_{A_0}, 0_{B_0} ,0_{B_1}} -i   \sin \br{J\Delta t}  \ket{0_{A_1}, 1_{A_0}, 1_{B_0} ,0_{B_1}}}  \notag \\
&=  \cos \br{J\Delta t}  \ket{\psi_0}  +  \sin \br{J\Delta t}  \ket{1_{A_1}, 2_{A_0}, 2_{B_0} ,1_{B_1}}.
\end{align} 
We note that the subsequent unitary dynamics $e^{-i(H_{s,A}+H_{s,B})\varepsilon}  e^{-iV_{A_0B_0}\Delta t}$ ($s\ge 2$) does not change $\ket{1_{A_1}, 2_{A_0}, 2_{B_0} ,1_{B_1}}$, i.e., 
\begin{align}
e^{-i(H_{s,A}+H_{s,B})\varepsilon}  e^{-iV_{A_0B_0}\Delta t} \ket{1_{A_1}, 2_{A_0}, 2_{B_0} ,1_{B_1}} = \ket{1_{A_1}, 2_{A_0}, 2_{B_0} ,1_{B_1}} 
\end{align} 
for $s\ge 2$. 
In the next step, we have 
\begin{align}
&e^{-i(H_{2,A}+H_{2,B})\varepsilon}  e^{-iV_{A_0B_0}\Delta t} \ket{\psi_1} \notag \\
&= \cos \br{J\Delta t} e^{-i(H_{2,A}+H_{2,B})\varepsilon}  e^{-iV_{A_0B_0}\Delta t} \ket{\psi_0} + \sin \br{J\Delta t}  \ket{1_{A_1}, 2_{A_0}, 2_{B_0} ,1_{B_1}} 
 \notag \\
&=  \cos^2 \br{J\Delta t} \ket{\psi_0} +\cos \br{J\Delta t}  \sin \br{J\Delta t}  \ket{2_{A_1}, 2_{A_0}, 2_{B_0} ,2_{B_1}} +  \sin \br{J\Delta t}  \ket{1_{A_1}, 2_{A_0}, 2_{B_0} ,1_{B_1}} . 
\end{align} 

Repeating this type of process iteratively in an analogous manner, we obtain
\begin{align}
\label{Fibnal_state_U_0to_t'}
U_{0\to t'}\ket{\psi_0} =\cos^{\mD_0}  \br{J\Delta t}\ket{0_{A_1}, 0_{A_0}, 0_{B_0} ,0_{B_1}}  +\sum_{s=1}^{\mD_0} \cos^{s-1}  \br{J\Delta t} \sin \br{J\Delta t}  \ket{s_{A_1}, 2_{A_0}, 2_{B_0} ,s_{B_1}}  .
\end{align} 
The above expression gives the Schmidt decomposition of $U_{0\to t'}\ket{\psi_0}$, and hence the R\'enyi entanglement is given by
\begin{align}
E_\alpha(\psi_t) =\frac{1}{1-\alpha}\log \br{ \cos^{2\mD_0 \alpha} \br{J\Delta t} + \sum_{s=1}^{\mD_0} \cos^{2(s-1) \alpha} \br{J\Delta t}  \sin^{2\alpha} \br{J\Delta t} },
\end{align} 
where we take the limit of $\varepsilon\to +0$ in Eq.~\eqref{Fibnal_state_U_0to_t'}.

Eq.~\eqref{condition_J_t_le1} imposes the constraint $J\Delta t\le Jt\le 1$.  
Because of $\cos(x) \ge 1-x/2$ and $\sin(x)\ge x/2$ for $0\le x\le1$, the following relation holds:
\begin{align}
\label{Lower_bound_E_alpha_psi_t_e}
E_\alpha(\psi_t) 
&\ge \frac{1}{1-\alpha}\log \br{ \sum_{s=1}^{\mD_0+1}
 \br{1-J\Delta t/2}^{2(s-1) \alpha} \br{J\Delta t/2}^{2\alpha} } \notag \\
&= \frac{1}{1-\alpha}\log \br{\br{J\Delta t/2}^{2\alpha} \frac{1-  \br{1-J\Delta t/2}^{2\mD_0 \alpha}}{J\Delta t/2}} . 
\end{align} 
Here, we note that the condition~\eqref{condition_J_t_le1} yields
\begin{align}
\frac{J\Delta t}{2} = \frac{J t}{2\mD_0} \le \frac{1}{2\mD_0}  \le  \frac{1}{2\mD_0  \alpha} \for 0<\alpha < 1 .
\end{align} 
Replacing $J\Delta t/2$ with $x$ in and treating $x$ as a variable in the RHS of~\eqref{Lower_bound_E_alpha_psi_t_e}, we are led to consider the function of
\begin{align}
x^{2\alpha} \frac{1-  \br{1-x}^{2\mD_0 \alpha}}{x} \for  0\le x \le \frac{1}{2\mD_0 \alpha} . 
\end{align} 
Since $\br{1-x}^{2\mD_0 \alpha} \le 1 - \mD_0 \alpha x$, we obtain
\begin{align}
x^{2\alpha} \frac{1-  \br{1-x}^{2\mD_0 \alpha}}{x} \ge \mD_0 \alpha x^{2\alpha} .
\end{align} 
Therefore, we reduce the lower bound~\eqref{Lower_bound_E_alpha_psi_t_e} to the main inequality:
\begin{align}
E_\alpha(\psi_t) 
&\ge  \frac{1}{1-\alpha}\log \brr{ \mD_0 \alpha \br{\frac{Jt}{2\mD_0}}^{2\alpha} } \notag \\
&= \frac{1-2\alpha}{1-\alpha}\log \br{ \mD_0} +  \frac{1}{1-\alpha}\log \brr{ \alpha \br{\frac{\bar{\mJ} (V_{A_0B_0}) t}{2}}^{2\alpha} }  ,
\end{align} 
where we use Eq.~\eqref{mJ_V_AB_1J} in the last equation. 
Note that $E_\alpha(\psi_{t=0}) =0$ from the definition~\eqref{initial_state_t=0}. 
This completes the proof. $\square$

\section{Operator approximation vs. entanglement generation} \label{sec:Op_approx_Dynamical_ent}

Here, we consider the following question: \textit{``Does a small spectral entanglement rate imply an efficient approximation with low Schmidt rank?''}
In general, the answer to this question is no.
Consequently, we obtain a complexity separation---measured in terms of the Schmidt‑rank required for a given approximation error---between approximating the time‑evolution unitary and approximating time‑evolved states.

Indeed, as we demonstrate below, the following proposition holds.
\begin{prop} \label{No_go_theorem_approx}
Let us consider a class of quantum dynamics $e^{-iH_{AB} t}$ such that 
\begin{align}
\bar{\mJ}(V_{AB}) = 1 .
\end{align} 
Then, there always exists a dynamics $e^{-iH_{AB} t}$ such that the error cannot be reduced for any operator $U_{AB,D}$ with a finite Schmidt-rank truncation:
\begin{align}
\label{error_between_dynamics_approx_main/results}
\inf_{U_{AB,D}} \norm{ e^{-iH_{AB}t} -U_{AB,D}}
\ge  1+ \frac{3t}{2}- e^t ,
\end{align} 
where $\inf_{U_{AB,D}} $ is taken over all the operators such that the Schmidt rank is smaller than or equal to $D$.  
\end{prop}

{\bf Remark.}
In this setup, for any low-entangled initial state, the corresponding time-evolved state can be well-approximated by a finite Schmidt rank $D$, whereas the full unitary time evolution admits no such approximation.

Here, we note that the Schmidt rank $D$ is assumed to be independent of the system size (i.e., the Hilbert space dimension of the total system).
If, however, $D$ is allowed to depend on the total Hilbert space dimension $\mD_{AB}$, one can obtain an approximation whose accuracy improves with $D$.
Indeed, as shown in~\eqref{Upper_lower_bound_Kolmogorov}, our example admits a good approximation with error decaying as $\orderof{D^{-1/2}}$, provided that $D \gtrsim \log(\mD_{AB})$.
A similar type of approximation has also been reported for dynamics generated by the quantum Fourier transform~\cite{PRXQuantum.4.040318,chen2024QFT}.
An interesting open question is whether this bound can be quantitatively improved; specifically, whether there exist dynamics that cannot be approximated even when $D \lesssim \poly(\mD_{AB})$.

\subsection{Proof of Proposition~\ref{No_go_theorem_approx}}

Consider a four-qudit system comprising subsystems $A_0$, $A_1$, $B_0$ and $B_1$, where $\mD_{A_0}=\mD_{B_0}=N$ with $N$ chosen sufficiently large.  
We then consider the Hamiltonian of the form
\begin{align}
\label{H_AB_Ising_int}
H_{AB} =\tilde{V}_{A_0B_0} = \sum_{s=1}^N  (\ket{s} \bra{s})_{A_0} \otimes (\ket{s}\bra{s})_{B_0} ,
\end{align} 
where we choose the boundary interaction $V_{AB}$ as the above $\tilde{V}_{A_0B_0}$. 
By the same calculation as in \eqref{bar_J_exact_cal_prove} used to prove Eq.~\eqref{bar_J_exact_cal}, we obtain
\begin{align}
\label{H_AB_Ising_int_exact_value}
\bar{\mJ}(\tilde{V}_{A_0B_0} ) = 1 . 
\end{align}

We then consider a time evolution such that 
\begin{align}
e^{-iH_{AB}  t} = 1 - i\tilde{V}_{A_0B_0}  t + \sum_{m=2}^\infty \frac{(-it)^m}{m!} \tilde{V}_{A_0B_0} ^m ,
\end{align} 
which satisfies 
\begin{align}
\norm{ e^{-iH_{AB} t} -\br{ 1 - i\tilde{V}_{A_0B_0}  t}} \le \sum_{m=2}^\infty \frac{t^m}{m!} \norm{\tilde{V}_{A_0B_0} }^m \le e^t-1-t ,
\end{align} 
where we use $\norm{\tilde{V}_{A_0B_0} }=1$.

We then define $U_{AB,D}$ as the optimal approximation of $e^{-iH_{AB}  t}$. 
Using this, we find the following relations:
\begin{align}
\label{error_between_dynamics_approx}
 \norm{ e^{-iH_{AB}t} -U_{AB,D}}
 & \ge \norm{1 - i \tilde{V}_{A_0B_0} t - U_{AB,D}} - \norm{ e^{-iH_{AB}t} -\br{ 1 - i \tilde{V}_{A_0B_0} t}} \notag \\
 & \ge  t \inf_{W_{2D}}    \norm{ \tilde{V}_{A_0B_0} - W_{2D}} - \br{e^t-1-t } .
\end{align} 
Note that we used
\begin{align}
 \norm{1 - i \tilde{V}_{A_0B_0} t - U_{AB,D}} 
 &=  \norm{1 - i \tilde{V}_{A_0B_0} t - \br{ {\rm Re}( U_{AB,D} ) + i {\rm Im}( U_{AB,D} )}} \notag \\
 &\ge  t \norm{\tilde{V}_{A_0B_0} - \frac{1}{t} {\rm Im}( U_{AB,D})} \notag \\
 &\ge t  \inf_{W_{2D}}    \norm{ \tilde{V}_{A_0B_0} - W_{2D}} ,
\end{align} 
This follows from the inequality $\norm{O_1+ i O_2} \ge \norm{O_2}$, which holds for any Hermitian operators $O_1$ and $O_2$.
The infimum $\inf_{W_{2D}}$ is taken over all the operators $W_{2D}$ with the Schmidt rank of $\le 2D$ between $A$ and $B$. 
Note that ${\rm Re}( U_{AB,D} ) = (U_{AB,D} + U_{AB,D}^\dagger)/2$ and ${\rm Im}( U_{AB,D} )= -i (U_{AB,D} - U_{AB,D}^\dagger)/2$, and ${\rm Im}( U_{AB,D} )$ can have the Schmidt rank up to $2D$.

The remaining task is to estimate a lower bound for the low-rank approximation of $\tilde{V}_{A_0B_0}$.
In detail, we calculate
\begin{align}
\label{prob_main_aprpox_rilde_V}
\delta_{2D} :=\inf_{W_{2D}}  \norm{ \tilde{V}_{A_0B_0}  - W_{2D}}   .
\end{align} 
We define $\tilde{W}_{2D}$ as
\begin{align}
\label{W_s_s'_definiton}
\tilde{W}_{2D} := \arg \inf_{W_{2D}}    \norm{ \tilde{V}_{A_0B_0} - W_{2D}} ,
\end{align} 
and prove that the form of $\tilde{W}_{2D}$ should be given by 
\begin{align}
\label{W_s_s'_definiton_form}
\sum_{s,s'=1}^NW_{s,s'}  (\ket{s} \bra{s})_{A_0} \otimes (\ket{s'}\bra{s'})_{B_0}  ,
\end{align} 
where $W_{s,s'} \in \mathbb{R}$ for all $s,s'$.

First, we demonstrate that the optimal operator $\tilde{W}_{2D}$ is Hermitian. 
Let us decompose an arbitrary (which may not be Hermitian) $W_{2D}$ as follows:
\begin{align}
W_{2D} = {\rm Re}(W_{2D}) + i {\rm Im}(W_{2D}) .
\end{align} 
Then, for an arbitrary quantum state $\ket{\psi}$, we have
\begin{align}
\abs{\bra{\psi} \tilde{V}_{A_0B_0} - {\rm Re}(W_{2D}) + i {\rm Im}(W_{2D})\ket{\psi} }  
&=\sqrt{(\bra{\psi} \tilde{V}_{A_0B_0}  - {\rm Re}(W_{2D})\ket{\psi} )^2 +(\bra{\psi}{\rm Im}(W_{2D})\ket{\psi} )^2 }  \notag \\
& \ge \abs{(\bra{\psi} \tilde{V}_{A_0B_0}  - {\rm Re}(W_{2D})\ket{\psi}} .
\end{align} 
Therefore, the non-Hermitian term always increases the approximation error, which implies that $W_{2D}$ must be Hermitian.

Next, we consider the component acting nontrivially on the subset $A_1B_1$ and prove that it cannot contribute to the error reduction. 
Let $W_{2D}$ be Hermitian. We decompose the operator $W_{2D}$ into 
\begin{align}
W_{2D}=W_{2D,A_0B_0} + \delta W_{2D,AB} , \quad \tr_{A_1B_1} \br{\delta W_{2D,AB}}=0 ,
\end{align} 
Note that $W_{2D,A_0B_0}$ abbreviates $W_{2D,A_0B_0} \otimes \hat{1}_{A_1B_1}$. 
We choose a quantum state $\rho$ defined as
\begin{align}
\rho = \ket{\psi_{A_0B_0}} \bra{\psi_{A_0B_0}} \otimes \frac{\hat{1}_{A_1B_1}}{\mD_{A_1B_1}} ,
\end{align} 
where the state $\ket{\psi_{A_0B_0}}$ satisfies 
\begin{align}
\norm{ \tilde{V}_{A_0B_0} - W_{2D,A_0B_0} } = \abs{ \bra{\psi_{A_0B_0}}  \tilde{V}_{A_0B_0} - W_{2D,A_0B_0} \ket{\psi_{A_0B_0}} } .
\end{align} 
For any $\delta W_{2D,AB} $, the following inequality holds:
\begin{align}
\norm{ \tilde{V}_{A_0B_0} - W_{2D,A_0B_0} } 
&= \abs{ \tr \brr{\rho \br{ \tilde{V}_{A_0B_0} - W_{2D,A_0B_0} -\delta W_{2D,AB} } }} \notag \\
&\le \norm{ \tilde{V}_{A_0B_0} - W_{2D,A_0B_0} -\delta W_{2D,AB}  }.
\end{align} 
This shows that $\norm{ \tilde{V}_{A_0B_0} - W_{2D,A_0B_0} -\delta W_{2D,AB}  }$ is still larger than or equal to $\norm{  \tilde{V}_{A_0B_0} - W_{2D,A_0B_0}}$.
We thus conclude that the existence of $\delta W_{2D,AB}$ cannot reduce the error of $\norm{ \tilde{V}_{A_0B_0} - W_{2D,A_0B_0}}$ for an arbitrary choice of $W_{2D,A_0B_0}$.

Finally, we consider the off-diagonal part of the operator on the subset $A_0B_0$.
We also decompose $W_{2D}$ into 
\begin{align}
W_{2D}=W_{\rm diag} + W_{\rm off} ,
\end{align} 
where $W_{\rm diag}$ is given in the form of Eq.~\eqref{W_s_s'_definiton_form}. 
For any choices of $\forall W_{\rm diag} $, we can also prove that the existence of $W_{\rm off}$ does not reduce the error $\norm{\tilde{V}_{A_0B_0}- W_{\rm diag}}$. 
Suppose that the error $\norm{\tilde{V}_{A_0B_0}-  W_{\rm diag}}$ attains its maximum for the quantum state $\ket{s_0,s_0'}$, which is an eigenstate of  
$\tilde{V}_{A_0B_0}- W_{\rm diag}$. 
In this situation, we have $\bra{s_0,s_0'}  W_{\rm off}  \ket{s_0,s_0'}=0$; hence the error $\norm{\tilde{V}_{A_0B_0}- W_{\rm diag}-W_{\rm off}}$ is non-decreasing when $W_{\rm off}$ is present. 
In conclusion, we need to consider the class of operators~\eqref{W_s_s'_definiton_form} for the solution of $\inf_{W_{2D}}  \norm{ \tilde{V}_{A_0B_0}  - W_{2D}}$.

Second, we consider the approximation error $\delta_{2D}$ in Eq.~\eqref{prob_main_aprpox_rilde_V}. 
For the purpose, using the fact that $\tilde{W}_{2D}$ is given in the form of Eq.~\eqref{W_s_s'_definiton_form}, we consider 
\begin{align}
 \norm{ \tilde{V}_{A_0B_0}- \sum_{s,s'=1}^NW_{s,s'}  (\ket{s} \bra{s})_{A_0} \otimes (\ket{s'}\bra{s'})_{B_0} } =  \max_{1\le s,s' \le N}  \abs{ \delta_{s,s'} - W_{s,s'} }  ,
\end{align} 
where the rank of the matrix $\{W_{s,s'}\}_{s,s'}$ is less than or equal to $2D$. 
Hence, the problem is equivalent to estimating 
\begin{align}
\label{Matrix_approx_Identity}
\inf_{{\rm rank}(W) \le 2D} \max_{1\le s,s' \le N} \abs{ \delta_{s,s'} - W_{s,s'} }  .
\end{align} 
The above quantity is known to be equal to the Kolmogorov width of the octahedron~\cite{Kas74,Gluskin1986,Foucart2013,malykhin2024}, denoted by $d_{2D}(B^N_1 ,\ell_\infty^N)$ (see also Sec.~\ref{Sec:A short review on the Kolmogorov width} for a detailed review). 
For $d_D(B^N_1 ,\ell_\infty^N)$, from Refs.~\cite{Kas74} and \cite[Proposition 10.10 therein]{Foucart2013}, the Kolmogorov width obeys the inequality
\begin{align}
\label{Upper_lower_bound_Kolmogorov}
\frac{1}{2} \min \brr{ \frac{2}{1+ 4\ln(9)} \frac{\log(eN/D)}{D} ,1 } \le d_D(B^N_1 ,\ell_\infty^N) \le 2\br{\frac{\log(N)}{D}}^{1/2} . 
\end{align} 
Note that the lower bound $d_D(B^N_1 ,\ell_\infty^N)=1/2$ is achieved by choosing $W_{s,s'}=1/2$ for all $s,s'$, which has rank $1$.  

From these results, as long as $N$ is much larger than $e^{D}$, one can always derive 
\begin{align}
\label{Kolmogorov_upper_analysis}
\inf_{{\rm rank}(W) \le 2D} \max_{1\le s,s' \le N} \abs{ \delta_{s,s'} - W_{s,s'} }  \ge \frac{1}{2} ,
\end{align} 
which also provides $ \inf_{W_{2D}} \norm{ \tilde{V}_{A_0B_0} - W_{2D} } \ge 1/2$.
By combining this upper bound with the inequality~\eqref{error_between_dynamics_approx}, we prove the main inequality~\eqref{error_between_dynamics_approx_main/results}. 
This completes the proof. $\square$

\subsubsection{A short review on the Kolmogorov width} \label{Sec:A short review on the Kolmogorov width}

We now consider $d_D(B^N_1 ,\ell_\infty^N)$ instead of $d_{2D}(B^N_1 ,\ell_\infty^N)$ for simplicity of notation. 
We introduce the unit ball $B^N_1$ as follows:
\begin{align}
B^N_1 = \{ \vec{x} \in \mathbb{R}^N: \norm{\vec{x}}_{\ell_1} \le 1 \} ,
\end{align} 
where $\norm{\vec{x}}_{\ell_m}$ ($m\in \mathbb{N}$) denotes the $\ell_m$ norm, i.e., $\norm{\vec{x}}_{\ell_m}= \br{\sum_{j} x_j^m}^{1/m}$.
We also denote by $\ell^N_\infty$ the space $\mathbb{R}^N$ equipped with the $\ell_\infty$ norm
We define the Kolmogorov width as 
\begin{align}
\label{Kolmogorov_width_derf}
d_D(B^N_1 ,\ell_\infty^N) := \inf_{\substack{ X_D \subset \ell^N_\infty \\  \dim(X_D) \le D}} \sup_{\vec{x}\in B_1^N} \inf_{\vec{y}\in X_D} \norm{\vec{x}-\vec{y}}_{\ell_\infty} . 
\end{align} 

To connect $d_D(B^N_1 ,\ell_\infty^N)$ to Eq.~\eqref{Matrix_approx_Identity}, we first denote an arbitrary vector $\vec{y}$ in the space $X_D$ by 
\begin{align}
y &=
\sum_{j=1}^D 
 \begin{pmatrix} w^{(j)}_{1} \\
   w^{(j)}_{2} \\
    \vdots    \\ 
   w^{(j)}_{N} 
   \end{pmatrix}    \begin{pmatrix}w^{(j)}_{1} &&w^{(j)}_{2}&& \cdots &&w^{(j)}_{N}    \end{pmatrix}
   \begin{pmatrix} x'_1 \\
   x'_2\\
    \vdots    \\ 
   x'_{N} 
   \end{pmatrix}   \notag \\
   &=\sum_{j=1}^D \vec{w}^{(j)} \vec{w}^{(j) \mathrm{T}} \vec{x}' = W \vec{x}',
\end{align}  
with $x' \in \ell^N_\infty$, where each of the vectors $\vec{w}^{(j)} =\{w^{(j)}_{m}\}_{1\le m\le N}$ forms an orthonormal basis. 
The matrix $W$ can be viewed as the projection onto the subspace spanned by the $D$ vectors $\{\vec{w}^{(j)}\}_{j=1}^D$.

Then, we immediately obtain the upper bound for the Kolmogorov width as follows:
\begin{align}
\label{Upper_bound_Kolmogorov}
 \inf_{\substack{ X_D \subset \ell^N_\infty \\  \dim(X_D) \le D}} \sup_{\vec{x}\in B_1^N} \inf_{\vec{y}\in X_D} \norm{\vec{x}-\vec{y}}_{\ell_\infty} 
 &\le  \inf_{{\rm rank}(W)=D} \sup_{\vec{x}\in B_1^N}  \norm{\vec{x}-W \vec{x} }_{\ell_\infty}  \notag \\
 &=  \inf_{{\rm rank}(W)=D}  \norm{1 -W}_{\rm max} =  \inf_{{\rm rank}(W)=D}  \max_{1\le s,s' \le N}  \abs{\delta_{s,s'} - w_{s,s'}} ,
\end{align} 
where we use $\sup_{\vec{x}\in B_1^N}\norm{M \vec{x}}_{\ell_\infty} = \norm{M}_{\rm max}:=  \max_{1\le s,s' \le N}  |M_{s,s'}|$ for an arbitrary matrix $M=\{M_{s,s'}\}$. 

Next, we prove that the LHS of the inequality~\eqref{Upper_bound_Kolmogorov} also provides the lower bound for the Kolmogorov width. 
We consider the choices of $\vec{x}$ as  
\begin{align}
\vec{x} = \vec{I}_s := \{\overbrace{0,0,\ldots, 0}^{s-1} ,1 \overbrace{0, \ldots,0}^{N-s}\} 
\end{align} 
for $s=1,2,\ldots , N$. For a given space $X_D$, we define 
 \begin{align}
\arg \inf_{\vec{y}\in X_D} \norm{\vec{I}_s -\vec{y}}_{\ell_\infty} = \vec{y}_s .
\end{align} 

Then, by defining the rank $D$ matrix $W(X_D)$ as 
 \begin{align}
W(X_D) =   \begin{pmatrix} \vec{y}_1 && \vec{y}_2&& \cdots && \vec{y}_N    \end{pmatrix} ,
\end{align} 
we derive   
\begin{align}
 \max_{1\le s,s' \le N}  \abs{\delta_{s,s'} - [W(X_D)]_{s,s'} } \le \max_s \norm{\vec{I}_s-\vec{y}_s  }_{\ell_\infty} 
&= \max_s  \inf_{\vec{y}\in X_D}  \norm{\vec{I}_s -\vec{y}}_{\ell_\infty}  \notag \\
&\le \sup_{\vec{x}\in B_1^N} \inf_{\vec{y}\in X_D}  \norm{\vec{x} -\vec{y}}_{\ell_\infty}   .
\end{align} 
This also implies 
\begin{align}
\label{lower_bound_Kolmogorov}
 \inf_{\substack{ X_D \subset \ell^N_\infty \\  \dim(X_D) \le D}} \sup_{\vec{x}\in B_1^N} \inf_{\vec{y}\in X_D}  \norm{\vec{x} -\vec{y}}_{\ell_\infty}   
 &\ge  \inf_{\substack{ X_D \subset \ell^N_\infty \\  \dim(X_D) \le D}}  \max_{1\le s,s' \le N}  \abs{\delta_{s,s'} - [W(X_D)]_{s,s'} } \notag \\
 &\ge \inf_{{\rm rank}(W)=D}   \max_{1\le s,s' \le N}  \abs{\delta_{s,s'} - W_{s,s'} }  .
\end{align} 
By combining the inequalities~\eqref{Upper_bound_Kolmogorov} and \eqref{lower_bound_Kolmogorov}, we obtain the desired identity as 
\begin{align}
d_D(B^N_1 ,\ell_\infty^N) = \inf_{{\rm rank}(W) \le D} \max_{1\le s,s' \le N} \abs{ \delta_{s,s'} - W_{s,s'} }  . 
\end{align}

\subsection{Conjecture on the operator low-Schmidt-rank approximation} \label{Sec:Conj_alpha}

As shown in Proposition~\ref{No_go_theorem_approx}, the condition that $V_{AB}$ (as well as $e^{-iH_{AB}t}$) has small SE strength (i.e., $\bar{\mJ}(V_{AB})=1$) does not by itself guarantee that the time‑evolution operator can be efficiently approximated by a low Schmidt-rank operator. It therefore remains an important open problem to determine under what additional conditions such efficient operator approximations become possible. To address this issue, we introduce an extension of the SE strength, namely the $\alpha$-SE strength $\bar{\mJ}_\alpha$ (Definition~\ref{Def:Interaction_strength_renyi}), and conjecture that having small $\alpha$-SE strength provides a sufficient condition for efficient low-rank approximation of the unitary dynamics $e^{-iH_{AB} t}$ (Conjecture~\ref{conj:operator_approx} below).

We first generalize the SE strength in Def.~\ref{Def:Interaction_strength}: 
\begin{definition}[$\alpha$-spectral Entangling (SE) Strength]
\label{Def:Interaction_strength_renyi}
Let $\Phi_{AB}$ be an arbitrary operator acting across subsystems $A$ and $B$. For any product state $\ket{\phi} = \ket{\phi_{AA'}} \otimes \ket{\phi_{BB'}}$ with the ancillas $A'$ and $B'$, we denote the Schmidt decomposition of $\Phi_{AB} \ket{\phi}$ by
\begin{align}
\Phi_{AB} \ket{\phi} = \sum_s \lambda_s(\phi) \ket{\phi_{AA',s}} \otimes \ket{\phi_{BB',s}}.
\end{align}
Then, the $\alpha$-SE strength of $\Phi_{AB}$ is defined as
\begin{align}
\bar{\mathcal{J}}_\alpha(\Phi_{AB}) := \sup_{\ket{\phi}} \br{\sum_s |\lambda_s(\phi)|^{2\alpha}}^{1/(2\alpha)},
\label{def_eq_bar_mJ_phi_renyi}
\end{align}
where the supremum is taken over all product states $\ket{\phi} = \ket{\phi_{AA'}} \otimes \ket{\phi_{BB'}}$.
\end{definition}

Using this definition, we propose the following conjecture: 
\begin{conj} \label{conj:operator_approx}
For any quantum operator such that 
\begin{align}
\label{condition_for_alpha_SE}
\bar{\mJ}_\alpha(\Phi_{AB}) =1, \quad \alpha<1/2 ,
\end{align} 
there exists an efficient approximation of $\Phi_{AB}$ by an operator $\Phi_{AB,D}$ with ${\rm SR}(\Phi_{AB,D})=D$:  
\begin{align}
\label{condition_for_alpha_SE_main_statement}
\inf_{\Phi_{AB,D}}  \norm{\Phi_{AB}-\Phi_{AB,D} } \le g_{\alpha}(D) ,
\end{align}
where the function $g_{\alpha}(D)$ decays as a power law in$D$, with the decay rate depending on $\alpha$. 
\end{conj}

{\bf Remark.} 
A frequently asked question is why the standard operator Schmidt decomposition cannot be directly applied to address Conjecture~\ref{conj:operator_approx}. 
Recall that the operator Schmidt decomposition expands a given operator $O$ as
\begin{align}
O=\sum_s \mu_s O_{A,s} \otimes O_{B,s},
\end{align}
where $\{O_{A,s}\}$ and $\{O_{B,s}\}$ form operator bases. 
It is well known that such a decomposition provides an optimal low-rank approximation with respect to the Frobenius norm:
\begin{align}
\norm{O - O_D}_F =\sqrt{ \tr \brr{ \br{O - O_D}^\dagger \br{O - O_D}} },
\end{align}
where $O_D$ is a low-rank approximation of the operator $O$. 

However, in many applications, the relevant notion of approximation is the operator norm rather than the Frobenius norm; for instance, in the case of unitary time-evolution operators $e^{-iH_{AB}t}$, the operator norm controls the worst-case state error relevant to complexity and simulation.
The approximation properties with respect to the Frobenius norm and the operator norm are qualitatively different. 
For example, the quantum Fourier transform is known to exhibit maximal operator entanglement when viewed through the lens of the operator Schmidt decomposition~\cite{PhysRevA.67.052301,Tyson_2003,10.5555/3179483.3179484}; with respect to the operator norm, however, efficient low‑rank approximations exist~\cite{PRXQuantum.4.040318}. 
This discrepancy explains why the operator Schmidt decomposition, while mathematically natural, is insufficient for analyzing low-rank approximations of dynamics with respect to the operator norm.

Under the condition of $\bar{\mJ}_\alpha(V_{AB}) =1$, we can exclude the interaction in Eq.~\eqref{H_AB_Ising_int}, i.e., $\tilde{V}_{A_0B_0} = \sum_{s=1}^N  (\ket{s} \bra{s})_{A_0} \otimes (\ket{s}\bra{s})_{B_0}$.
A simple computation yields
\begin{align}
\bar{\mJ}_\alpha(\tilde{V}_{A_0B_0}) \ge N^{1/(2\alpha)-1} ,
\end{align} 
which increases with the local Hilbert space dimension $N$. 

Finally, in the absence of ancillas, the conjecture is trivially false. 
In this case, the swap operator $S_{AB}$ maps between $A$ and $B$ serves as a counterexample.
It creates no entanglement but cannot be approximated by a finite Schmidt rank, as can be shown using similar arguments to~\eqref{Kolmogorov_upper_analysis}.

\subsection{Theorem on low-rank approximation of operators}

The conjecture above suggests that small $\alpha$-SE strength may suffice for efficient low-rank approximation, but establishing such a result remains challenging under this purely information-theoretic condition. In this section, to gain further insight, we turn to a complementary approach: by imposing stronger structural constraints on the generator of the dynamics, one can rigorously prove the existence of efficient low-rank approximations. Theorem~\ref{thm:Schmidt_rank_truncation} below provides such constructive evidence, thereby reinforcing the broader intuition underlying Conjecture~\ref{conj:operator_approx} from a different perspective.

\subsubsection{Stronger assumption}

Instead of the small $\alpha$-SE strength, we consider the following property: 

\begin{assump} \label{Assumpt_Schmidt_rank_truncation}
Given the decomposition of $V_{AB}$ into operator bases $\{V_j\}_j$, each having small Schmidt rank:
 \begin{align}
 \label{Decomposition_V_A_B_assump}
V_{AB} =\sum_{j=1}^\infty V_j , \quad {\rm SR}(V_j) \le D_0 \ \ \forall j ,
\end{align}
where $D_0$ is a constant of order $\orderof{1}$.  
Then, there exist constants $C_0$ and $\kappa$ such that the following inequality holds:
 \begin{align}
 \label{primary_assumption_V_AB}
\sum_{j\ge D+1} \norm{V_j} \le C_0 \tilde{g} (D+1)^{-\kappa}, 
\quad (C_0\ge 1), \quad  
\tilde{g}:= \sum_{j=1}^\infty \norm{V_j} .
\end{align}
\end{assump}

{\bf Remark.} The inequality~\eqref{primary_assumption_V_AB} also implies a low-rank approximation of $V_{AB}$: 
  \begin{align}
&\norm{V_{AB} - V_{AB}^{(D)} } \le C_0 \tilde{g} (D+1)^{-\kappa} ,\quad 
V_{AB}^{(D)}= \sum_{j=1}^D V_j   .
\end{align}
We also emphasize that under Assumption~\ref{Assumpt_Schmidt_rank_truncation} $V_{AB}$ has a constant $\alpha$-SE strength with $\alpha<1/2$  as shown in the following lemma:
\begin{lemma} \label{lemm:alpha_SE_strength}
Under the assumption~\ref{Assumpt_Schmidt_rank_truncation} with $D_0=1$, the interaction $V_{AB}$ satisfies 
 \begin{align}
 \label{lemm:alpha_SE_strength/main}
\bar{\mJ}_{\alpha}(V_{AB}) \le \frac{2^{1/(2\alpha)-1} C_0 \tilde{g}}{\brr{1-2^{1-2\alpha(1 + \kappa)}}^{1/(2\alpha)}}  \for  \frac{1}{2(1+\kappa)}<\alpha \le \frac{1}{2} . 
\end{align}
\end{lemma}

{\bf Remark.} From the lemma, we obtain the converse of Conjecture~\ref{conj:operator_approx}; namely, a low-rank approximation of the operator (Assumption~\ref{Assumpt_Schmidt_rank_truncation}) implies a small $\alpha$-SE strength for $\alpha<1/2$. 
For general $D_0$, one should multiply the right-hand side of~\eqref{lemm:alpha_SE_strength/main} by the factor $D_0^{1/(2\alpha)}$.

{~}

\textit{Proof of Lemma~\ref{lemm:alpha_SE_strength}.}
For an arbitrary product state $\ket{\phi_0}$, $V_{AB} \ket{\phi_0}$ can be expressed as
 \begin{align}
V_{AB} \ket{\phi_0} = \sum_{j}V_j  \ket{\phi_0}=  \sum_{j} g_j  \ket{\phi_{A,j}} \otimes \ket{\phi_{B,j}} ,
\end{align}
where $g_j \le \norm{V_j}$ and $\{ \ket{\phi_{A,j}},\ket{\phi_{B,j}}\}_j$ are normalized states on $A$ and $B$, respectively, that are not orthogonal to each other.  
Next, we consider the Schmidt decomposition of $V_{AB} \ket{\phi_0} $ as 
 \begin{align}
V_{AB} \ket{\phi_0} =\sum_s \lambda_s \ket{\phi_{A,s}} \otimes  \ket{\phi_{B,s}} . 
\end{align}
Now, observe that the following inequality holds:
 \begin{align}
 \label{lambda_s_sum_g_upp}
\sum_s \lambda_s^{2\alpha} \le   \sum_{j} g^{2\alpha}_j   \for  0< \alpha \le \frac{1}{2}.
\end{align}
The above inequality follows directly from the matrix inequality as~\cite{McCarthy1967,Rotfeld1969,10.1214/10-AOS860} 
 \begin{align}
\norm{\sum_{j} M_j}_p^p \le \sum_{j}  \norm{M_j}_p^p  \for 0< p \le 1 ,
\end{align}
where $\norm{\cdots}_p$ denotes the Schatten $p$-quasi norm $\norm{M}_p:= [\tr (|M|^p)]^{1/p}$ ($p\le 1$). 

We set $N_m=2^m$ and rewrite $ \sum_{j} g^{2\alpha}_j $ as double sums, which leads to the following inequalities:
 \begin{align}
  \label{lambda_s_sum_g_upp_2}
 \sum_{j} g^{2\alpha}_j  
 &=\sum_{m=1}^\infty \sum_{j\in [N_{m-1} N_m)} g^{2\alpha}_j   \notag \\
 &\le \sum_{m=1}^\infty \br{ \sum_{j\in [N_{m-1} N_m)} 1}^{1-2\alpha} \br{ \sum_{j\in [N_{m-1} N_m)} g_j}^{2\alpha}    \notag \\
 &\le \sum_{m=1}^\infty  N_m^{(1-2\alpha)} \cdot \brr{C_0 \tilde{g} (N_{m-1})^{-\kappa}}^{2\alpha} ,
\end{align}
where we apply the H\"older inequality in the first inequality, and use the upper bound~\eqref{primary_assumption_V_AB} in the last inequality. 
Finally, we calculate 
 \begin{align}
   \label{lambda_s_sum_g_upp_3}
\sum_{m=1}^\infty  N_m^{(1-2\alpha)} (N_{m-1})^{-2\alpha \kappa} 
&=2^{2\alpha \kappa}  \sum_{m=1}^\infty  2^{m(1-2\alpha) - 2\alpha \kappa m} \notag \\
&= \frac{2^{1-2\alpha}}{1-2^{1-2\alpha(1 + \kappa)}}.
\end{align}

By applying the inequalities~\eqref{lambda_s_sum_g_upp_2} and \eqref{lambda_s_sum_g_upp_3} to \eqref{lambda_s_sum_g_upp}, we obtain
 \begin{align}
 \label{lambda_s_sum_g_upp_4}
\sum_s \lambda_s^{2\alpha} \le \frac{2^{1-2\alpha} \br{C_0 \tilde{g}}^{2\alpha}}{1-2^{1-2\alpha(1 + \kappa)}},
\end{align}
which reduces to the main inequality~\eqref{lemm:alpha_SE_strength/main} from Definition~\ref{Def:Interaction_strength_renyi}. 
This completes the proof. $\square$

\subsubsection{Result on the low-rank approximation}

We show the following theorem on the operator approximation.
\begin{theorem} \label{thm:Schmidt_rank_truncation}
Let us introduce a parameter $\mQ$ that satisfies the following inequality:
 \begin{align}
 \label{introduce_parameter_mQ}
\max_{s\in [0,\infty)} \frac{1}{s!}\norm{\ad_{H_0}^s (V_j)}  \le \mQ^s \norm{V_j} , 
\end{align}
where $\{V_j\}_j$ are operator bases that constitute the boundary interaction terms as in Eq.~\eqref{Decomposition_V_A_B_assump}.
Then, under Assumption~\ref{Assumpt_Schmidt_rank_truncation}, for both the unitary time evolution $e^{-iHt}$ and the imaginary time evolution $e^{\beta H}$, we can construct approximate operators $\tilde{U}_t$ and $\tilde{\rho}_\beta$ such that 
\begin{align}
\norm{ e^{-iHt} - \tilde{U}_t } \le \epsilon_0 ,\quad \norm{ e^{\beta H} - \tilde{\rho}_\beta }_p \le \epsilon_0 \norm{e^{\beta H}}_p,
\end{align}
and 
\begin{align}
\label{thm:Schmidt_rank_truncation/main}
&{\rm SR}(\tilde{U}_t)  \le \br{\frac{8\ceil{t\mQ_0}}{\epsilon_0}}^{\brr{6+4/\kappa+ \log_2(D_0)}\ceil{t\mQ_0}},  \notag \\
& {\rm SR}(\tilde{\rho}_\beta)  \le \br{ \frac{48 \ceil{\beta\mQ_0}}{ \epsilon_0}}^{2\brr{6+4/\kappa+ \log_2(D_0)} \ceil{\beta\mQ_0}}  ,
\end{align}
respectively, where $\mQ_0$ is defined as 
\begin{align}
\label{Definition_maQ_0_thm}
\mQ_0:= \brr{ \min \br{ \frac{1}{4\mQ} ,\frac{1}{4 e C_0 \tilde{g}} }}^{-1} . 
\end{align}
Note that $\norm{\cdots}_p$ is the Schatten $p$ norm, i.e., $\norm{O}_p = [\tr\br{\abs{O}^p}]^{1/p}$. 
\end{theorem}

{\bf Remark.}  The theorem shows that the required Schmidt rank grows only polynomially with respect to the inverse error $1/\epsilon_0$, with an exponent proportional to $t\mQ_0$. Importantly, the result guarantees efficient operator approximation provided that the generator $V_{AB}$ itself admits a suitable low-rank decomposition. Although Assumption~\ref{Assumpt_Schmidt_rank_truncation} may appear rather strong at first sight, we will show in Lemma~\ref{Lem:Long-range_decomposition} that it is in fact satisfied by a broad class of physically relevant systems, including long-range or power-law decaying interactions. This observation serves as a key ingredient for the proof of Theorem~\ref{poly_approx:MPO_gibbs}, where we establish polynomially efficient approximations of long-range quantum Gibbs states.

\subsection{Proof of Theorem~\ref{thm:Schmidt_rank_truncation}}

We first define the following merging operator between the subsets $A$ and $B$:
\begin{align}
\Psi   &:=e^{-z H_{0} } e^{z(H_{0}+ V_{AB})}= \mathcal{T} \left( e^{z\int_0^{1} V_{AB}(x)  dx} \right)  \notag \\
&= \sum_{s = 0}^{\infty} z^s \int_{0}^{1} dx_{1} \int_{0}^{x_{1}} dx_{2} \cdots \int_{0}^{x_{s-1}} dx_{s}  
    V_{AB}(x_1) V_{AB}(x_2)  \cdots V_{AB}(x_s)  , \label{supp_interaction_pic_general}
\end{align}
where $V_{AB}(x) = e^{-x z H_0} V_{AB} e^{xz H_0}$ and $z \in \mathbb{C}$.
Using it, we have
\begin{align}
e^{z(H_{0}+ V_{AB})} = e^{z H_{0} }\Psi  .
\end{align}
We notice that $e^{z H_{0}}$ is a product operator between $A$ and $B$. 
First, we prove that for small $|z|$ the operator $\Psi$ can be approximated by another operator $\tilde{\Psi}_\epsilon$ with a small Schmidt rank: 

\begin{prop} \label{Prop:short_time_approx}
Let $z$ be chosen such that 
\begin{align}
\label{basic_cond_for_z}
|z| \le  \min \br{ \frac{1}{4\mQ} ,\frac{1}{4 e C_0 \tilde{g}}} =: \mQ_0^{-1} . 
\end{align}
Then, there exists an operator $\tilde{\Psi}_\epsilon$ satisfying 
\begin{align}
\norm{ \Psi  - \tilde{\Psi}_\epsilon} \le \epsilon,
\end{align}
and 
\begin{align}
{\rm SR} \br{ \tilde{\Psi}_\epsilon } \le \br{4/\epsilon}^{6+4/\kappa + \log_2(D_0)} .
\end{align}
\end{prop}

We defer the proof of the proposition to Sec.~\ref{Sec:Proof of Proposition_Prop_short_time_approx} below.
Using this proposition, we now give a proof of the main statement, as follows.

We first consider the real-time evolution $z=-it$ ($t\in \mathbb{R}$) and decompose the total time $t$ into $m t_0$ with $t_0$ satisfying the condition~\eqref{basic_cond_for_z}. 
We generally obtain
\begin{align}
e^{- i  m t_0 (H_{0}+ V_{AB})}  - \br{ e^{-it_0 H_0} \tilde{\Psi}_\epsilon}^m
= \sum_{j=0}^{m-1} e^{- i  (m-j-1) t_0 (H_{0}+ V_{AB})}  \br{e^{-i t_0 (H_{0}+ V_{AB})}  - e^{-it_0 H_0} \tilde{\Psi}_\epsilon} \br{ e^{-it_0 H_0} \tilde{\Psi}_\epsilon}^j ,
\end{align}
which gives 
\begin{align}
\label{general_formula_norm}
&\norm{e^{- i  m t_0 (H_{0}+ V_{AB})}  - \br{ e^{-it_0 H_0} \tilde{\Psi}_\epsilon}^m} 
\le \sum_{j=0}^{m-1} \norm{ e^{-i t_0 (H_{0}+ V_{AB})}  - e^{-it_0 H_0} \tilde{\Psi}_\epsilon} \cdot \norm{ e^{-it_0 H_0} \tilde{\Psi}_\epsilon}^j .
\end{align}

Using Proposition~\ref{Prop:short_time_approx}, we obtain 
\begin{align}
\label{short_tine_eror_estimation}
\norm{ e^{-i t_0 (H_{0}+ V_{AB})}  - e^{-it_0 H_0} \tilde{\Psi}_\epsilon}
\le   \norm{\Psi - \tilde{\Psi}_\epsilon}\le \epsilon,
\end{align}
where we use $\norm{e^{-it_0 H_0}}=1$ and $e^{-i t_0 (H_{0}+ V_{AB})} =e^{-it_0 H_{0} } \Psi$.
This also yields 
\begin{align}
\label{short_tine_norm_estimation}
\norm{e^{-it_0 H_0} \tilde{\Psi}_\epsilon} =  \norm{e^{-i t_0 (H_{0}+ V_{AB})} -e^{-i t_0 (H_{0}+ V_{AB})}  + e^{-it_0 H_0} \tilde{\Psi}_\epsilon} \le  1+ \epsilon .
\end{align}
By applying the inequalities~\eqref{short_tine_eror_estimation} and \eqref{short_tine_norm_estimation} to \eqref{general_formula_norm}, we get 
\begin{align}
\label{general_formula_norm/2}
&\norm{e^{- i  m t_0 (H_{0}+ V_{AB})}  - \br{ e^{-it_0 H_0} \tilde{\Psi}_\epsilon}^m} 
\le \sum_{j=0}^{m-1} \epsilon (1+\epsilon)^j \le m e^{m\epsilon} \epsilon .
\end{align}

We now choose 
\begin{align}
t_0=\frac{t}{\ceil{t\mQ_0}} \ (\le \mQ_0^{-1}) ,\quad m=\frac{t}{t_0}= \ceil{t\mQ_0},\quad \epsilon= \frac{\epsilon_0}{2m}, \quad \tilde{U}_t :=  \br{ e^{-it_0 H_0} \tilde{\Psi}_\epsilon}^m .
\end{align}
We then obtain $\norm{ e^{- iHt} -\tilde{U}_t } \le \epsilon_0/2 e^{\epsilon_0/2} \le \epsilon_0$ from $\epsilon_0\le 1$ and 
\begin{align}
{\rm SR}\br{ \tilde{U}_t}  \le \brr{ {\rm SR}\br{\tilde{\Psi}_\epsilon}}^m 
&\le\br{\frac{8\ceil{t\mQ_0}}{\epsilon_0}}^{\brr{6+4/\kappa + \log_2(D_0)}\ceil{t\mQ_0}}.
\end{align}
This establishes the first main inequality in \eqref{thm:Schmidt_rank_truncation/main}.

Next, we consider the imaginary-time evolution $z=\beta$ with the decomposition $z=2m \beta_0$, where $\beta_0$ satisfies the condition~\eqref{basic_cond_for_z}. 
Using $ \tilde{\Psi}_\epsilon$, we examine $ e^{\beta_0 H_0} \tilde{\Psi}_\epsilon $, which approximates $e^{\beta_0(H_{0}+ V_{AB})}$ by
\begin{align}
\norm{ e^{\beta_0(H_{0}+ V_{AB})} - e^{\beta_0 H_0} \tilde{\Psi}_\epsilon}_p 
&= \norm{ e^{\beta_0H_0} \br{\Psi - \tilde{\Psi}_\epsilon}}_p \notag \\
&\le\norm{ e^{\beta_0H_0}}_p \norm{\Psi - \tilde{\Psi}_\epsilon}_\infty \le \norm{ e^{\beta_0H_0}}_p  \epsilon  ,
\end{align}
where we use the H\"older inequality to obtain $\norm{O_1O_2}_p \le \norm{O_1}_p \norm{O_2}_\infty$.
Furthermore, from Ref.~\cite[Lemma 7 therein]{kimura2024clustering}, we have
\begin{align}
\norm{ e^{O_1} - e^{O_1-O_2}}_p \le e^{\norm{O_2}} \norm{O_2} \norm{ e^{O_1}}_p ,
\end{align}
which gives 
\begin{align}
\norm{ e^{\beta_0(H_{0}+ V_{AB})} - e^{\beta_0 H_0}}_p \le  \norm{ e^{\beta_0(H_{0}+ V_{AB})}}_p  e^{\norm{\beta_0 V_{AB}}} \norm{\beta_0 V_{AB}}
\le  \norm{ e^{\beta_0(H_{0}+ V_{AB})}}_p  e^{\beta_0 \tilde{g}} \beta_0\tilde{g}. 
\end{align}
From $\norm{ e^{\beta_0(H_{0}+ V_{AB})} - e^{\beta_0 H_0}}_p \ge \norm{e^{\beta_0 H_0}}_p - \norm{ e^{\beta_0(H_{0}+ V_{AB})} }_p$, we obtain 
\begin{align}
\norm{ e^{\beta_0 H_0}}_p \le  \norm{ e^{\beta_0(H_{0}+ V_{AB})}}_p  \br{1+e^{\beta_0 \tilde{g}} \beta_0 \tilde{g}} \le   e^{2 \beta_0 \tilde{g}} \norm{ e^{\beta_0(H_{0}+ V_{AB})}}_p . 
\end{align}
In total, we derive the following error bound:
\begin{align}
\norm{ e^{\beta_0(H_{0}+ V_{AB})} - e^{\beta_0 H_0} \tilde{\Psi}_\epsilon}_p 
\le e^{2 \beta_0 \tilde{g}}  \epsilon \norm{ e^{\beta_0(H_{0}+ V_{AB})}}_p   .
\end{align}

We now utilize Ref.~\cite[Lemma~12 therein]{PhysRevX.11.011047}, which gives 
\begin{align}
\norm{ e^{\beta (H_{0}+ V_{AB})}  - \br{e^{\beta_0 H_0} \tilde{\Psi}^\dagger_\epsilon  \tilde{\Psi}_\epsilon e^{\beta_0 H_0}}^m }_p 
\le 3\delta_\epsilon m e^{3\delta_\epsilon m} \norm{ e^{\beta (H_{0}+ V_{AB})}}_p   ,
\end{align}
where we use $e^{2m\beta_0(H_{0}+ V_{AB})}=e^{\beta (H_{0}+ V_{AB})} $ and set $\delta_\epsilon=  e^{2 \beta_0 \tilde{g}}  \epsilon$.
We then choose 
\begin{align}
\beta_0=\frac{\beta}{2\ceil{\beta \mQ_0}} \ (\le \mQ_0^{-1}) ,\quad 2m=\frac{\beta}{ \beta_0}= 2\ceil{\beta\mQ_0},\quad \epsilon= \frac{e^{-2 \beta_0 \tilde{g}} \epsilon_0}{6m}, \quad \tilde{\rho}_\beta := \br{e^{\beta_0 H_0} \tilde{\Psi}^\dagger_\epsilon  \tilde{\Psi}_\epsilon e^{\beta_0 H_0}}^m.
\end{align}
Here, the condition $\beta_0 \le \mQ_0^{-1} \le 1/(4eC_0 \tilde{g})$ implies $e^{2 \beta_0 \tilde{g}} \le e^{1/(2eC_0)} \le 2$ because of $C_0 \ge 1$. 
The above choices give $\norm{ e^{\beta H} - \tilde{\rho}_\beta}_p \le \epsilon_0 \norm{ e^{\beta H}}_p$ for any $p \in\mathbb{N}$, and the Schmidt rank is given by 
\begin{align}
{\rm SR}( \tilde{\rho}_\beta ) \le  [ {\rm SR}( \tilde{\Psi}_\epsilon )]^{2m} \le \br{ \frac{48 \ceil{\beta\mQ_0}}{ \epsilon_0}}^{2\brr{6+4/\kappa + \log_2(D_0)} \ceil{\beta\mQ_0}} .
\end{align}
We thus prove the second main inequality in \eqref{thm:Schmidt_rank_truncation/main}.

This completes the proof of Theorem~\ref{thm:Schmidt_rank_truncation}. $\square$

\subsubsection{Proof of Proposition~\ref{Prop:short_time_approx}}\label{Sec:Proof of Proposition_Prop_short_time_approx}

We introduce the decomposition of $V_{AB}$ as follows:
 \begin{align}
V_{AB} = \sum_{m=0}^\infty V_{r_m} ,
\end{align}
with $V_{r_m}$ defined by
 \begin{align}
 \label{Def_V_ar_m}
V_{r_m} := \sum_{j \in [r_m,r_{m+1})} V_j, \quad r_m=4^{m/\kappa}  .
\end{align}
Here, we introudce the notation $\overline{V_{r_m}}$ as 
 \begin{align}
\overline{V_{r_m}}= \sum_{j \in [r_m,r_{m+1})} \norm{V_j} , 
\end{align}
and from the assumption~\eqref{primary_assumption_V_AB}, we obtain
 \begin{align}
\norm{V_{r_m}} \le  \overline{V_{r_m}} \le \sum_{j\ge 4^{m/\kappa}}^\infty \norm{V_j} \le  C_0 \tilde{g} 4^{-m} . 
\end{align}

We also define $V_{r_m}(q,x)$ as follows:
 \begin{align}
V_{r_m}(x,q) = \frac{(-xz)^q}{q!} \ad_{H_0}^q (V_{r_m}),
\end{align}
which gives 
 \begin{align}
V_{r_m}(x) = \sum_{q=0}^\infty \frac{(-xz)^q}{q!} \ad_{H_0}^q (V_{r_m}) =  \sum_{q=0}^\infty V_{r_m}(x,q) . 
\end{align}
Using the parameter $\mQ$ in Eq.~\eqref{introduce_parameter_mQ}, we obtain 
 \begin{align}
\norm{V_{r_m}(x,q)} \le \frac{|z|^q}{q!} \sum_{j \in [r_m,r_{m+1})} \norm{ \ad_{H_0}^q (V_j)} \le (|z| \mQ)^q  \sum_{j \in [r_m,r_{m+1})} \norm{V_j} \le   C_0 \tilde{g}(|z| \mQ)^q  4^{-m} ,
\end{align}
where we use $x\le 1$. 

By using the above notations, we first rewrite $\Psi$ in the following form:
\begin{align}
\Psi   &:= \sum_{s = 0}^{\infty}  z^s \sum_{m_1,m_2,\ldots,m_s=0}^\infty \sum_{q_1,q_2,\ldots,q_s=0}^\infty \int_{0}^{1} dx_{1} \int_{0}^{x_{1}} dx_{2} \cdots \int_{0}^{x_{s-1}} dx_{s}  
    V_{r_{m_1}}(x_1,q_1) V_{r_{m_2}}(x_2,q_2)  \cdots V_{r_{m_s}}(x_s,q_s)  . \label{supp_interaction_pic_general_rre}
\end{align}
Then, we approximate $\Psi$ by $\Psi_{Q,M,s_0}$:
\begin{align}
\Psi_{Q,M,s_0}  := \sum_{s = 0}^{s_0}  z^s \sum_{m_1+m_2+\cdots+m_s\le M}  \sum_{q_1+q_2+\cdots+q_s \le Q} \int_{0}^{1} dx_{1}  \cdots \int_{0}^{x_{s-1}} dx_{s}  
V_{r_{m_1}}(x_1,q_1)  \cdots V_{r_{m_s}}(x_s,q_s) .
\end{align}

By using the inequality as 
\begin{align}
\norm{ \int_{0}^{1} dx_{1}  \cdots \int_{0}^{x_{s-1}} dx_{s}  V_{r_{m_1}}(x_1,q_1) \cdots V_{r_{m_s}}(x_s,q_s) } \le  \frac{1}{s!} \prod_{j=1}^s  C_0 \tilde{g}(|z| \mQ)^{q_j}  4^{-m_j}   ,
\end{align}
we obtain
\begin{align}
\label{Sum_1st_2nd_3rd}
&\norm{\Psi- \Psi_{Q,M,s_0}}  \notag \\
&\le \br{ \sum_{s >s_0} \sum_{m_1,m_2,\ldots,m_s=0}^\infty \sum_{q_1,q_2,\ldots ,q_s =0}^\infty +\sum_{s = 0}^{\infty} \sum_{m_1+m_2+\cdots+m_s>M} \sum_{q_1,q_2,\ldots ,q_s =0}^\infty +  \sum_{s = 0}^{\infty}  \sum_{m_1,m_2,\ldots,m_s=0}^\infty \sum_{q_1+q_2+\cdots+q_s > Q} } \notag \\
&\quad \quad  \frac{|z|^s}{s!}  \prod_{j=1}^s  C_0 \tilde{g}(|z| \mQ)^{q_j}  4^{-m_j}  .
\end{align}

For the first summation in~\eqref{Sum_1st_2nd_3rd}, we can derive 
\begin{align}
\sum_{m_1,m_2,\ldots,m_s=0}^\infty \sum_{q_1,q_2,\ldots ,q_s =0}^\infty  \prod_{j=1}^s  C_0 \tilde{g}(|z| \mQ)^{q_j}  4^{-m_j} 
 &=( C_0 \tilde{g})^s \br{\sum_{m=0}^\infty 4^{-m} }^s \br{\sum_{q=0}^\infty (|z| \mQ)^{q} }^s  \notag \\
 &=\br{ \frac{16C_0 \tilde{g}}{9}}^s \le (2C_0 \tilde{g})^s    ,
\end{align}
where we use the condition $|z| \le 1/(4\mQ)$. 
Hence, we derive  
\begin{align}
 &\sum_{s >s_0} \sum_{m_1,m_2,\ldots,m_s=0}^\infty \sum_{q_1,q_2,\ldots ,q_s =0}^\infty  \frac{z^s}{s!}  \prod_{j=1}^s  C_0 \tilde{g}(|z| \mQ)^{q_j}  4^{-m_j}\notag \\
& \le \sum_{s >s_0} \frac{|z|^s}{s!} (2C_0 \tilde{g})^s 
\le\br{ \frac{ 2eC_0 \tilde{g} |z| }{s_0+1}  }^{s_0+1}e^{2C_0 \tilde{g} |z|} \le 2^{-s_0-1} e^{1/(2e)},
\end{align}
where we use $\sum_{s>s_0} x^s/s! \le e^{x} x^{s_0}/s_0!$ in the second inequality and use $|z| \le 1/(4eC_0 \tilde{g})$.

For the second summation in~\eqref{Sum_1st_2nd_3rd}, we obtain 
\begin{align}
\sum_{m_1+m_2+\cdots+m_s>M} \sum_{q_1,q_2,\ldots ,q_s =0}^\infty  \prod_{j=1}^s  C_0 \tilde{g}(|z| \mQ)^{q_j}  4^{-m_j} 
 &=\br{ \frac{4C_0 \tilde{g}}{3}}^s \sum_{\bar{m}=M+1}^\infty \sum_{m_1+m_2+\cdots+m_s=\bar{m}}4^{-\bar{m}} \notag \\
 &= \br{ \frac{4C_0 \tilde{g}}{3}}^s \sum_{\bar{m}=M+1}^\infty \binom{s+\bar{m}-1}{\bar{m}}4^{-\bar{m}} \notag \\
 &\le \br{ \frac{8C_0 \tilde{g}}{3}}^s \sum_{\bar{m}=M+1}^\infty 2^{-\bar{m}-1} 
 \le 2^{-M-1}(3C_0 \tilde{g})^s ,
 \label{Sum_ma_1_m2_M}
\end{align}
where the number of combinations such that $m_1+m_2+\cdots+m_s =\bar{m}$ is given by $\multiset{s}{\bar{m}} = \binom{s+\bar{m}-1}{\bar{m}}$.
We thus derive
\begin{align}
 &\sum_{s =0}^\infty \sum_{m_1,m_2,\ldots,m_s=0}^\infty \sum_{q_1,q_2,\ldots ,q_s =0}^\infty  \frac{z^s}{s!}  \prod_{j=1}^s  C_0 \tilde{g}(|z| \mQ)^{q_j}  4^{-m_j}
  \le 2^{-M-1} e^{3C_0 \tilde{g} |z|} \le  2^{-M-1} e^{3/(4e)} ,
\end{align}
where we use $|z| \le 1/(4eC_0 \tilde{g})$

Finally, for the third summation in~\eqref{Sum_1st_2nd_3rd}, we deduce
\begin{align}
\sum_{m_1,m_2,\ldots,m_s=0}^\infty \sum_{q_1+q_2+\cdots+q_s > Q}   \prod_{j=1}^s  C_0 \tilde{g}(|z| \mQ)^{q_j}  4^{-m_j} 
 &\le \br{ \frac{4C_0 \tilde{g}}{3}}^s \sum_{\bar{q}=Q+1}^\infty \sum_{q_1+q_2+\cdots+q_s=\bar{q}}4^{-\bar{q}}   \notag \\
 &\le 2^{-Q-1}(3C_0 \tilde{g})^s ,
\end{align}
where we use the same calculations as in~\eqref{Sum_ma_1_m2_M}. 
This yields 
\begin{align}
\sum_{s =0}^\infty \sum_{m_1,m_2,\ldots,m_s=0}^\infty \sum_{q_1+q_2+\cdots+q_s > Q}   \prod_{j=1}^s  C_0 \tilde{g}(|z| \mQ)^{q_j}  4^{-m_j} 
\le 2^{-Q-1} e^{3/(4e)},
\end{align}
In total, we obtain the approximation bound of 
\begin{align}
\label{approx_Psi_Q_M_S_0}
&\norm{\Psi- \Psi_{Q,M,s_0}} \le 2^{-s_0-1} e^{1/(2e)}  + \br{ 2^{-M-1} +2^{-Q-1}} e^{3/(4e)} .
\end{align}

We next estimate the Schmidt rank of $\Psi_{Q,M,s_0}$, which is upper-bounded by
\begin{align}
\label{SR_Psi_Q_M_s_0}
{\rm SR} \br{ \Psi_{Q,M,s_0} } \le \sum_{s = 0}^{s_0}  \sum_{m_1+m_2+\cdots+m_s\le M}  \sum_{q_1+q_2+\cdots+q_s \le Q} 
 \prod_{j=1}^s {\rm SR} \brr{ \ad_{H_0}^{q_j} (V_{r_{m_j}})}.
\end{align}
By combining the condition~\eqref{Decomposition_V_A_B_assump} and Eq.~\eqref{Def_V_ar_m}, we obtain 
\begin{align}
 {\rm SR} \br{V_{r_m}} \le  r_{m+1} D_0= 4^{(m+1)/\kappa}  D_0 . 
\end{align}
 By applying the above inequality to 
\begin{align}
 \ad_{H_0}^{q} (V_{r_{m}}) = \br{ \ad_{H_A} + \ad_{H_B}}^{q} (V_{r_{m}}) = 
 \sum_{j=0}^{q} \binom{q}{j}  \ad_{H_A}^{j} \otimes  \ad_{H_B}^{q-j}  (V_{r_m}) , 
\end{align}
we have 
\begin{align}
{\rm SR} \brr{ \ad_{H_0}^{q_j} (V_{r_{m_j}})} \le  (q_j+1) 4^{(m_j+1)/\kappa} D_0.
\end{align}
Therefore, we obtain
\begin{align}
\label{SR_Psi_Q_M_s_0_pre}
\prod_{j=1}^s {\rm SR} \brr{ \ad_{H_0}^{q_j} (V_{r_{m_j}})} \le 4^{(M+s)/\kappa} \br{\frac{Q+s}{s} D_0}^s ,
\end{align}
where we use the inequality of arithmetic and geometric means.

By applying the inequality~\eqref{SR_Psi_Q_M_s_0_pre} to \eqref{SR_Psi_Q_M_s_0}, we obtain the upper bound on the Schmidt rank as follows:
\begin{align}
{\rm SR} \br{ \Psi_{Q,M,s_0} } 
&\le \sum_{s = 0}^{s_0} \binom{M+s}{s}\binom{Q+s}{s}  4^{(M+s)/\kappa} \br{\frac{Q+s}{s}D_0}^s \notag \\
&\le 2^{M+s_0} \cdot 2^{Q+s_0} \cdot  4^{(M+s_0)/\kappa}D_0^{s_0}  e^{Q}   \notag \\ 
&\le 2^{(1+2/\kappa) M + 3Q + \brr{2+2/\kappa + \log_2(D_0)} s_0 },
\end{align}
where we use 
\begin{align}
 \sum_{m_1+m_2+\cdots+m_s\le M} 1 = \sum_{\bar{m}=0}^M  \sum_{m_1+m_2+\cdots+m_s= \bar{m}} 1 =  \sum_{\bar{m}=0}^M \binom{s+\bar{m}-1}{\bar{m}}  = \binom{s+M}{s}.
\end{align}
Note that the last equation follows from the hockey-stick identity. 

By choosing 
\begin{align}
s_0= M=Q = \ceil{\log_2(2/\epsilon )} \le\log_2(4/\epsilon ) ,
\end{align}
we reduce the error bound~\eqref{approx_Psi_Q_M_S_0} to
\begin{align}
&\norm{\Psi- \Psi_{Q,M,s_0}} \le 2^{-s_0-1} e^{1/(2e)}  + \br{ 2^{-M-1} +2^{-Q-1}} e^{3/(4e)} \le 2^{-s_0+1} \le \epsilon .
\end{align}
Moreover, we have 
\begin{align}
{\rm SR} \br{ \Psi_{Q,M,s_0} } 
\le 2^{ \brr{ 6+4/\kappa+ \log_2(D_0)} s_0} =  \br{4/\epsilon}^{6+4/\kappa+ \log_2(D_0)}  .
\end{align}
By choosing $\Psi_{Q,M,s_0} $ as $\tilde{\Psi}_\epsilon$, we prove the main statement.
This completes the proof. $\square$

\section{Generalized entanglement area law} \label{Sec:Generalized area law}

In a conventional notation, we consider the Hamiltonian such that 
\begin{align}
\label{conventional_notation}
H = H_{A_1A_0} + V_{A_0B_0} + H_{B_0B_1} , \quad \norm{V_{A_0B_0}} =1 , 
\end{align} 
where the subsets $A$ and $B$ are given by $A=A_1\sqcup A_0$ and $B=B_0\sqcup B_1$, respectively.
Assuming a constant spectral gap $\Delta={\rm const.}$, we argue that the following upper bound for the entanglement entropy holds:
\begin{align}
E_{\alpha=1}(\Omega) \le C \log[\min(\mD_{A_0},\mD_{B_0})] ,
\end{align} 
where the gap $\Delta$ is defined as the energy difference between the ground energy $E_0$ and the first excited energy $E_1$, i.e., $\Delta:=E_1-E_0$.
Generally, it has been known that this generalized area law does not hold~\cite{aharonov2014local}.

The motivation in this section is to figure out the rigorous condition such that the generalized area law is recovered (Fig.~\ref{fig:area_law}).
Here, we impose the following additional assumptions, which are typically considered for the area-law proof under the adiabatic path~\cite{PhysRevB.72.045141,michalakis2012,PhysRevLett.111.170501,PhysRevLett.113.197204}.
We restrict ourselves to the following boundary-adiabatic path:
\begin{assump}[Boundary-adiabatic path] \label{assump:Boundary-adiabatic path}
We consider a family of Hamiltonians parametrized by $\nu$ through the boundary interaction $V_{AB}^{(\nu)}$:
 \begin{align}
 \label{H^nu_decomp_def}
H^{(\nu)}= H_{A} + V_{AB}^{(\nu)}+ H_{B} ,  \quad \bar{\mJ}\brr{V_{AB}^{(\nu)}} \le \tilde{g} , 
\end{align} 
where we assume that the ground state $\ket{\Omega^{(\nu)}}$ of $H^{(\nu)}$ is non-degenerate for all $\nu \in [0,1]$.  
We say that there exists a boundary adiabatic path from $\nu=0$ to $\nu=1$ if 
 \begin{align}
 \label{gap_condition_Delta_0}
\Delta^{(\nu)}\ge \Delta ,
\end{align} 
where $\Delta^{(\nu)}$ denotes the spectral gap of $H^{(\nu)}$. 

Finally, we introduce a parameter $\tilde{c}_0$ such that 
\begin{align}
\max\br{ \norm{\frac{dV_{AB}^{(\nu)}}{d\nu}} , \norm{\frac{d^2V_{AB}^{(\nu)}}{d^2\nu}}} \le \tilde{c}_0\tilde{g} .
\end{align} 
\end{assump}

{\bf Remark.} 
In contrast to the conventional notation, the interaction $V_{AB}^{(\nu)}$ is allowed to act more generally across the boundary between $A$ and $B$, rather than being restricted to $A_0$ and $B_0$. 
The restriction on the interaction region imposed by $A_0$ and $B_0$ is effectively absorbed into the condition on the SE strength, i.e., $\bar{\mJ}\brr{V_{AB}^{(\nu)}} \le \tilde{g}$.  

Additionally, in a more general setting, one may consider an arbitrary parameterization of the Hamiltonian, for example, $H_{A}^{(\nu)} + V_{AB}^{(\nu)} + H_{B}^{(\nu)}$. In such a general framework, there may exist adiabatic paths that connect the initial ground state to an anomalous, highly entangled ground state, as demonstrated in Ref.~\cite{aharonov2014local}. Although we have not yet identified an explicit counterexample, we conjecture that the standard adiabatic path condition by itself is insufficient to guarantee the generalized area law.

\begin{figure}[ttt]
  \centering
  \includegraphics[width=0.8\textwidth]{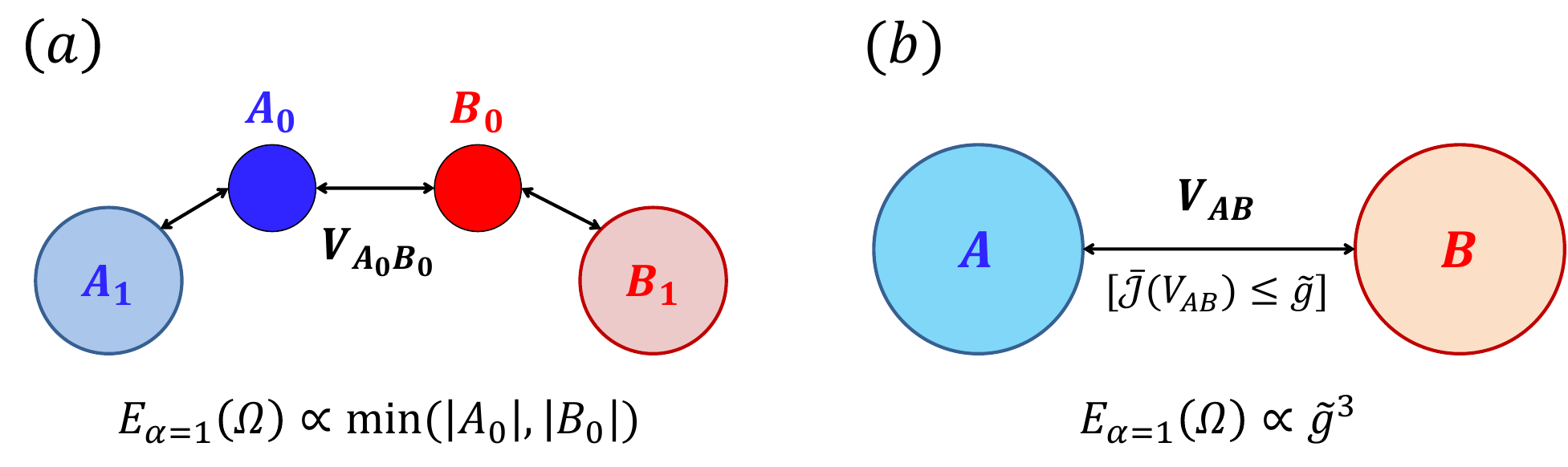}
\caption{Schematic picture of the generalized area law. 
(a) The generalized area law states that the entanglement entropy is bounded in terms of the dimension of the interacting region. 
However, this bound is known to be violated for general ground states. 
(b) By assuming an adiabatic path for the boundary interaction, we establish the generalized area law under a condition formulated with the SE strength $\tilde{g}$. 
Roughly speaking, $\tilde{g}$ corresponds to the effective size of the boundary. 
For example, in systems defined on graphs, it becomes more evident that $\tilde{g}$ is proportional to the boundary size. 
If $A_0$ and $B_0$ consist of $m$ qubits and the interaction $V_{A_0B_0}$ is described by $\orderof{m}$ two-qubit interactions, then one obtains 
$\tilde{g} \propto m = \log(\mD_{A_0})$ from the inequality~\eqref{Trivial_Ineq_interaction_strength}.}
  \label{fig:area_law}
\end{figure}

Under the above notations, we prove the following theorem:
\begin{theorem} \label{thm:generalized_area_law}
Let us denote the initial-state entanglement by the $\infty$-R\'enyi entanglement: 
\begin{align}
\label{E_alpha_0_nu_0}
E_{\infty}(\Omega^{(0)}) = S_0 . 
\end{align} 
Under the assumption~\ref{assump:Boundary-adiabatic path}, there exists a state $\ket{\psi_D}$ with Schmidt rank $D$ such that 
\begin{align}
\label{K_beta_psi_1,t_3_main_ineq_1}
\norm{\ket{\psi_D} -  \ket{\Omega^{(1)}} } \le C_{\tilde{g},\Delta,S_0} D^{-\kappa_\Delta} ,
\end{align}
where we define
\begin{align}
\label{Def_Delta/_eta0}
&\kappa_\Delta:=\frac{\Delta}{2\Delta+4\tilde{g}},  \\
&C_{\tilde{g},\Delta,S_0}:=\exp \brr{ \frac{S_0+3}{4\kappa_\Delta} + \log(12) + \frac{3\tilde{c}_0\tilde{g}^3(2 \Delta/\tilde{g}+ 7\tilde{c}_0)}{\Delta^3} }. 
\end{align}
Moreover, the entanglement entropy $E_{\alpha=1}(\Omega^{(1)})$ for the ground state $\ket{\Omega^{(1)}}$ is upper-bounded by
\begin{align}
\label{K_beta_psi_1,t_3_main_ineq_2}
E_{\alpha=1}(\Omega^{(1)}) \le c_{\kappa_\Delta,1} \log(C_{\tilde{g},\Delta,S_0}) + c_{\kappa_\Delta,2} ,
\end{align}
where 
\begin{align}
c_{\kappa_\Delta,1} := \frac{2-2^{-2\kappa_\Delta}}{\kappa_\Delta(1-2^{-2\kappa_\Delta})}, \quad c_{\kappa_\Delta,2} :=  \frac{(6 + 2\kappa_\Delta)\log(2)}{(1- 2^{-2\kappa_\Delta})^2}.
\end{align}
\end{theorem}

{\bf Remark.} Since $E_{\infty}(\Omega^{(0)}) \le E_{\alpha}(\Omega^{(0)})$, the initial condition involving the $\infty$-R\'enyi entanglement is the least restrictive, enabling the broadest applicability.

In our setup, the boundary size is roughly estimated as $\tilde{g}/\Delta$ since $\tilde{g}$ characterizes the dynamical entanglement rate from the spectral SIE. 
Indeed, in a graph system with finite-range interactions, $\tilde{g}$ is indeed proportional to the boundary size between $A$ and $B$. 
From this perspective, our result qualitatively yields
\begin{align}
E_{\alpha=1}(\Omega^{(1)}) \propto (\textrm{Boundary Size})^3  .
\end{align}
It is an important open question whether the upper bound~\eqref{K_beta_psi_1,t_3_main_ineq_2} can be improved to a bound of the form
$
E_{\alpha=1}(\Omega^{(1)}) = \orderof{\tilde{g}/\Delta}. 
$

\subsection{Approximate-Ground-Space-Projection (AGSP)}
In the proof of Theorem~\ref{thm:generalized_area_law}, we combine the quantum adiabatic theorem~\cite{10.1063/1.2798382} with the so-called {\it approximate ground-state projection (AGSP)} approach~\cite{PhysRevB.85.195145,arad2013area,Kuwahara2020arealaw,10.1145/3519935.3519962}.

We consider Schmidt-rank truncation for arbitrary gapped ground states. 
To this end, we introduce an AGSP operator $K$, defined by the conditions
\begin{align}
\label{pro_AGSP_0}
\norm{(K-1)\ket{\Omega}} \approx 0 ,\quad \norm{K\ket{\Omega_\perp}} \approx 0 ,
\end{align}
for any state $\ket{\Omega_\perp}$ orthogonal to the ground state $\ket{\Omega}$. 
Thus, $K$ serves as an approximation to the ground-state projector, while allowing us to impose constraints on the Schmidt-rank structure.

In contrast to previous studies, the present use of the AGSP has the following distinctive feature. 
\begin{enumerate}
\item In the standard AGSP framework, the primary focus is to bound the Schmidt rank of the AGSP operator.
\item In our setting, however, the focus shifts to the {\it entanglement generation} induced by the AGSP operator. 
This distinction originates from the fact that there is a fundamental difference between (i) low-rank approximation of the operator and (ii) the entanglement generated by the action of the operator (see Proposition~\ref{No_go_theorem_approx}). 
\end{enumerate}

\begin{prop} \label{Prop:AGSP_construction_spectral_SIE}
For any Hamiltonian in the form of $H=H_A+H_B + V_{AB}$, 
there exists an AGSP operator $K_\beta$ with the following properties:
\begin{align}
\label{Prop:AGSP_construction_spectral_SIE_main1}
\norm{(K_{\beta}-1) \ket{\Omega}}\le e^{-\beta \Delta^2} ,\quad \norm{ K_\beta \ket{\Omega_\bot}}\le  2e^{-\beta \Delta^2} ,
\end{align}
and 
\begin{align}
\label{Prop:AGSP_construction_spectral_SIE_main2}
\bar{\mJ} (K_{\beta}) \le e^{2\beta \Delta \bar{\mJ}(V_{AB})},
\end{align} 
where $\beta$ is arbitrarily chosen, and $\Delta$ has been defined as the spectral gap.
\end{prop}

{\bf Remark.} 
For any product state $\ket{\phi} = \ket{\phi_A} \otimes \ket{\phi_B}$, we apply the upper bound~\eqref{Prop:AGSP_construction_spectral_SIE_main2} to the inequality~\eqref{corol:Interaction_strength_S_Coeff_main} in Corollary~\ref{corol:Interaction_strength_S_Coeff}.
Then, the quantum state $K_\beta \ket{\phi}$ is well-approximated by $\ket{\phi_D}$ with the Schmidt rank $D$ as follows:  
\begin{align}
\label{Prop:AGSP_construction_spectral_SIE_apply}
\norm{ K_\beta \ket{\phi} - \ket{\phi_D}  } \le \frac{\bar{\mJ}(K_{\beta})}{\sqrt{D}} \le  \frac{e^{2\beta \Delta \bar{\mJ}(V_{AB})}}{\sqrt{D}} .
\end{align}

Here, the AGSP $K_\beta$ itself may not be approximated by an operator with the Schmidt rank $D$ (see also Sec.~\ref{sec:Op_approx_Dynamical_ent}). 
However, in applying the AGSP formalism to the ground state, the Schmidt rank approximation with respect to $K_\beta \ket{\phi} $ is more crucial.
We use this proposition to achieve the ground state approximation in Theorem~\ref{thm:generalized_area_law}, as well as in Theorem~\ref{thnm:Gs_approx} below, concerning the polynomial complexity of the ground state in 1D long-range interacting systems.

\subsubsection{Proof of Proposition~\ref{Prop:AGSP_construction_spectral_SIE}}
In the proof, as a candidate for the AGSP operator $K_\beta$, we employ the following one:
\begin{align}
K_\beta = \frac{1}{\sqrt{4\pi \beta}} \int_{-t_c}^{t_c} e^{-t^2/(4\beta)} e^{-iHt} dt .
\end{align}
First, we demonstrate that the above operator $K_\beta$ indeed satisfies the properties in~\eqref{Prop:AGSP_construction_spectral_SIE_main1}.

To this end, we first obtain the following bounds on the norm difference between $K_\beta$ and $e^{-\beta H^2}$:
\begin{align}
\norm{ K_\beta - e^{-\beta H^2}}
& \le \frac{1}{\sqrt{\pi \beta}} \int_{t_c}^{\infty} e^{-t^2/(4\beta)} dt  \notag \\
&\le  \frac{1}{\sqrt{\pi \beta}} e^{-t_c^2/(4\beta)} \int_{t_c}^{\infty} e^{-(t^2-t_c^2)/(4\beta)} dt \notag \\
&\le  \frac{1}{\sqrt{\pi \beta}} e^{-t_c^2/(4\beta)} \int_{0}^{\infty} e^{-( \tilde{t}^2+ 2\tilde{t}  t_c)/(4\beta)} d\tilde{t} \le e^{-t_c^2/(4\beta)}  .
\end{align}
This upper bound immediately implies
\begin{align}
\norm{(K_{\beta}-1) \ket{\Omega}}\le e^{-t_c^2/(4\beta)}  ,
\end{align}
where we use $e^{-\beta H^2} \ket{\Omega}= \ket{\Omega}$. 
For an arbitrary quantum state $\ket{\Omega_\bot}$, we also obtain 
\begin{align}
\norm{ K_\beta \ket{\Omega_\bot}} \le \norm{ e^{-\beta H^2} \ket{\Omega_\bot}}+  \norm{ K_\beta - e^{-\beta H^2}} \le  e^{-\beta \Delta^2} + e^{-t_c^2/(4\beta)}.
\end{align}
By setting 
\begin{align}
\label{choice_t_cAGSP}
t_c= 2\beta \Delta , 
\end{align}
we obtain the first and second inequalities in~\eqref{Prop:AGSP_construction_spectral_SIE_main1}.  

Next, we estimate the SE strength $\bar{\mJ}(K_\beta)$ based on Definition~\ref{Def:Interaction_strength} to prove the inequality~\eqref{Prop:AGSP_construction_spectral_SIE_main2}.
We use Lemma~\ref{lemm:Sum_bar_J} and obtain the inequality 
\begin{align}
\bar{\mJ}(K_\beta) \le   \frac{1}{\sqrt{4\pi \beta}} \int_{-t_c}^{t_c} e^{-t^2/(4\beta)} \bar{\mJ}( e^{-iHt})  dt .
\end{align}
From Corollary~\ref{corol:Renyi_MJ}, we get $\bar{\mJ}(e^{-iHt}) \le e^{\abs{t} \bar{\mJ}(V_{AB})}$, and hence the desired inequality yields as follows: 
\begin{align}
\label{uppper_bound_mJ_K_beta}
\bar{\mJ}(K_\beta) 
&\le   \frac{1}{\sqrt{4\pi \beta}}\int_{-t_c}^{t_c}  e^{-t^2/(4\beta)+ \abs{t} \bar{\mJ}(V_{AB})}  dt  \notag \\
&\le \frac{e^{t_c \bar{\mJ}(V_{AB})}}{\sqrt{4\pi \beta}}\int_{-t_c}^{t_c}  e^{-t^2/(4\beta)}  dt  \le e^{2\beta \Delta \bar{\mJ}(V_{AB})}  ,
\end{align}
where we use Eq.~\eqref{choice_t_cAGSP} in the last inequality. 
This completes the proof. $\square$

\subsection{Proof of Theorem~\ref{thm:generalized_area_law}} 

We begin by considering the Schmidt-rank truncation problem for the initial ground state $\ket{\Omega^{(\nu)}}$ with $\nu=0$.
In general, any restrictions to R\'enyi entanglement with $\alpha\ge 1$ cannot ensure the efficiency guarantee of the truncation~\cite{PhysRevLett.100.030504,PRXQuantum.1.010304}.
However, by utilizing the fact that $\ket{\Omega^{(0)}}$ is the gapped ground state, we can prove the following proposition, which plays a key role in the proof: 
\begin{prop} \label{Prop:Renyi-infty_to_finite}
For any Hamiltonian in the form of~\eqref{H^nu_decomp_def}, we assume the gap condition~\eqref{gap_condition_Delta_0} and the restriction~\eqref{E_alpha_0_nu_0} to the $\infty$-R\'enyi entanglement. 
Then, there exists a quantum state $\ket{\psi_D}$ with the Schmidt rank $D$ such that  
\begin{align}
\label{Prop:Renyi-infty_to_finite/main}
\norm{\ket{\Omega^{(\nu)}}- \ket{\psi_D}} \le 4 e^{S_0/2} D^{-\kappa_\Delta} ,
\end{align} 
with
\begin{align}
\label{Def_Delta/_eta}
\kappa_\Delta := \frac{\Delta}{2\Delta+4\tilde{g}} . 
\end{align} 
We remind that $\tilde{g}$ is the upper bound for the SE strength for all $V_{AB}^{(\nu)}$ ($0\le \nu \le 1$) as in~\eqref{H^nu_decomp_def}. 
\end{prop}

\textit{Proof of Proposition~\ref{Prop:Renyi-infty_to_finite}.}
We first denote the Schmidt decomposition of $\ket{\Omega^{(\nu)}}$ by
\begin{align}
\label{Schmidt_Omega^nu}
\ket{\Omega^{(\nu)}}= \sum_s  \lambda_s \ket{\phi_{A,s}} \otimes  \ket{\phi_{B,s}} .
\end{align} 
Then, we have the following general equation: 
 \begin{align}
 \label{upper_bound_lambda_1}
\lambda_1 = e^{-S_0/2} . 
\end{align} 
where we use $-\log\br{ \lambda^2_1} = E_{\infty} (\Omega^{(\nu)}) =  S_0 $. 
We cannot get further information on $\lambda_s$ with $s\ge 2$ only from the information on the $\infty$-R\'enyi entanglement. 

To circumvent the difficulty stemming from this limitation, we recover the decay rate of $\{\lambda_s\}_s$ from the condition~\eqref{upper_bound_lambda_1} using the AGSP formalism.
To proceed along this line, we decompose $\ket{\phi_{A,1},\phi_{B,1}}$ in Eq.~\eqref{Schmidt_Omega^nu} into
 \begin{align}
\ket{\phi_{A,1},\phi_{B,1}} =:\ket{\phi_1}  = \lambda_1\ket{\Omega^{(\nu)}} + \sqrt{1-\lambda_1^2} \ket{\Omega_\bot^{(\nu)}} ,  
\end{align} 
where $\langle \Omega^{(\nu)} \ket{\Omega_\bot^{(\nu)}}=0 $. 
To recover $\ket{\Omega^{(\nu)}}$ from $\ket{\phi_{A,1},\phi_{B,1}}$, we utilize the AGSP $K_\beta$ constructed in Proposition~\ref{Prop:AGSP_construction_spectral_SIE}, which yields
\begin{align}
\label{K_beta_phi_1,nu0}
\norm{\lambda_1^{-1}K_{\beta}\ket{\phi_1}  -  \ket{\Omega^{(\nu)}} }
&\le  \norm{(K_{\beta}-1) \ket{\Omega^{(\nu)}}}+ \frac{\sqrt{1-\lambda_1^2}}{\lambda_1} \norm{ K_\beta \ket{\Omega_\bot^{(\nu)}}} \notag \\
&\le \frac{3}{\lambda_1} e^{-\beta \Delta^2} . 
\end{align}

On the other hand, from the inequality~\eqref{Prop:AGSP_construction_spectral_SIE_apply}, there exists an (unnormalized) quantum state $\ket{\phi_D}$ satisfying
\begin{align}
\label{K_beta_phi_1,nu0__/2}
\norm{K_{\beta}\ket{\phi_1}  -  \ket{\phi_D}}\le \frac{e^{2\beta \Delta \tilde{g}}}{D^{1/2}}  ,
\end{align}
where we use the parameter $\tilde{g}$, which upper-bounds the SE strength as in~\eqref{H^nu_decomp_def}.
Therefore, by letting $\ket{\psi_D}:=\lambda_1^{-1} \ket{\phi_D}$ and combining the inequalities~\eqref{K_beta_phi_1,nu0} and \eqref{K_beta_phi_1,nu0__/2}, we obtain 
\begin{align}
\label{K_beta_phi_1,nu0___D}
\norm{ \ket{\Omega^{(\nu)}} -  \ket{\phi_D}}
&\le \frac{3}{\lambda_1} e^{-\beta \Delta^2} + \frac{e^{2\beta \Delta \tilde{g}}}{\lambda_1 D^{1/2}}.
\end{align}

By choosing $\beta$ such that $e^{-\beta \Delta^2}= e^{2\beta \Delta \tilde{g}}D^{-1/2}$, or equivalently,
\begin{align}
\beta = \frac{1}{2\Delta (\Delta+2\tilde{g})} \log(D) ,
\end{align}
we have the main inequality as 
\begin{align}
\norm{ \ket{\Omega^{(\nu)}} -  \ket{\phi_D}} 
&\le \frac{4}{\lambda_1} D^{-\Delta / (2\Delta+4\tilde{g})}=4 e^{S_0/2} D^{-\kappa_\Delta} ,
\end{align}
where we use Eq.~\eqref{upper_bound_lambda_1} and the notation $\kappa_\Delta$ introduced in Eq.~\eqref{Def_Delta/_eta}.
This completes the proof. $\square$ 

{~}

\hrulefill{\bf [ End of Proof of Proposition~\ref{Prop:Renyi-infty_to_finite}]}

{~}

To analyze the target ground state $\ket{\Omega^{(\nu)}}$ with $\nu=1$, we next utilize the adiabatic time evolution using the parametrized Hamiltonian $H^{(\nu)}$ ($0\le \nu \le 1$) as follows:
\begin{align}
U_{0\to 1/\varepsilon} := \mathcal{T} e^{-i\int_0^{1/\varepsilon} H^{(\varepsilon x)} dx } ,
\end{align} 
where $\varepsilon$ is sufficiently small as will be chosen in Eq.~\eqref{choice_varepsilon_3}.  
We then approximate the target ground state $\ket{\Omega^{(1)}}$ by $U_{0\to 1/\varepsilon} \ket{\Omega^{(0)}}$, which we denote by $\ket{\psi_\varepsilon}$:
\begin{align}
\label{def:psi_varepsilon}
\ket{\psi_\varepsilon}:= U_{0\to 1/\varepsilon} \ket{\Omega^{(0)}} .
\end{align} 
From Ref.~\cite[]{10.1063/1.2798382}, we obtain the approximation error as 
\begin{align}
\label{adiabatic_error_up}
\norm{ \ket{\psi_\varepsilon}- \ket{\Omega^{(1)}} } \le \frac{\tilde{c}_0\tilde{g} \varepsilon}{\Delta^2}   \br{2 + \frac{7\tilde{c}_0\tilde{g}}{\Delta}  } .
\end{align}

Let us denote the Schmidt decomposition of $\ket{\Omega^{(0)}}$ by
\begin{align}
\ket{\Omega^{(0)}}= \sum_s  \lambda_s \ket{\phi_{A,s}} \otimes  \ket{\phi_{B,s}} .
\end{align} 
We then define $\ket{\psi_{\ell,t}}$ as follows:
\begin{align}
\label{def_psi_ell_t}
\ket{\psi_{\ell,t}} :=a^{-1}_\ell U_{0\to t}  \sum_{s=1}^\ell    \lambda_s \ket{\phi_{A,s}} \otimes  \ket{\phi_{B,s}} , \quad a_\ell^2 =  \sum_{s\le \ell}  \lambda_s^2,
\end{align} 
where the truncation number $\ell$ will be appropriately chosen afterward, and $a_\ell$ is a normalization factor. 
We obtain the overlap between $\ket{\psi_\varepsilon}$ and $\ket{\psi_{\ell,1/\varepsilon}}$ by 
\begin{align}
\label{over_lap_a_ell}
\langle \psi_{\ell,1/\varepsilon} \ket{\psi_\varepsilon}=\br{a^{-1}_\ell \sum_{s=1}^\ell    \lambda_s \bra{\phi_{A,s},\phi_{B,s}}}  \ket{\Omega^{(0)}}= \br{\sum_{s\le \ell}  \lambda_s^2}^{1/2} =a_\ell , 
\end{align} 
where we use the definition~\eqref{def:psi_varepsilon} for $\ket{\psi_\varepsilon}$. 

We next estimate the lower bound of $a_\ell$.
Applying the Eckart--Young theorem~\cite{Eckart1936} and employing the quantum state $\ket{\psi_D}$ described in Proposition~\ref{Prop:Renyi-infty_to_finite}, we obtain 
 \begin{align}
 \label{Eckart--Young_upper_bound}
\sum_{s> \ell}  \lambda_s^2  \le \norm{\ket{\Omega^{(\nu)}}- \ket{\psi_{D=\ell}}}^2 \le 16 e^{S_0} \ell^{-2\kappa_\Delta}  ,
\end{align} 
where we use the error bound~\eqref{Prop:Renyi-infty_to_finite/main} with $D= \ell$.
Hence, we obtain the lower bound of $a_\ell^2$ in the form of 
 \begin{align}
a_\ell^2 = 1- \sum_{s> \ell}  \lambda_s^2 \ge 1 - 16 e^{S_0} \ell^{-2\kappa_\Delta} .
\end{align} 
Therefore, to achieve $a_\ell^2 \ge 1/2$, we need to choose $\ell$ such that
 \begin{align}
 \label{Choice_of_ell}
\ell = \ceil{ \br{16e^{S_0}}^{1/(2\kappa_\Delta)}} \le e^{(S_0+3)/(2\kappa_\Delta)} .
\end{align}

We further estimate the overlap between $\ket{\psi_{\ell,1/\varepsilon}}$ and the target ground state $\ket{\Omega^{(1)}}$. 
We use the decomposition of 
\begin{align}
\langle \psi_{\ell,1/\varepsilon} \ket{\psi_\varepsilon} =  \bra{ \psi_{\ell,1/\varepsilon} }\br{\ket{\psi_\varepsilon}- \ket{\Omega^{(1)}}} + \langle \psi_{\ell,1/\varepsilon}\ket{\Omega^{(1)}} .
\end{align}
By combining the error bound~\eqref{adiabatic_error_up} and Eq.~\eqref{over_lap_a_ell} with $a_\ell^2 \ge 1/2$, we derive the inequality 
\begin{align}
\abs{ \langle \psi_{\ell,1/\varepsilon}\ket{\Omega^{(1)}} } \ge \frac{1}{\sqrt{2}}- \frac{\tilde{c}_0\tilde{g} \varepsilon}{\Delta^2}   \br{2 + \frac{7\tilde{c}_0\tilde{g}}{\Delta}  }.
\end{align}
We thus choose the parameter $\varepsilon$ as 
 \begin{align}
 \label{choice_varepsilon_3}
\varepsilon =\frac{\Delta^3}{\tilde{c}_0\tilde{g} (2 \Delta+ 7\tilde{c}_0\tilde{g})}\cdot \frac{1}{2\sqrt{2}} ,
\end{align}
which gives 
\begin{align}
\abs{ \langle \psi_{\ell,1/\varepsilon}\ket{\Omega^{(1)}} } \ge\frac{1}{2\sqrt{2}} \ge \frac{1}{3}. 
\end{align}

We then construct an approximation of the target ground state $\ket{\Omega^{(1)}}$ based on the state $\ket{\psi_{\ell,1/\epsilon}}$.
For this purpose, we first decompose 
 \begin{align}
\ket{\psi_{\ell,1/\epsilon}} = c_1\ket{\Omega^{(1)}} + \sqrt{1-c_1^2} \ket{\Omega_\bot^{(1)}} ,  \quad c_1 \ge \frac{1}{3} , 
\end{align} 
where $\langle \Omega^{(1)} \ket{\Omega_\bot^{(1)}}=0 $. 
To recover $\ket{\Omega^{(1)}}$ from $\ket{\psi_{\ell,t}}$, we utilize the same inequality as~\eqref{K_beta_phi_1,nu0} 
that is based on the AGSP $K_\beta$ in Proposition~\ref{Prop:AGSP_construction_spectral_SIE}:
\begin{align}
\label{K_beta_psi_1,t}
\norm{c_1^{-1}K_{\beta}\ket{\psi_{\ell,1/\epsilon}} -  \ket{\Omega^{(1)}} }
&\le \frac{3}{c_1} e^{-\beta \Delta^2} \le 9 e^{-\beta \Delta^2} .
\end{align}

We consider the approximation of $c_1^{-1}K_{\beta}\ket{\psi_{\ell,1/\epsilon}}$ using a quantum state with a small Schmidt rank.  
We denote the Schmidt decomposition of the quantum states $\ket{\psi_{\ell,1/\epsilon}}$ and $c_1^{-1} K_{\beta}\ket{\psi_{\ell,1/\epsilon}}$ by
\begin{align}
&\ket{\psi_{\ell,1/\epsilon}} = \sum_{s=1}^\infty \tilde{\lambda}_{s}  \ket{\tilde{\phi}_{A, s}} \otimes \ket{\tilde{\phi}_{B,s}} ,\notag \\ 
&c_1^{-1} K_{\beta}\ket{\psi_{\ell,1/\epsilon}} = \sum_{s=1}^\infty \tilde{\lambda}_{\beta,s}  \ket{\tilde{\phi}_{A,\beta, s}} \otimes \ket{\tilde{\phi}_{B,\beta,s}} ,
\end{align}
respectively. 
The R\'enyi entanglement $E_{\alpha=1/2}(\psi_{\ell,1/\varepsilon})$ is upper-bounded from the spectral SIE theorem~\eqref{d/dt_E_alpha_psi_t} as follows:
\begin{align}
E_{\alpha=1/2}(\psi_{\ell,1/\varepsilon}) 
&\le E_{\alpha=1/2}\br{\psi_{\ell,t=0}} + \int_0^{1/\varepsilon} \frac{d}{dt}  E_{\alpha=1/2}\br{\psi_{\ell,t}} dt \le \log(\ell) + \frac{2\tilde{g}}{\varepsilon} ,
\end{align} 
where $\ket{\psi_{\ell,t=0}}$ has the Schmidt rank $\ell$ [see Eq.~\eqref{def_psi_ell_t}] and hence $E_{\alpha=1/2}\br{\psi_{\ell,t=0}} \le \log(\ell)$ as a trivial bound.  
Hence, we obtain 
\begin{align}
\sum_{s=1}^\infty  \tilde{\lambda}_{\beta,s} \le  \sum_{s=1}^\infty \tilde{\lambda}_{s}  \bar{\mJ} \br{K_{\beta}/c_1}
&= \bar{\mJ} \br{K_{\beta}/c_1} e^{E_{\alpha=1/2}(\psi_{\ell,1/\varepsilon})/2}  \le 3\ell^{1/2} e^{2\beta \Delta \tilde{g} + \tilde{g}/\varepsilon}  ,
\end{align}
where we use the inequality~\eqref{lambda_s_sum_g_upp} and the upper bound~\eqref{Prop:AGSP_construction_spectral_SIE_main2} for the SE strength of $K_\beta$ in Proposition~\ref{Prop:AGSP_construction_spectral_SIE}.
We can thus truncate the Schmidt rank of $c_1^{-1}K_{\beta}\ket{\psi_{\ell,t}}$ up to $D$ and construct a quantum state $\ket{\psi_{\beta,D}}$, which satisfies 
\begin{align}
\norm{c_1^{-1}K_{\beta}\ket{\psi_{\ell,1/\epsilon}}- \ket{\psi_{\beta,D}}} \le \frac{3\ell^{1/2} e^{(2\beta \Delta  + 1/\varepsilon)\tilde{g}}}{D^{1/2}} ,
\end{align}
where we use an inequality similar to~\eqref{upper_bound_error_approx_D}.

By applying the above inequality to~\eqref{K_beta_psi_1,t}, we obtain
\begin{align}
\label{K_beta_psi_1,t_2}
\norm{\ket{\psi_{\beta,D}} -  \ket{\Omega^{(1)}} } \le9 e^{-\beta \Delta^2}+  \frac{3\ell^{1/2} e^{(2\beta \Delta  + 1/\varepsilon)\tilde{g}}}{D^{1/2}} 
\end{align}
By choosing the parameter $\beta$ such that $e^{-\beta \Delta^2}= e^{2\beta \Delta \tilde{g}}D^{-1/2}$, or equivalently, 
\begin{align}
\beta = \frac{1}{2\Delta (\Delta+2\tilde{g})} \log(D) ,
\end{align}
we have 
\begin{align}
\label{K_beta_psi_1,t_3}
\norm{\ket{\psi_{\beta,D}} -  \ket{\Omega^{(1)}} } \le  \br{9 + 3\ell^{1/2}e^{\tilde{g}/\varepsilon}} D^{-\Delta / (2\Delta+4\tilde{g})} 
\le 12 \ell^{1/2}e^{\tilde{g}/\varepsilon}D^{-\kappa_\Delta} , 
\end{align}
where we use the definition of $\kappa_\Delta$ in Eq.~\eqref{Def_Delta/_eta0}.  

Finally, we upper-bound the coefficient $12 \ell^{1/2}e^{\tilde{g}/\varepsilon}$ in Eq.~\eqref{K_beta_psi_1,t_3}.
As in Eqs.~\eqref{Choice_of_ell} and \eqref{choice_varepsilon_3}, $\ell$ and $\varepsilon$ have the following upper and lower bounds, respectively:
\begin{align}
\ell \le e^{(S_0+3)/(2\kappa_\Delta)}, \quad\varepsilon\ge \frac{\Delta^3}{3\tilde{c}_0\tilde{g} (2 \Delta+ 7\tilde{c}_0\tilde{g})}
\end{align}
and hence 
\begin{align}
 \log\br{12 \ell^{1/2}e^{\tilde{g}/\varepsilon}} \le \frac{S_0+3}{4\kappa_\Delta} + \log(12) + \frac{3\tilde{c}_0\tilde{g}^3(2 \Delta/\tilde{g}+ 7\tilde{c}_0)}{\Delta^3} =\log\br{ C_{\tilde{g},\Delta,S_0}}.
\end{align}
This reduces the inequality~\eqref{K_beta_psi_1,t_3} to the first main inequality~\eqref{K_beta_psi_1,t_3_main_ineq_1} by setting $\ket{\psi_D}$ to be $\ket{\psi_{\beta,D}}$.

To derive the second main inequality~\eqref{K_beta_psi_1,t_3_main_ineq_2}, we prove the following lemma:
\begin{lemma} \label{lemm:approx_Ent_entroy}
Let us consider a quantum state $\ket{\psi}$ that is approximated by another unnormalized quantum state $\ket{\psi_D}$ with the Schmidt rank $D$ as follows:
\begin{align}
\label{psi_psi_D_approx}
\norm{\ket{\psi} - \ket{\psi_D}} \le \frac{C}{D^{\kappa}}  \quad \forall D \in \mathbb{N},
\end{align}
where $C$ and $\kappa$ are positive constants. 
Then, the entanglement entropy is upper-bounded by
\begin{align}
E_{\alpha=1}(\psi)\le c_{\kappa,1} \log(C) + c_{\kappa,2} ,
\end{align}
where the constants $c_{\kappa,1}$ and $c_{\kappa,2}$ are defined by 
\begin{align}
c_{\kappa,1} := \frac{2-2^{-2\kappa}}{\kappa(1-2^{-2\kappa})}, \quad c_{\kappa,2} := \frac{(6 + 2\kappa)\log(2)}{(1- 2^{-2\kappa})^2}.
\end{align}
\end{lemma}

By applying this lemma with the first main inequality~\eqref{K_beta_psi_1,t_3_main_ineq_1}, i.e., $\kappa\to \kappa_\Delta$ and $C\to C_{\tilde{g},\Delta,S_0}$, we immediately prove the second main inequality~\eqref{K_beta_psi_1,t_3_main_ineq_2}.
This completes the proof of Theorem~\ref{thm:generalized_area_law}. $\square$

\textit{Proof of Lemma~\ref{lemm:approx_Ent_entroy}.}
We consider a set of $\{D_p\}_{p=1}^\infty$ such that $D_p=2^{p+p_0}$ and define $\delta_p$ as 
 \begin{align}
\delta_p:= \frac{C}{D_p^{\kappa}} = C 2^{-(p+p_0)\kappa},
\end{align}
where the value of $p_0$ will be chosen appropriately to satisfy a condition required for the proof.
Temporarily, we leave $p_0$ free and proceed.  
Then, by using Ref.~\cite[Supplemental Proposition~3]{Kuwahara2020arealaw}, we obtain 
\begin{align}
E_{\alpha=1}(\psi)
&\le \log(D_{0})+\sum_{p=0}^\infty \delta_p^2 \log \br{\frac{3D_{p+1}}{\delta_p^2}}  \notag \\
&=p_0 \log(2) + C^2  \sum_{p=0}^\infty 2^{-2(p+p_0)\kappa} \log \br{\frac{6}{C^2} 2^{(p+p_0)(2\kappa + 1)}}   \notag \\
&\le p_0 \log(2) + \frac{2^{-2p_0 \kappa} C^2}{1-2^{-2\kappa}} \log \br{\frac{6\cdot  2^{p_0(2\kappa + 1)}}{C^2} } +  \frac{C^2 2^{-2\kappa (p_0+1)}  (2\kappa +1) \log(2)}{(1- 2^{-2\kappa})^2},
\end{align}
where we use $\sum_{p=0}^\infty p x^p = x (d/dx)  \sum_{p=0}^\infty x^p=x (d/dx) 1/(1-x)=x/(1-x)^2 $.
Finally, by choosing $p_0= \ceil{(2\kappa)^{-1} \log_2(C^2)} \le (2\kappa)^{-1} \log_2(C^2)+1$, we obtain $2^{-2p_0 \kappa} \le C^{-2}$, and hence the main inequality is proven as follows:
\begin{align}
E_{\alpha=1}(\psi)
&\le p_0 \log(2) + \frac{1}{1-2^{-2\kappa}} \log \br{6\cdot  2^{p_0+2\kappa} } +  \frac{2^{-2\kappa}  (2\kappa +1) \log(2)}{(1- 2^{-2\kappa})^2} \notag \\
&\le \frac{2-2^{-2\kappa}}{\kappa(1-2^{-2\kappa})} \log(C) + \log(2) \frac{6 + 2\kappa}{(1- 2^{-2\kappa})^2} .
\end{align}
This completes the proof. $\square$

\section{Simulating 1D long-range interactions with polynomial complexity} \label{Sec:1D_long}

In this section, we investigate the simulation of one-dimensional quantum systems with long-range interactions. 
Our focus is on the computational complexity of simulating ground states, quantum dynamics, and quantum thermal states on a classical computer. 
The best known results to date achieve only quasi-polynomial complexity, requiring computational cost of the form $e^{\poly\log(n/\epsilon)}$ in terms of the system size $n$ and error tolerance $\epsilon$. 
The goal of this work is to improve this to true polynomial complexity, namely $e^{\log(n/\epsilon)}$.

More concretely, we establish the following results:
\begin{itemize}
\item {\bf Ground states and thermal states:} We prove the existence of efficient descriptions in terms of matrix product states (MPSs) and matrix product operators (MPOs) with polynomial bond dimension. 
\item {\bf Quantum dynamics:} We provide a rigorous, polynomial-time simulation algorithm based on MPS representations. 
In particular, our analysis gives the first formal accuracy guarantees for the widely used time-dependent density-matrix renormalization group (t-DMRG) algorithm. 
\end{itemize}

\subsection{SE Strength for $k$-local Hamiltonians}
We first show a simple lemma on the SE strength of the 1D long-range interacting Hamiltonians [see Eq.~\eqref{def__long_range}]. 
This lemma shows that the SE strength is given by $\orderof{1}$ constant as long as the power law decay is faster than $r^{-2}$. 
\begin{lemma} \label{lem:Long-range_SE_Strength}
Let $H$ be an arbitrary 1D Hamiltonians defined by Eqs.~\eqref{def_interaction_decay} and \eqref{def__long_range}.
Then, for an arbitrary decomposition $\Lambda = A \sqcup B$ with $A = \{1, 2, \ldots, i\}$ and $B = \{i+1, i+2, \dots, n\}$, the SE strength $\bar{\mJ}(V_{AB})$ is upper-bounded by 
\begin{align}
\label{Long-range_SE_Strength_main}
\bar{\mJ}(V_{AB}) \le \tilde{J} :=  \frac{\eta J_0}{\eta-2} ,
\end{align} 
where $J_0$ and $\eta$ characterizes the interaction decay as in Eq.~\eqref{def__long_range}. 
\end{lemma}

{\bf Remark.} When the interaction decay is slower than $r^{-2}$, the entanglement generation is known to be highly enhanced even in the context of the standard SIE theorem~\cite{PhysRevLett.128.010603}.
Most of the locality properties break down for $\eta<2$~\cite{10.1143/PTPS.64.12,Kuwahara2020arealaw,PhysRevX.11.031016,kimura2024clustering,PhysRevLett.134.190404} with a few exceptions~\cite{kuwahara2020absence,kim2024thermal}.

\textit{Proof of Lemma~\ref{lem:Long-range_SE_Strength}.}
From the upper bound~\eqref{Trivial_Ineq_interaction_strength}, one can derive 
\begin{align}
\bar{\mJ}(V_{AB}) \le \sum_{\substack{Z : Z \cap A \neq \emptyset,\\ Z \cap B \neq \emptyset}} \norm{h_Z}.
\end{align}
The above quantity was simply treated using the inequality~\eqref{def__long_range} in Ref.~\cite[Supplementary Lemma~1]{Kuwahara2020arealaw}, which reduces the above inequality to the desired one~\eqref{Long-range_SE_Strength_main}.
This completes the proof. $\square$

\subsection{MPS description of the long-range ground state}

In this subsection, we discuss the matrix product state (MPS) representation of gapped ground states in one-dimensional long-range interacting systems. 
The best known result to date shows that approximating the ground state within error $\epsilon$ requires an MPS with bond dimension 
$e^{\log^{5/2}(n/\epsilon)}$,
where $n$ denotes the system size. 

Here, we improve upon this bound by extending the result of Ref.~\cite{Kuwahara2020arealaw}. 
In particular, building on Proposition~\ref{Prop:AGSP_construction_spectral_SIE}, which establishes an AGSP construction tailored to long-range systems, we show that the ground state admits an efficient polynomially bounded MPS approximation:

\begin{theorem} \label{thnm:Gs_approx}
Let $\ket{\Omega}$ be the ground state for a given long-range interacting Hamiltonian $H$.
For an arbitrary cut of the total system as $\Lambda=A_s \sqcup B_s$ (i.e., $A_s=\{1,2,\ldots,s\}$), we consider the Schmidt decomposition of $\ket{\Omega}$ as follows:
\begin{align}
\label{ket_Omega_Schmidt}
\ket{\Omega} = \sum_{j=1} \lambda_j^{(s)}\ket{\Omega_{A_s,j}}  \otimes \ket{\Omega_{B_s,j}}.
\end{align} 
Then, the truncation error of the Schmidt rank is bounded above by
 \begin{align}
 \label{thnm:Gs_approx_firts_main_ineq}
 \sum_{j>D}\br{ \lambda^{(s)}_j}^2 \le 32 D_0 D^{-\Delta/(2 \tilde{J}+\Delta)}  ,
\end{align} 
with
\begin{align}
\log(D_0)= c^\ast \log^2 (d) \br{\frac{\log(d)}{\Delta}}^{1+2/\bar{\eta}} \log^{3+3/\bar{\eta}}\br{\frac{\log(d)}{\Delta}} ,
\end{align} 
where $\bar{\eta}=\eta-2$, and $c^\ast$ is an $\orderof{1}$ constant depending on $g$ in~\eqref{eq:Hdef} and $\bar{\eta}$. 
Note that the constant $\tilde{J}$ has been defined in Eq.~\eqref{Long-range_SE_Strength_main} as an upper bound for the SE strength of $V_{AB}$. 

Moreover, there exists an MPS $\ket{\Mt(D)}$ with the bond dimension $D$ that approximates the gapped ground state $\ket{\Omega}$ on a subset $X$ within an error of 
\begin{align}
 \label{thnm:Gs_approx_second_main}
\norm{\tr_{X^\co} \br{ \ket{\Omega} \bra{\Omega}}- \tr_{X^\co}\br{\ket{\Mt(D)} \bra{\Mt(D)}}}^2 \le 64(|X|+1) D_0 D^{-\Delta/(2 \tilde{J}+\Delta)} ,
\end{align}
where the subset $X$ is concatenated and can be chosen arbitrarily. 
\end{theorem}

{\bf Remark.} 
From the theorem, to achieve an arbitrary error $\epsilon$ for the MPS approximation $\norm{ \ket{\Omega} - \ket{\Mt(D)}}$, we need to choose the bond dimension $D$ as 
\begin{align}
D=\br{ \frac{128 n D_0}{\epsilon}}^{1+ 2 \tilde{J}/\Delta} ,
\end{align}
where we use $|\Lambda|+1\le 2n$. 

\subsubsection{Proof of Theorem~\ref{thnm:Gs_approx}}

For the proof, we utilize the Eckart--Young theorem~\cite{Eckart1936}: 
 \begin{align}
\label{Eckart--Young_gs}
 \sum_{j>D}\br{ \lambda^{(s)}_j}^2 \le \norm{\ket{\Omega} - \ket{\tilde{\phi}_D}}^2 , 
\end{align} 
where $\ket{\tilde{\phi}_D}$ is an arbitrary (unnormalized) quantum state with the Schmidt rank $D$ between $A_s$ and $B_s$. 
Then, we utilize the result in Ref.~\cite[Combining Lemma~7 and Ineq. (193) in its Supplementary Note]{Kuwahara2020arealaw}, which ensures the existence of a product state $\ket{\phi}$ satisfying
 \begin{align}
 \label{upper_bound_Omega_phi_D0}
\abs{ \langle \phi \ket{\Omega} } \ge \frac{1}{\sqrt{2D_0}}   .
\end{align}

We then apply the AGSP $K_\beta$ to $\ket{\phi_{D_0}}$ and estimate the error of 
\begin{align}
\norm{\langle \Omega \ket{\phi} \ket{\Omega} - K_\beta \ket{\phi}}.   
\end{align} 
By expanding $\ket{\phi}$ as
 \begin{align}
\ket{\phi} = \langle \Omega \ket{\phi} \ket{\Omega} + \sqrt{1- \abs{\langle \Omega \ket{\phi} }^2} \ket{\Omega_\bot} ,\quad \bra{\Omega} \Omega_\bot \rangle =0 , 
\end{align} 
we obtain 
\begin{align}
\label{norm_Omega_AGSP_error}
\norm{ \langle \Omega \ket{\phi} \ket{\Omega} - K_\beta \ket{\phi}} 
&\le    \norm{ \langle \Omega \ket{\phi} \br{1- K_\beta}  \ket{\Omega} + \sqrt{1- \abs{\langle \Omega \ket{\phi} }^2}K_\beta \ket{\Omega_\bot} } \notag\\
&\le   \abs{\langle \Omega \ket{\phi}} \cdot \norm{1- K_\beta}+ \sqrt{1- \abs{\langle \Omega \ket{\phi} }^2} \norm{K_\beta \ket{\Omega_\bot}}
\le 3 e^{-\beta \Delta^2}  ,
\end{align} 
where we use the inequality~\eqref{Prop:AGSP_construction_spectral_SIE_main1} in Proposition~\ref{Prop:AGSP_construction_spectral_SIE}. 

Moreover, using the inequality~\eqref{Prop:AGSP_construction_spectral_SIE_apply} in Proposition~\ref{Prop:AGSP_construction_spectral_SIE}, we can also prove the following approximation:
\begin{align}
\norm{ K_\beta \ket{\phi} - \ket{\phi_D}  } \le \frac{e^{2\beta \Delta \tilde{J}}}{\sqrt{D}} ,
\end{align}  
where $\ket{\phi_D}$ has the Schmidt rank of $D$ for the bipartition $A_s$ and $B_s$. 
This reduces the inequality~\eqref{norm_Omega_AGSP_error} to 
\begin{align}
\label{norm_Omega_AGSP_error_2}
\norm{ \langle \Omega \ket{\phi} \ket{\Omega}  - \ket{\phi_D}}
&\le 3 e^{-\beta \Delta^2} + \frac{e^{2\beta \Delta \tilde{J}}}{\sqrt{D}}  .
\end{align} 
Finally, by applying the lower bound~\eqref{upper_bound_Omega_phi_D0} to the above inequality, we arrive at the error bound of
\begin{align}
\label{norm_Omega_AGSP_error_3}
\norm{\ket{\Omega}  -\frac{1}{ \langle \Omega \ket{\phi}} \ket{\phi_D}}
&\le 3\sqrt{2D_0} e^{-\beta \Delta^2} + \frac{\sqrt{2D_0} e^{2\beta \Delta \tilde{J}}}{\sqrt{D}}  .
\end{align} 

By choosing $\beta$ such that $e^{-\beta \Delta^2} = e^{2\beta \Delta \tilde{J}}/\sqrt{D}$, i.e., 
 \begin{align}
\beta = \frac{1}{2\Delta (2 \tilde{J}+\Delta)} \log(D) ,
\end{align} 
we further reduce the inequality~\eqref{norm_Omega_AGSP_error_3} to 
\begin{align} 
\label{norm_Omega_AGSP_error_fin}
\norm{\ket{\Omega}  -\frac{1}{ \langle \Omega \ket{\phi}} \ket{\phi_D}}
&\le 4 \sqrt{2D_0} e^{-\frac{\Delta}{2(2 \tilde{J}+\Delta)} \log(D) } . 
\end{align} 
Therefore, by letting $\ket{\tilde{\phi}_D}= \ket{\phi_D}/\langle \Omega \ket{\phi}$ in the inequality~\eqref{Eckart--Young_gs}, we prove the first main inequality~\eqref{thnm:Gs_approx_firts_main_ineq}.
 
To derive the approximation error by MPS, we utilize the statement in Ref.~\cite[Lemma~1 therein]{PhysRevB.73.094423} as follows.
Let $\psi$ be an arbitrary quantum state. 
Then, for the Schmidt decomposition of $\ket{\Omega}$ in Eq.~\eqref{ket_Omega_Schmidt},  we define the parameter $\delta_s(D)$ as 
 \begin{align}
\label{definition_delta_s_D}
\delta_s(D) = \norm{ \sum_{j>D} \lambda^{(s)}_j \ket{\phi_{A_s,j}}  \otimes \ket{\phi_{B_s,j}}  } =\brr{ \sum_{j>D}\br{ \lambda^{(s)}_j}^2}^{1/2} .
\end{align} 
Then, there exists an MPS $\ket{\Mt(D)}$ such that 
\begin{align}
\label{Verstraete_Cirac_PRB}
\norm{\tr_{X^\co} \br{ \ket{\Omega} \bra{\Omega}}- \tr_{X^\co}\br{\ket{\Mt(D)} \bra{\Mt(D)}}}^2 \le 2 \sum_{s=i_0-1}^{i_1} \delta^2_s(D) ,
\end{align} 
where we denote the set $X$ by $\{i_0 , i_0+1 ,\ldots, i_1\}$ with $|i_1-i_0|=|X|$. 
By combining the inequalities~\eqref{Eckart--Young_gs} and~\eqref{norm_Omega_AGSP_error_fin}, we obtain 
\begin{align}
\delta^2_s(D) \le 32 D_0 D^{-\Delta/(2 \tilde{J}+\Delta)} ,
\end{align} 
which reduces the upper bound~\eqref{Verstraete_Cirac_PRB} to the second main inequality~\eqref{thnm:Gs_approx_second_main}. 
This completes the proof. $\square$


\subsection{Quantum dynamics}  \label{Sec:Quantum dynamics}

We next consider the quantum dynamics.
We begin with an initial product state $\ket{\phi}$ and consider the efficient computation of the time-evolved state $\ket{\phi_t}=e^{-iHt} \ket{\phi}$ 

We first prove the existence of a good MPS approximation for the time-evolved state as a simple application of Theorem~\ref{thm:Renyi_SIE}. 
\begin{prop} \label{Prop:MPO_approximation}
For any given $\ket{\Pro}$, there always exists an MPS $\ket{\Mt_t(D)}$ with the bond dimension $D$ that approximates $e^{-iHt} \ket{\Pro}$, satisfying the following error bound:
\begin{align}
 \label{Prop:MPO_approximation_main}
\norm{ e^{-iHt} \ket{\Pro} -\ket{\Mt_t(D)}}^2 \le \frac{2 e^{2\tilde{J} t}}{D} n ,
\end{align} 
where $\tilde{J}$ was defined in Eq.~\eqref{Long-range_SE_Strength_main}.
\end{prop}

\subsubsection{Proof of Proposition~\ref{Prop:MPO_approximation}}

Let us denote the Schmidt decomposition of $\ket{\phi_t}$ by
\begin{align}
\ket{\phi_t}= \sum_{j} \lambda_j^{(s)}\ket{\phi_{A_s,j}}  \otimes \ket{\phi_{B_s,j}} ,
\end{align}
where the index $s$ characterizes the decomposition of the total system ($\Lambda= A_s \sqcup B_s$). 
Then, we utilize the inequality~\eqref{Verstraete_Cirac_PRB}. For this purpose, we need to estimate the parameter $\delta_s(D)$ for the Schmidt decomposition:
 \begin{align}
\label{definition_delta_s_D_again} 
\delta_s(D) = \brr{ \sum_{j>D}\br{ \lambda^{(s)}_j}^2}^{1/2} ,
\end{align} 
which gives 
\begin{align}
 \label{Verstraete_Cirac_PRB__again} 
\norm{ \ket{\Pro_t} -\ket{\Mt_t(D)}}^2 \le 2\sum_{s=1}^{n-1} \delta_s^2(D) .
\end{align} 

Using Theorem~\ref{thm:Renyi_SIE} with the parameter $\tilde{J}$, we can upper-bound the (1/2)-R\'enyi entanglement entropy between $A_s$ and $B_s$ as 
\begin{align}
E_{1/2} (\Pro_t) \le 2 \tilde{J} t , 
\end{align} 
and hence, we can derive from the inequality~\eqref{lemm:Renyi_Schmidt_main_ineq} in Lemma~\ref{lemm:Renyi_Schmidt}
\begin{align}
\lambda^{(s)}_j \le \frac{e^{E_{1/2} (\Pro_t)/2}}{j} \le \frac{e^{\tilde{J} t}}{j} ,
\end{align} 
where we choose $\alpha=1/2$.
Therefore, we obtain 
\begin{align}
\label{uuper_bound_detal_2_s_d}
\delta^2_s(D)= \sum_{j>D}\br{ \lambda^{(s)}_j}^2  \le e^{2\tilde{J} t} \sum_{j>D} j^{-2}  \le \frac{e^{2\tilde{J} t}}{D} ,
\end{align} 
which reduces the upper bound~\eqref{Verstraete_Cirac_PRB__again} to the main inequality~\eqref{Prop:MPO_approximation_main}. 
This completes the proof. $\square$

%
%

%
%

\subsection{Precision guarantee for time-dependent density-matrix-renormalization-group (t-DMRG) algorithm}

In this section, we prove that the MPS in Proposition~\ref{Prop:MPO_approximation} can be efficiently computed with a precision guarantee by using the t-DMRG algorithm. 
As has been pointed out in Ref.~\cite{PhysRevLett.97.157202} (see also the inequality~\eqref{Error_exp_increase_t-DMRG} below), a rigorous precision guarantee for the t-DMRG algorithm has been an open problem.

\subsubsection{Review of t-DMRG}
We review the t-DMRG algorithm~\cite{PhysRevLett.93.040502,PhysRevLett.93.076401,PAECKEL2019167998}. 
The purpose of the time-dependent DMRG (t-DMRG) algorithm is to construct a matrix product state (MPS) that approximates the time evolution
\[
e^{-iHt}\ket{\psi_0}.
\]
We divide the total time into $\mathcal{N}$ steps and set $\Delta t = t/\mathcal{N}$. At each step, the action of $e^{-iH\Delta t}$ on the state is approximated by an MPS with bond dimension $D$ through the following two-stage procedure:

\begin{enumerate}
\item{} First, we apply the first-order expansion
\begin{align}
e^{-iH\Delta t}\ket{\psi} \approx (1 - iH\Delta t)\ket{\psi}.
\end{align}

\item{} The resulting state is then approximated by an MPS truncated up to bond dimension $D$:
\begin{align}
(1 - iH\Delta t)\ket{\psi} \approx \ket{M_1(D)} ,
\end{align}
where $\ket{M_1(D)}$ is not necessarily normalized, i.e., $\|\ket{M_1(D)}\|\leq 1$.
Repeating the same procedure, we approximate
\[
e^{-iH\Delta t}\ket{M_m(D)} \approx \ket{M_{m+1}(D)}.
\]
\item{} After $\mathcal{N}$ steps, we obtain the MPS $\ket{M_{\mathcal{N}}(D)}$ as an approximation to $e^{-iHt}\ket{\psi_0}$.
\end{enumerate}

Note that $\ket{M_{\mathcal{N}}(D)}$ may not be normalized. However, the following error bound holds:
\begin{align}
\left\|
\frac{\ket{M_{\mathcal{N}}(D)}}{\|\ket{M_{\mathcal{N}}(D)}\|}
- e^{-iHt}\ket{\psi_0}
\right\|
&\leq 
\|\ket{M_{\mathcal{N}}(D)} - e^{-iHt}\ket{\psi_0}\|
+\left\|
\frac{\ket{M_{\mathcal{N}}(D)}}{\|\ket{M_{\mathcal{N}}(D)}\|}
-\ket{M_{\mathcal{N}}(D)}
\right\| \notag \\
&\leq 
\|\ket{M_{\mathcal{N}}(D)} - e^{-iHt}\ket{\psi_0}\|
+ \br{1 - \|\ket{M_{\mathcal{N}}(D)}\|}.
\end{align}
Using the inequality
\[
\|\ket{M_{\mathcal{N}}(D)}\| \geq 
\|e^{-iHt}\ket{\psi_0}\| - \|\ket{M_{\mathcal{N}}(D)} - e^{-iHt}\ket{\psi_0}\|
\geq 1 - \|\ket{M_{\mathcal{N}}(D)} - e^{-iHt}\ket{\psi_0}\|,
\]
we obtain the bound
\begin{align}
\left\|
\frac{\ket{M_{\mathcal{N}}(D)}}{\|\ket{M_{\mathcal{N}}(D)}\|}
- e^{-iHt}\ket{\psi_0}
\right\|
\leq 2\epsilon ,
\quad
\epsilon := \|\ket{M_{\mathcal{N}}(D)} - e^{-iHt}\ket{\psi_0}\|.
\end{align}

\subsubsection{Error guarantee and the challenging point} \label{sec:Err_gua_challenge}

Here, we consider the error accumulation in each of the processes. 
We prove the following lemma. 
\begin{lemma}\label{lem:error_precision_guarantee_t-DMRG}
The number of steps $\mathcal{N}$ is chosen to satisfy
\begin{align}
\label{Delta_t_condition}
gn \Delta t =\frac{gn t}{\mN} \le 1 .
\end{align} 
Then, the total error in the t-DMRG algorithm is given by 
\begin{align}
\label{lem:error_precision_guarantee_t-DMRG_main}
\norm{ - e^{-i Ht}  \ket{\Pro}-  \ket{\Mt_{\mN}(D)}} \le   \frac{(g nt)^2}{\mN}  +\sqrt{2n} \sum_{m=0}^{\mN-1} \bar{\delta}_m(D)  ,
\end{align}  
where $\bar{\delta}_m(D) $ is defined by Eqs.~\eqref{Second_approx_t-DMRG_m} and \eqref{Second_approx_t-DMRG_m_delta} below. 
\end{lemma}

\textit{Proof of Lemma~\ref{lem:error_precision_guarantee_t-DMRG}. }
We first note that the condition~\eqref{Delta_t_condition} implies 
\begin{align}
\Delta t \norm{H} \le  gn  \Delta t   \le 1 , 
\end{align}  
because of 
\begin{align}
\norm{H} \le \sum_{i\in \Lambda} \sum_{Z: Z\ni i} \norm{h_Z} \le\sum_{i\in \Lambda} g =gn,
\end{align}  
where we use the parameter $g$ in~\eqref{eq:Hdef}. 

For the proof, we define the error $\epsilon_m$ as 
\begin{align}
\epsilon_m:= \norm{e^{-i H m \Delta t}   \ket{\Pro} -\ket{\Mt_{m}(D)} } . 
\end{align}  
We aim to derive an upper bound for $\epsilon_\mN$. 
Then, we analyze how the error increases with the approximation step. 
Given the error $\epsilon_m$ and the MPS $\ket{\Mt_{m}(D)}$, we consider $\epsilon_{m+1}$ and $\ket{\Mt_{m+1}(D)}$.
We start from 
\begin{align}
\norm{ e^{-i H \Delta t} \ket{\Mt_{m}(D)} - e^{-i H (m+1) \Delta t}  \ket{\Pro}} \le 
\norm{ \ket{\Mt_{m}(D)} -  e^{-i H m\Delta t}  \ket{\Pro}} \le \epsilon_m . 
\end{align}  
Hence, we get 
\begin{align}
\label{Upper_ep_m+1_from_ep_m}
\epsilon_{m+1}=\norm{  e^{-i H (m+1) \Delta t}  \ket{\Pro} -\ket{\Mt_{m+1}(D)} } \le 
\epsilon_m + \norm{ e^{-i H \Delta t} \ket{\Mt_{m}(D)} -\ket{\Mt_{m+1}(D)} }  .
\end{align}

In the following, we estimate the error $\norm{ e^{-i H \Delta t} \ket{\Mt_{m}(D)} -\ket{\Mt_{m+1}(D)} }$. 
By using $e^{x}-x-1 \le x^2$ for $0\le x\le1$ and $\Delta t \norm{H}\le 1$, we obtain 
\begin{align}
\label{error_estimation_Delta_t_H_exp}
\norm{e^{-i H \Delta t} - \br{1-i H \Delta t}} \le e^{\Delta t \norm{H}} - 1 - \Delta t \norm{H} \le (\Delta t \norm{H})^2 \le (gn \Delta t)^2,
\end{align}  
which yields 
\begin{align}
\label{First_approx_t-DMRG_m}
\norm{ e^{-i H \Delta t} \ket{\Mt_{m}(D)} - \br{1-i H \Delta t} \ket{\Mt_{m}(D)} } \le (gn \Delta t)^2 = \frac{(g nt)^2}{\mN^2}  . 
\end{align}

To bound the second term
\begin{align}
\norm{\br{1-i H \Delta t} \ket{\Mt_{m}(D)} -  \ket{\Mt_{m+1}(D)} } ,
\end{align}  
we define the Schmidt decomposition of $\br{1-i H \Delta t} \ket{\Mt_{m}(D)}$ between $A_s$ and $B_s$ as follows: 
\begin{align}
\label{1-iH_Delta_t_Mtm}
\br{1-i H \Delta t} \ket{\Mt_{m}(D)} = \sum_{j} \lambda^{(s)}_{m,j} \ket{\phi^{(m)}_{A_s,j}} \otimes \ket{\phi^{(m)}_{B_s,j}} .
\end{align}  
Using the same upper bound as Eq.~\eqref{Verstraete_Cirac_PRB__again} with \eqref{definition_delta_s_D_again}, we obtain  
\begin{align}
\label{Second_approx_t-DMRG_m}
\norm{\br{1-i H \Delta t} \ket{\Mt_{m}(D)} -  \ket{\Mt_{m+1}(D)} }^2 \le 2 \sum_{s=1}^{n-1} \delta^2_{m,s} (D) \le 2n \bar{\delta}^2_m(D)
\end{align}  
with
\begin{align}
\label{Second_approx_t-DMRG_m_delta}
&\bar{\delta}_m(D) := \max_{s\in [n-1]} [\delta_{m,s} (D) ]  , \notag \\
&\delta_{m,s} (D) =\brr{ \sum_{j>D}\br{\lambda^{(s)}_{m,j}}^2}^{1/2} .
\end{align} 
Note that the explicit form of $\ket{\Mt_{m+1}(D)}$ is systematically constructed by using the canonical form of the matrix product state~\cite{PhysRevB.73.094423,SCHOLLWOCK201196}.

By combining the inequalities~\eqref{First_approx_t-DMRG_m} and \eqref{Second_approx_t-DMRG_m}, we derive 
\begin{align}
\norm{ e^{-i H \Delta t} \ket{\Mt_{m}(D)} -\ket{\Mt_{m+1}(D)} } \le  \frac{(g nt)^2}{\mN^2}  +\sqrt{2n} \bar{\delta}_m(D)  .
\end{align}  
Therefore, from the upper bound~\eqref{Upper_ep_m+1_from_ep_m}, we obtain 
\begin{align}
\label{upp_epsilon_m+1}
\epsilon_{m+1} \le 
\epsilon_m +  \frac{(g nt)^2}{\mN^2}  + \sqrt{2n} \bar{\delta}_m(D)   .
\end{align}  
By solving the above recurrence inequality and considering $\epsilon_{\mN}$, we arrived at the desired inequality~\eqref{lem:error_precision_guarantee_t-DMRG_main}. 
This completes the proof.  $\square $

{~}

\hrulefill{\bf [ End of Proof of Lemma~\ref{lem:error_precision_guarantee_t-DMRG}]}

{~}

The remaining problem is to estimate the upper bound of $\bar{\delta}_m(D)$. 
The most frequently used way is to utilize the Eckart--Young theorem as follows.
Let $\ket{\phi_{s,D}}$ be an arbitrary quantum state with the Schmidt rank $D$ between $A_s$ and $B_s$.
Then, we upper-bound the approximation error $\delta_{m,s} (D) $ in Eq.~\eqref{Second_approx_t-DMRG_m_delta} by
\begin{align}
\label{Eckart_Young_up}
&\delta_{m,s} (D) \le \norm{\br{1-i H \Delta t} \ket{\Mt_{m}(D)} - \ket{\phi_{s,D}}}. 
\end{align} 
This reduces the problem to find a good state $\ket{\phi_{s,D}}$ to approximate $\br{1-i H \Delta t} \ket{\Mt_{m}(D)}$. 

For this purpose, we consider $e^{-i H(m+1) \Delta t}\ket{\Pro}$ and consider the Schmidt decomposition between $A_s$ and $B_s$ as follows:
\begin{align}
e^{-i H(m+1) \Delta t}\ket{\Pro}= \sum_{j} \tilde{\lambda}^{(s)}_{j,m+1} \ket{\tilde{\phi}^{(m)}_{A_s,j}} \otimes \ket{\tilde{\phi}^{(m)}_{B_s,j}} ,
\end{align}  
which is approximated by $\ket{\psi_{s,D}}$ as 
\begin{align}
\ket{\psi_{s,D}} := \sum_{j=1}^D \tilde{\lambda}^{(s)}_{j,m+1} \ket{\tilde{\phi}^{(m)}_{A_s,j}} \otimes \ket{\tilde{\phi}^{(m)}_{B_s,j}} .
\end{align}  
Then, by following the proof of Proposition~\ref{Prop:MPO_approximation}, we derive 
\begin{align}
\norm{ e^{-i H(m+1) \Delta t}\ket{\Pro} - \ket{\psi_{s,D}}}^2 \le \frac{e^{2\tilde{J} t}}{D}  .
\end{align} 
Therefore, we also obtain 
\begin{align}
&\norm{\br{1-i H \Delta t}\ket{\Mt_m(D)} - \ket{\psi_{s,D}}}  \notag \\
&\le \norm{e^{-i H \Delta t}\ket{\Mt_m(D)} - \ket{\psi_{s,D}}} +\norm{\br{1-i H \Delta t}- e^{-i H \Delta t}}  \notag \\
&\le \norm{e^{-i H \Delta t} \br{ \ket{\Mt_m(D)} - e^{-i Hm \Delta t}\ket{\Pro}}}+
\norm{e^{-i H(m+1) \Delta t}\ket{\Pro}-\ket{\psi_{s,D}}} +  \frac{(g nt)^2}{\mN^2}   \notag \\
&\le \epsilon_m + \frac{e^{\tilde{J} t}}{D^{1/2}} +  \frac{(g nt)^2}{\mN^2}  ,
\end{align}  
which, when applied to~\eqref{Eckart_Young_up}, implies 
\begin{align}
\label{Eckart_Young_up_apply}
&\delta_{m,s} (D) \le \epsilon_m + \frac{e^{\tilde{J} t}}{D^{1/2}} +  \frac{(g nt)^2}{\mN^2}    . 
\end{align} 

Combining \eqref{Eckart_Young_up_apply} with \eqref{upp_epsilon_m+1} yields
\begin{align}
\epsilon_{m+1}
&\le \epsilon_m +  \frac{(g nt)^2}{\mN^2}
   + \sqrt{2n} \bar\delta_m(D) \notag\\
&\le (1+\sqrt{2n}) \epsilon_m
 + \frac{\sqrt{2n} e^{\tilde{J} t}}{\sqrt D}
 + (1+\sqrt{2n}) \frac{(g nt)^2}{\mN^2}. \label{eq:recurrence_final}
\end{align}
Solving the linear recurrence, we obtain for all $m\ge 0$,
\begin{align}
\label{Error_exp_increase_t-DMRG}
\epsilon_m \le (1+\sqrt{2n})^m \epsilon_0
 + \frac{(1+\sqrt{2n})^m-1}{\sqrt{2n}}
   \!\left(
     \frac{\sqrt{2n} e^{\tilde{J} t}}{\sqrt D}
     + (1+\sqrt{2n}) \frac{(g nt)^2}{\mN^2}
   \right).
\end{align}
In particular, the error grows exponentially in the number of steps $m$ (with rate $\log(1+\sqrt{2n})$), highlighting the difficulty of establishing an efficiency guarantee for t-DMRG~\cite{PhysRevLett.97.157202}.

%

\subsubsection{Main theorem on the precision guarantee of t-DMRG}

Here, we resolve the problem of obtaining a meaningful precision guarantee for the t-DMRG method: 
Indeed, we prove the following theorem:
\begin{theorem} \label{Thm:error_efficiency_t-DMRG}
Let $\bar{\delta}_m(D)$ be the error parameter defined in Eq.~\eqref{Second_approx_t-DMRG_m_delta}.
Then, for any step $m\in\{0,\dots,\mathcal{N}-1\}$, we obtain
\begin{align}
\label{Thm:error_efficiency_t-DMRG/main1}
\bar{\delta}_m(D)
\le
\frac{ e^{\tilde{J} t + (gnt)^2 / \mN }}{\sqrt{D}} .
\end{align}
Combining Lemma~\ref{lem:error_precision_guarantee_t-DMRG}, the following inequality holds:
\begin{align}
\label{Thm:error_efficiency_t-DMRG/main2}
\bigl\|e^{-iHt}\ket{\psi_0}-\ket{\Mt_{\mathcal{N}}(D)}\bigr\|
\le
\frac{(gnt)^2}{\mathcal{N}}
+
\frac{\mathcal{N} }{\sqrt{D/(2n)}}e^{\tilde{J} t + (gnt)^2 / \mN } . 
\end{align}
\end{theorem}

{\bf Remark.} According to the theorem, to achieve accuracy $\epsilon$, one may set $\mathcal{N} \propto t^2 n^2/\epsilon$ and $D \propto e^{\orderof{t}} n^5/\epsilon^4$. This choice leads to a time complexity that is merely polynomial in $n/\epsilon$.

\subsubsection{Key Idea: R\'enyi-entanglement monitoring}

\begin{figure}[ttt]
  \centering
  \includegraphics[width=1\textwidth]{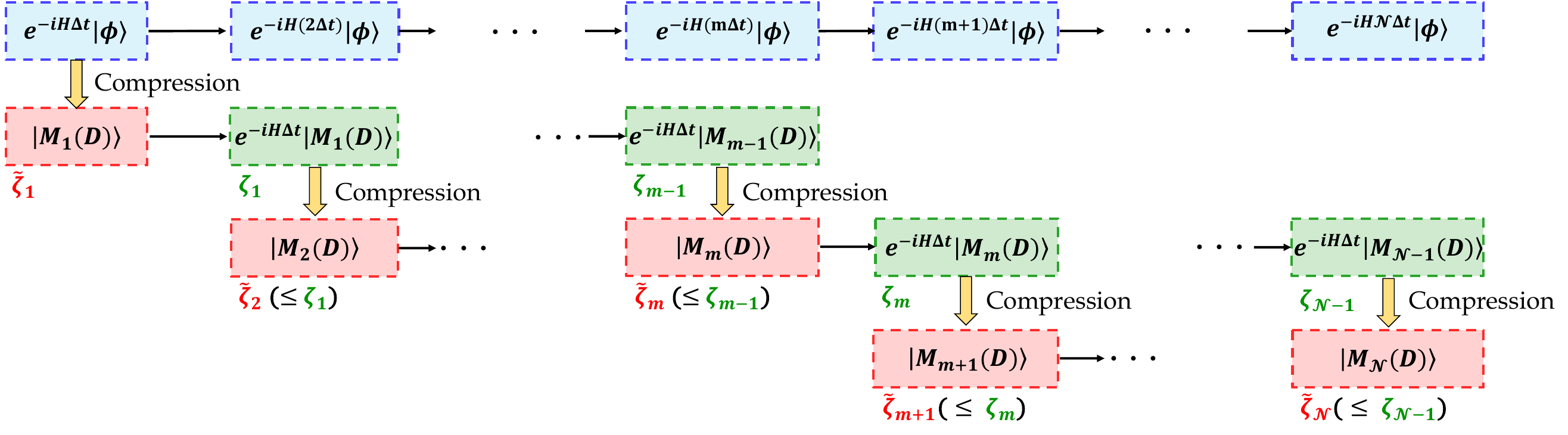}
  \caption{
Conceptual illustration of the monitoring-based analysis in t-DMRG algorithm.
The exact evolution (top row, blue) is compared with the algorithmic trajectory: 
each pre-truncation state (middle row, green, dashed) is compressed to a rank-$D$ MPS (bottom row, red, dashed).
Vertical arrows labelled ``Compression'' indicate the truncation step, during which the R\'enyi-$1/2$ proxy (sum of Schmidt coefficients) can only decrease.
By tracking this proxy across steps, one can bound the entanglement growth and thus control the accumulated truncation error.
}
  \label{fig:t-DMRG}
\end{figure}

Our goal is to obtain a rigorous and practically meaningful precision guarantee for t-DMRG.
The core strategy is to \emph{monitor not the error itself (i.e., $\epsilon_m$) but the entanglement proxy that controls it}, namely the sum of Schmidt coefficients (equivalently, the exponential of R\'enyi-$1/2$ entanglement), at each time step. The argument proceeds in three steps.
\begin{enumerate}
\item{[Step 1: Monitor R\'enyi-$1/2$ instead of the raw error $\epsilon_m$]} \\
Let the Schmidt decompositions of the (generally unnormalized) MPS $\ket{\Mt_{m}(D)}$ at step $m$ across a fixed cut $(A_s,B_s)$ be
\begin{align}
\ket{\Mt_{m}(D)} = \sum_{j=1}^D \tilde{\lambda}^{(s)}_{m,j} \ket{\psi^{(m)}_{A_s,j}} \otimes \ket{\psi^{(m)}_{B_s,j}} .
\end{align}  
In the same way, we denote the Schmidt decomposition of $\br{1-i H \Delta t} \ket{\Mt_{m}(D)}$ by 
\begin{align}
\label{Schmidt_decompo_1-iHDeltat}
\br{1-i H \Delta t} \ket{\Mt_{m}(D)} = \sum_{j=1}^\infty \lambda^{(s)}_{m,j} \ket{\phi^{(m)}_{A_s,j}} \otimes \ket{\phi^{(m)}_{B_s,j}} ,
\end{align}   
We track
\begin{align}
\label{def_zeta_m}
\zeta_m:=\sum_{j=1}^D \lambda^{(s)}_{m,j}, \quad \tilde{\zeta}_m:= \sum_{j=1}^D \tilde{\lambda}^{(s)}_{m,j} , 
\end{align}   
which is the exponential of the R\'enyi-$1/2$ entanglement for normalized states. 
This scalar quantity serves as a \emph{handle} to control the truncation error.

\item{[Step 2: Show how t-DMRG controls $\tilde{\zeta}_m$ step-by-step]} \\
We split one time step into a local part and an across-the-cut part [see Eq.~\eqref{bar_mJ_1-iH_Delta_t__decomp} below] and derive a precise upper bound on the increase of the Schmidt-coefficient sum after applying $(1-iH\Delta t)$ and truncation:
\begin{align}
\label{zeta_m+1_recursive_0}
\zeta_m \le \brr{1+ \tilde{J} \Delta t +(gn \Delta t)^2 } \tilde{\zeta}_m.
\end{align}
 Importantly, the \emph{Schmidt-rank truncation itself is nonincreasing for R\'enyi-$1/2$ entanglement}, so any entanglement growth must come from the Hamiltonian step: 
 \begin{align}
\label{zeta_m+1_recursive}
\tilde{\zeta}_m \le \zeta_{m-1}.
\end{align}

\item{[Step 3: Convert entanglement control into truncation-error control]} \\
Given $\zeta_m$ for the pre-truncation state $\br{1-i H \Delta t} \ket{\Mt_{m}(D)}$ at step $m$, the tail satisfies
\begin{align}
\label{upp_delta/m/s}
\delta_{m,s}(D) := \br{\sum_{j>D}\abs{\lambda^{(s)}_{m,j}}^2}^{1/2}
 \le  \frac{\zeta_m}{\sqrt D},
\end{align}
where we use the same analysis as in~\eqref{uuper_bound_detal_2_s_d}. 
Thus, an upper bound on $\zeta_m$ (hence on $\tilde{\zeta}_m$) yields a quantitative bound on the truncation error at that step. Summing over steps and cuts gives the final precision guarantee for t-DMRG.
\end{enumerate}

By replacing direct error tracking with a sharp control of R\'enyi-$1/2$ entanglement along the algorithmic trajectory, we obtain explicit trade-offs among $(\Delta t,\mathcal{N})$ and bond dimension $D$, and hence a transparent precision guarantee.

\subsubsection{Proof of Theorem~\ref{Thm:error_efficiency_t-DMRG}}
We use the notation introduced in the previous section. 
Our goal is to derive an upper bound for $\delta_{m,s} (D) $ in Eq.~\eqref{Second_approx_t-DMRG_m_delta} that holds for any $s$: 
\begin{align}
&\delta_{m,s}^2 (D) =\sum_{j>D}\br{\lambda^{(s)}_{m,j}}^2. 
\end{align} 
Note that $\bar{\delta}_m(D)$ is given by $\bar{\delta}_m(D) := \max_{s\in [n-1]} [\delta_{m,s} (D)]$. 
For this purpose, we estimate the exponential of the $(1/2)$-R\'enyi entanglement as in Eq.~\eqref{def_zeta_m}.

In the following, we upper-bound $\zeta_{m}$ based on $\zeta_{m-1}$ (or $\tilde{\zeta}_m$). 
For this purpose, we utilize 
\begin{align}
\label{tilde_Zeta_m}
\zeta_m \le \bar{\mJ} (1-i H \Delta t) \tilde{\zeta}_m .
\end{align} 
Note that $\zeta_m$ is connected to the (1/2) R\'enyi entanglement of $\br{1-i H \Delta t} \ket{\Mt_{m}(D)}$. 
A simple estimate gives $ \bar{\mJ} (1-i H \Delta t) \propto 1 + \orderof{\norm{H} \Delta t}$, which is too loose and leads to an exponential increase in $\tilde{\zeta}_m$ with $m$.
To refine this, we use the following expansion:
\begin{align}
 \label{bar_mJ_1-iH_Delta_t__decomp}
1-i H \Delta t  
&= e^{-i(H_{A_s} + H_{B_s}) \Delta t} -i V_{A_sB_s} \Delta t + \brr{ 1-i (H_{A_s} + H_{B_s}) \Delta t-e^{-i(H_{A_s} + H_{B_s}) \Delta t}} 
\notag \\
&= e^{-i(H_{A_s} + H_{B_s}) \Delta t} -i V_{A_sB_s} \Delta t +  \sum_{m=2}^{\infty} \frac{(-i \Delta t)^m}{m!}(H_{A_s} + H_{B_s})^m ,
\end{align}  
where we decompose the total Hamiltonian into $H=H_{A_s} + H_{B_s}+V_{A_sB_s}$.
Using Lemma~\ref{lemm:Sum_bar_J}, we obtain 
 \begin{align}
 \label{bar_mJ_1-iH_Delta_t0}
 \bar{\mJ} (1-i H \Delta t) \le 1+ \tilde{J} \Delta t +  \sum_{m=2} \frac{(\Delta t)^m}{m!} \sum_{m_1=0}^m \binom{m}{m_1} \bar{\mJ} \br{H_{A_s}^{m_1} \otimes H_{B_s}^{m-m_1} } ,
\end{align} 
where we use $\bar{\mJ}(e^{-i(H_{A_s} + H_{B_s}) \Delta t})=1$ and $\bar{\mJ}(V_{A_sB_s}) \le \tilde{J}$.

Using the parameter $g$ in~\eqref{eq:Hdef}, we have 
 \begin{align}
\bar{\mJ} \br{H_{A_s}^{m_1} \otimes H_{B_s}^{m-m_1} } \le \norm{H_{A_s}}^{m_1} \times \norm{H_{B_s}}^{m-m_1} \le g^m |A_s|^{m_1} |B_s|^{m-m_1},
\end{align} 
which reduces the upper bound~\eqref{bar_mJ_1-iH_Delta_t0} to 
 \begin{align}
 \label{bar_mJ_1-iH_Delta_t}
 \bar{\mJ} (1-i H \Delta t) \le 1+ \tilde{J} \Delta t +  \sum_{m=2} \frac{(gn \Delta t)^m}{m!} 
 &= 1+ \tilde{J} \Delta t + \br{ e^{gn\Delta t} -  gn\Delta t -1} \notag \\
 &\le 1+ \tilde{J} \Delta t +(gn \Delta t)^2 ,
\end{align} 
where we use $ |A_s|+ |B_s|= |\Lambda|=n$ and the inequality~\eqref{error_estimation_Delta_t_H_exp} from the condition~\eqref{Delta_t_condition}.

By combining the inequalities~\eqref{tilde_Zeta_m} and \eqref{bar_mJ_1-iH_Delta_t}, we obtain 
\begin{align}
\label{tilde_Zeta_m_fin0}
\zeta_m \le \brr{1+ \tilde{J} \Delta t +(gn \Delta t)^2 } \tilde{\zeta}_m \le  \brr{1+ \tilde{J} \Delta t +(gn \Delta t)^2 } \zeta_{m-1}, 
\end{align} 
where we use \eqref{zeta_m+1_recursive} in the last inequality. 
By solving the above inequality with $\tilde{\zeta}_{m+1} \le \zeta_m$, we can derive 
\begin{align}
\label{tilde_Zeta_m_fin}
\zeta_m 
&\le \brr{1+ \tilde{J} \Delta t +(gn \Delta t)^2 }^{m+1} \zeta_0  \notag \\
&\le e^{\tilde{J} m \Delta t +(gn \Delta t)^2 m}   \zeta_0 \notag \\
&\le e^{\tilde{J} t + (gnt)^2 / \mN } 
\end{align} 
for $m\in [0,\mN-1]$, where we use $\zeta_0=1$, $\mN \Delta t =t$, and apply $m=\mN-1$ in the last inequality.

Therefore, from the upper bound~\eqref{tilde_Zeta_m_fin} and the inequality~\eqref{upp_delta/m/s}, we arrive at the Schmidt-rank truncation error for $\br{1-i H \Delta t} \ket{\Mt_{m}(D)} $ as follows:
\begin{align}
\delta_{m,s} (D)   \le \frac{ e^{\tilde{J} t + (gnt)^2 / \mN }}{\sqrt{D}}  . 
\end{align} 
We thus prove the main inequality~\eqref{Thm:error_efficiency_t-DMRG/main1} from $\bar{\delta}_m(D) := \max_{s\in [n-1]} [\delta_{m,s} (D)]$.  
This completes the proof. $\square$

{~}

\hrulefill{\bf [ End of Proof of Theorem~\ref{Thm:error_efficiency_t-DMRG}]}

{~}

\subsection{Quantum Gibbs states}  \label{Sec:Quantum Gibbs states}

In considering the quantum Gibbs states, we utilize the following purification:
\begin{align}
\ket{\rho_\beta} = \frac{1}{Z_\beta} \br{e^{-\beta H/2} \otimes \hat{1}_{\Lambda'} } \sum_{j=1}^{\mD_\Lambda} \ket{j_\Lambda} \otimes \ket{j_{\Lambda'}}
\end{align}  
with $\Lambda'$ a copied system, 
where $\{\ket{j_\Lambda}\}_j$ ($\{\ket{j_{\Lambda'}}\}_j$ are arbitrary orthonormal bases on $\Lambda$ ($\Lambda'$), and the normalization factor $Z_\beta$ has been defined as the partition function. 
Note that $\tr_{\Lambda'} \br{\ket{\rho_\beta} \bra{\rho_\beta}} = \rho_\beta$. 

Our task is to implement the imaginary-time evolution to $\sum_{j=1}^{\mD_\Lambda} \ket{j_\Lambda} \otimes \ket{j_{\Lambda'}}$~\cite{PhysRevX.11.011047}. 
However, we cannot directly apply the analyses of the real-time evolution.
The primary challenge occurs in the analyses of~\eqref{bar_mJ_1-iH_Delta_t__decomp}, which characterized the SE strength of $ \bar{\mJ} (1-i H \Delta t)$. 
In extending it to the imaginary-time evolution, we need to consider $ \bar{\mJ} (1-H \Delta \tau)$ with $\Delta \tau$ the decomposed unit of the total imaginary time $\beta$. 
We then replace Eq.~\eqref{bar_mJ_1-iH_Delta_t__decomp} with 
\begin{align}
 \label{bar_mJ_1-iH_Delta_t__decomp_beta}
1- H \Delta \tau  
= e^{-(H_{A_s} + H_{B_s}) \Delta \tau } - V_{A_sB_s} \Delta \tau  +  \sum_{m=2}^{\infty} \frac{(-\Delta \tau )^m}{m!}(H_{A_s} + H_{B_s})^m .
\end{align}  
Then, the problem is that the SE strength of $\bar{\mJ}(e^{-(H_{A_s} + H_{B_s})  \Delta \tau})$ is not equal to $1$ unlike the real-time case; instead, we have  $\bar{\mJ}(e^{-(H_{A_s} + H_{B_s})  \Delta\tau}) \propto 1 + n  \Delta \tau$. This also modifies the inequality~\eqref{tilde_Zeta_m_fin} by 
\begin{align}
\zeta_m \le e^{\tilde{J} \beta + (gn\beta)^2 / \mN + \orderof{n\beta}}  ,
\end{align}  
which is meaningless for a large system size $n$. 

Currently, we can only prove the existence of an efficient MPS approximation of $\ket{\rho_\beta}$ as follows: 
\begin{theorem} \label{poly_approx:MPO_gibbs}
There exists an efficient MPS approximation $\ket{\Mt_\beta(D)}$ for $\ket{\rho_\beta}$ such that 
\begin{align}
\label{main_ineq_Gibb_Poly_approx}
\norm{  \ket{\rho_\beta} - \ket{\Mt_\beta(D)}} \le 960 n \ceil{\beta\mQ_0/4}D^{-1/\kappa_\beta} ,
\end{align} 
where we defined $\kappa_\beta$ and $\mQ_0$ as 
\begin{align}
\label{notation_kappa_beta}
\kappa_\beta := 4\brr{6+\frac{4(k+1)}{\eta-2}+ 2 k \log_2(d)} \ceil{\beta\mQ_0/4},
\end{align}
and
\begin{align}
\mQ_0:= \brr{ \min \br{ \frac{1}{8gk} ,\frac{\eta-2}{16 e J_0 (\eta-1)^2 2^{\eta-2}} } }^{-1}  .
\end{align}
\end{theorem}

{\bf Remark.} 
From the definition, we have $\kappa_\beta \propto \beta$, and hence to achieve an error $\epsilon$, we need to set the bond dimension $D$ to be as large as $(n/\epsilon)^{\orderof{\beta}}$.  
This is qualitatively better than the state-of-the-art result of $e^{\orderof{\beta} \log^3(n/\epsilon)}$ in Ref.~\cite{PhysRevLett.134.190404}. 
On the other hand, it is still an open question whether one can also prove the polynomial time complexity to find such an MPS approximation. 
Even under the high-temperature condition (i.e., $\beta\ll 1$), the current best time complexity is given by  $e^{\log^2(n/\epsilon)}$~\cite{HighT_Alhambra}.

\subsubsection{Proof of Theorem~\ref{poly_approx:MPO_gibbs}}

We utilize the inequality~\eqref{Verstraete_Cirac_PRB} as in the proof of Proposition~\ref{Prop:MPO_approximation}.
Let us denote the Schmidt decomposition of $\ket{\rho_\beta}$ by
\begin{align}
 \ket{\rho_\beta} = \sum_{j} \lambda_j^{(s)}\ket{\phi_{A_s,j}}  \otimes \ket{\phi_{B_s,j}} ,
\end{align}
where the index $s$ characterizes the decomposition of the total system, i.e.,  $\Lambda= A_s \sqcup B_s$ with $A_s=\{1,2,\ldots, s\}$ and $B_s=\Lambda \setminus A_s$. 
Then, we obtain the same inequality as~\eqref{Verstraete_Cirac_PRB__again}: 
\begin{align}
\label{Verstraete_Cirac_PRB__again_again}
\norm{  \ket{\rho_\beta} -\ket{\Mt_\beta(D)}}^2 \le 2\sum_{s=1}^{n-1} \delta_s^2(D) ,
\end{align} 
and 
\begin{align}
\label{definition_delta_s_D_again_again} 
\delta_s^2(D) = \sum_{j>D}\br{ \lambda^{(s)}_j}^2 . 
\end{align} 

To estimate Eq.~\eqref{definition_delta_s_D_again_again}, we utilize the Eckart--Young theorem as follows: 
\begin{align}
\label{Eckart--Young_beta}
 \sum_{j>D}\br{ \lambda^{(s)}_j}^2 \le \norm{ \ket{\rho_\beta} -  \ket{\psi_D}  }^2 , 
\end{align} 
where $\ket{\psi_D} $ is an arbitrary quantum state with the Schmidt rank $D$ for the cut $A_s \sqcup B_s$. 
To construct $\ket{\psi_D}$, we approximate $e^{\beta H/4}$ by $\tilde{\rho}_{\beta/4}$ such that 
\begin{align}
\label{approx_e_beta_H/4}
\norm{ e^{\beta H/4} - \tilde{\rho}_{\beta/4}}_p \le \bar{\epsilon} \norm{e^{\beta H/4}}_p .
\end{align} 
Then, for the quantum state of
\begin{align}
 \ket{\tilde{\rho}_\beta} =\frac{\tilde{\rho}_{\beta/4}^\dagger \tilde{\rho}_{\beta/4} \otimes \hat{1}_{\Lambda'} }{\tr \brr{ \tilde{\rho}_{\beta/4}^\dagger \tilde{\rho}_{\beta/4} \tilde{\rho}_{\beta/4} \tilde{\rho}_{\beta/4}^\dagger } }  \sum_{j=1}^{\mD_\Lambda} \ket{j_\Lambda} \otimes \ket{j_{\Lambda'}} ,
\end{align}
we have 
\begin{align}
\label{approx_rho_beta_tilde_rho}
\norm{  \ket{\rho_\beta} - \ket{\tilde{\rho}_\beta} }^2 \le 10 \bar{\epsilon} , 
\end{align}
where the inequality is derived in Ref.~\cite[Ineq. (60) therein]{PhysRevX.11.011047}. 
By applying $\ket{\tilde{\rho}_\beta} $ to the inequality~\eqref{Eckart--Young_beta}, we need to set $D=\brr{{\rm SR} (\tilde{\rho}_{\beta/4})}^2$, which yields 
\begin{align}
\label{approx_rho_beta_tilde_rho_2}
 \sum_{j>\brr{{\rm SR} (\tilde{\rho}_{\beta/4})}^2}\br{ \lambda^{(s)}_j}^2 \le 10 \bar{\epsilon}. 
\end{align} 

The remaining problem is to get the Schmidt rank ${\rm SR} (\tilde{\rho}_{\beta/4})$ for the approximation~\eqref{approx_e_beta_H/4}. 
For this purpose, we use Theorem~\ref{thm:Schmidt_rank_truncation}. 
Here, the boundary interaction $V_{A_sB_s}$ is decomposed as 
\begin{align}
\label{decomp_V_AB_r_m}
V_{A_sB_s}= \sum_{Z: Z\cap A_s\neq \emptyset,\ Z\cap B_s\neq \emptyset} h_Z = \sum_{j=1}^\infty V_j ,
\end{align}
where we correspond each of $h_Z$ to $V_j$ in the decomposition~\eqref{Decomposition_V_A_B_assump_again} for $V_{AB}$. 
To apply the theorem, we first prove that the boundary interaction $V_{A_sB_s}$ satisfies Assumption~\ref{Assumpt_Schmidt_rank_truncation} in the following sense:
\begin{lemma} \label{Lem:Long-range_decomposition}
Under the decomposition of Eq.~\eqref{decomp_V_AB_r_m}, the operator $V_{A_sB_s}$ ($\forall s \in [1,n-1]$) satisfies the conditions~\eqref{Decomposition_V_A_B_assump} and \eqref{primary_assumption_V_AB} by choosing $\kappa$, $D_0$, $\tilde{g}$ and $C_0$ as follows:
\begin{align}
\label{Lem:Long-range_decomposition/main}
\kappa= \frac{\eta-2}{k+1}, \quad D_0 = d^{2k} , \quad \tilde{g}\le 4J_0 \br{1+\frac{1}{\eta-2}}, \quad C_0= (\eta-1)2^{\eta-2} .
\end{align}
Note that $\eta$ was set to be larger than $2$ as in~\eqref{def__long_range}.  
\end{lemma}

Moreover, by using Ref.~\cite[Lemma~3 therein]{KUWAHARA201696}, we have 
\begin{align}
\norm{\ad_{H_0}^s (h_Z)}  \le (2gk)^s s!  \norm{h_Z} , 
\end{align}
and hence we also obtain the parameter $\mQ$ in \eqref{introduce_parameter_mQ} as 
\begin{align}
\label{Eq/mQ_2gk}
\mQ= 2gk.  
\end{align}
We recall that the parameter $g$ is defined as one-site energy as in Eq.~\eqref{eq:Hdef}. 
The equations~\eqref{Lem:Long-range_decomposition/main} and \eqref{Eq/mQ_2gk} make
\begin{align}
\mQ_0:= \brr{ \min \br{ \frac{1}{4\mQ} ,\frac{1}{4 e C_0 \tilde{g}} }}^{-1} = \brr{ \min \br{ \frac{1}{8gk} ,\frac{\eta-2}{16 e J_0 (\eta-1)^2 2^{\eta-2}} } }^{-1}  .
\end{align}

Therefore, from the inequality~\eqref{thm:Schmidt_rank_truncation/main}, we obtain 
\begin{align}
 {\rm SR}(\tilde{\rho}_{\beta/4}) 
 & \le \br{ \frac{48 \ceil{\beta\mQ_0/4}}{\bar{\epsilon}}}^{2\brr{6+4(k+1)/(\eta-2)+ 2 k \log_2(d)} \ceil{\beta\mQ_0/4}}  \notag \\
 &= \br{ \frac{48 \ceil{\beta\mQ_0/4}}{\bar{\epsilon}}}^{\kappa_\beta/2},
\end{align}
where we use the notation $\kappa_\beta$ in Eq.~\eqref{notation_kappa_beta}. 
By letting $\brr{{\rm SR}(\tilde{\rho}_{\beta/4})}^2 =D$, we have 
\begin{align}
\bar{\epsilon} \le 48 \ceil{\beta\mQ_0/4}D^{-1/\kappa_\beta}.
\end{align}
Combining the above inequality with~\eqref{approx_rho_beta_tilde_rho_2}, we reduce the inequality~\eqref{definition_delta_s_D_again_again} to 
\begin{align}
\delta_s^2(D) = \sum_{j>D}\br{ \lambda^{(s)}_j}^2 \le 480 \ceil{\beta\mQ_0/4}D^{-1/\kappa_\beta}. 
\end{align}
From the inequality~\eqref{Verstraete_Cirac_PRB__again_again}, we finally obtain the main inequality~\eqref{main_ineq_Gibb_Poly_approx}. 
This completes the proof. $\square$

\subsubsection{Proof of Lemma~\ref{Lem:Long-range_decomposition}}
For simplicity, we denote $A_s$ and $B_s$ by $A$ and $B$, respectively. 
We aim to derive the inequalities of 
  \begin{align}
 \label{Decomposition_V_A_B_assump_again}
V_{AB} =\sum_{j=1}^\infty V_j , \quad {\rm SR}(V_j) \le D_0  \quad \forall j ,
\end{align}
and 
 \begin{align}
 \label{primary_assumption_V_AB_again}
\sum_{j\ge D+1} \norm{V_j} \le C_0 \tilde{g} (D+1)^{-\kappa} \quad (C_0\ge 1), \quad  \tilde{g}:= \sum_{j=1}^\infty \norm{V_j} ,
\end{align}
under the choice of~\eqref{Lem:Long-range_decomposition/main}. 

We first upper-bound the Schmidt rank ${\rm SR}(h_Z)$ for an arbitrary interaction term $h_Z$ with $|Z|\le k$.
For an arbitrary local operator defined on a subset $X$, the total number of operator bases is upper-bounded by $d^{2|X|}$, and hnece 
\begin{align}
{\rm SR}(h_Z) \le d^{2k} . 
\end{align}
Therefore, we can choose $D_0$ as 
\begin{align}
D_0=  d^{2k} . 
\end{align}

Second, we introduce the following decomposition of $V_{AB}$:
\begin{align}
&V_{AB}= \sum_{r=1}^\infty V_{r}  ,  \quad V_{r} := \sum_{\substack{Z: Z\cap A\neq \emptyset,\ Z\cap B\neq \emptyset \\ \diam(Z)= r}} h_Z .
\end{align}
We also denote the set $\mS_r$ such that 
\begin{align}
\mS_r:=\brrr{ Z\subset \Lambda:   |Z|\le k ,\   \diam(Z)= r ,\ Z\cap A\neq \emptyset,\ Z\cap B\neq \emptyset }.
\end{align}
Using it, we obtain
\begin{align}
V_{r} := \sum_{Z \in \mS_r} h_Z .
\end{align}
We define the subset $X_r$ ($r\in \mathbb{N}$) to span the region within the distance $r$ from the boundary between $A$ and $B$. 
Then, any subset $Z \in \mS_r$ is supported on $X_r$, and hence we get 
\begin{align}
|\mS_r| \le \binom{|X_r|}{k} = \binom{2r}{k} \le (2r)^k . 
\end{align}

We next estimate the norm summation of 
\begin{align}
V_{r} := \sum_{Z \in \mS_r} \norm{h_Z} .
\end{align}
By using the inequality~\eqref{def_interaction_decay} with \eqref{def__long_range}, we obtain 
\begin{align}
\label{Upper_bound_overline/_V_rm_2}
\sum_{\substack{Z: Z\cap A\neq \emptyset,\ Z\cap B\neq \emptyset \\ \diam(Z) =r}}\norm{h_Z} 
& \le  \sum_{i\in X_r} \sum_{j: \dist_{i,j} =r}  \sum_{Z: Z\ni \{i,j\}} \norm{h_Z} \notag \\
& \le 2 J_0 \sum_{i\in X_r} r^{-\alpha}  \le 4 J_0r^{-\eta+1} ,
\end{align}
where we use $|X_r| \le 2r$ in the last inequality. 

After straightforward calculations, we obtain 
\begin{align}
\sum_{r=1}^\infty \sum_{Z \in \mS_r} \norm{h_Z} \le \sum_{r=1}^\infty 4 J_0r^{-\eta+1} \le 4J_0 \br{1+\frac{1}{\eta-2}}=:\tilde{g}_0,
\end{align}
and 
\begin{align}
\sum_{r\ge R}^\infty \sum_{Z \in \mS_r} \norm{h_Z} \le \sum_{r\ge R} 4 J_0r^{-\eta+1} \le \frac{4J_0}{\eta-2} (\eta-1) R^{-(\eta-2)}  ,
\end{align}
where we use the upper bound of $[(D+1)/D]^{\eta-2}\le 2^{\eta-2}$ for $D\ge 1$. 
The total number of interaction operators $\{h_{Z}\}_{Z \in \mS_r, r< R}$ is smaller than 
\begin{align}
\sum_{r< R} |\mS_r| \le \sum_{r< R} (2r)^k \le 2^k (R-1)^{k+1} , 
\end{align}
and hence, we formally obtain 
\begin{align}
\sum_{j \ge D+1} \norm{V_j} & \le \sum_{r\ge R}^\infty \sum_{Z \in \mS_r} \norm{h_Z} \le \frac{4J_0}{\eta-2} (\eta-1) R^{-(\eta-2)} \notag \\
&\le \tilde{g}_0 (\eta-1) R^{-(\eta-2)}  \for  2^k (R-1)^{k+1} \le D < 2^k R^{k+1} .
\end{align}
By using $R\ge 2^{-k/(k+1)} D^{1/(k+1)} \ge D^{1/(k+1)}/2$, we reduce the above inequality to 
\begin{align}
\sum_{j \ge D+1} \norm{V_j}  \le (\eta-1)2^{\eta-2}  \cdot \tilde{g}_0 \cdot D^{-(\eta-2)/(k+1)} .
\end{align}
Therefore, we can derive the inequality~\eqref{primary_assumption_V_AB_again} under the choice of Eq.~\eqref{Lem:Long-range_decomposition/main}.
This completes the proof. $\square$

\section{Systems with short-range interactions} \label{sec:System with short-range interactions}

Finally, let us turn to geometrically local systems, namely those governed by short-range interactions, corresponding to the Hamiltonian defined in Eq.~\eqref{def_short_range}. 
Regarding the instantaneous R\'enyi entanglement rate, we show that the main theorem (i.e., Theorem~\ref{thm:Renyi_SIE}) cannot be improved even in this setting. 
On the other hand, when considering the entanglement rate over a finite time span, it has been known that substantial improvements are possible~\cite{PhysRevLett.97.157202,PRXQuantum.2.040331,PhysRevX.11.011047,PhysRevA.109.042404,ycdh-z8zf} compared with the general setup (i.e., Propositions~\ref{prop:optimality_1} and \ref{prop:optimality_2}). 
In what follows, we present two representative results: a trivial toy model (Lemma~\ref{lem:toy_short_range}) that saturates the upper bound~\eqref{lemm:Renyi_SIE_main_ineq} for the instantaneous R\'enyi entanglement rate at $t=0$, and a state-of-the-art bound on the finite-time entanglement generation (Lemma~\ref{trivial_lem:upp_short_range}).

\begin{lemma} 
\label{lem:toy_short_range}
Let us consider the following simple case in a 2-qubit system: 
\begin{align}
\ket{\psi_0} = \ket{0,0} ,   \quad H= \ket{0,0}\bra{1,1} + {\rm h.c.}\ .
\end{align}
Note that $\bar{\mJ}(H)=1$. 
Then, the R\'enyi entanglement rate satisfies 
\begin{align}
\label{diverging_entanglement_rate}
&\lim_{t\to +0}\abs{ \frac{dE_{\alpha}(t)}{dt} }  =\begin{cases} 
\infty  &\for  0 < \alpha < 1/2,  \\
2  &\for  \alpha = 1/2,  \\
0  &\for  \alpha >1/2.
\end{cases}	
\end{align}
\end{lemma}

\textit{Proof of Lemma~\ref{lem:toy_short_range}.}
By solving the eigen-problem of $H$, we obtain 
\begin{align}
\ket{\psi_t} = e^{-iHt}\ket{0,0} =\cos(t) \ket{0,0} - i \sin(t) \ket{1,1} , 
\end{align}
which yields the exact solution for $t\le \pi/2$
\begin{align}
\label{diverging_entanglement_rate__0}
\frac{dE_{\alpha}(t)}{dt} 
&= \frac{2\alpha}{1-\alpha}\frac{\cos(t)\sin^{2\alpha-1}(t) - \sin(t)\cos^{2\alpha-1}(t)}{\cos^{2\alpha}(t)+\sin^{2\alpha}(t)} \notag \\
&\approx \frac{2\alpha}{1-\alpha} (t^{2\alpha-1} - t)  .
\end{align}
The above equation immediately leads to the main statement~\eqref{diverging_entanglement_rate}.
This completes the proof. $\square$

We also demonstrate that the average R\'enyi entanglement rate is upper-bounded by $\orderof{1/\alpha}$ for a finite time-interval: 
\begin{lemma} [A modified version of Corollary~8 in Ref.~\cite{PhysRevX.11.011047}]
\label{trivial_lem:upp_short_range}
Let $H_{AB}$ be a Hamiltonian with short-range interactions. 
For arbitrary time‑evolution operator $e^{-iH_{AB}t}$, there exists an operator $U_{AB,D}$ such that
$\norm{ U_{AB,D}  - e^{-iH_{AB}t}} \le \epsilon$, 
where $U_{AB,D}$ has the Schmidt rank $D$ with  
 \begin{align}
D =  e^{\tilde{O}\brr{ \sqrt{t \log(1/\epsilon)+t^2} }} .
\end{align}
Note that $\tilde{O}(x)$ means $\orderof{x\log x}$. 
\end{lemma}

{\bf Remark.} 
Combining the lemma with the Eckart--Young theorem~\eqref{Eckart--Young_upper_bound}, one can estimate the decay rate of the Schmidt coefficients $\{\lambda_s\}_s$ ($\lambda_1 \ge \lambda_2 \ge \cdots $) after time evolution, which is given by  
 \begin{align}
 s \le e^{\tilde{O}\brr{ \sqrt{t \log(1/\lambda_s)+t^2} }} \quad \longrightarrow \quad 
\lambda_s \le \exp\br{C_1 t - \frac{C_2 \log^2(s)}{t \log\log(s)}}.
\end{align}
This gives the generation of the R\'enyi entanglement of the form
 \begin{align}
E_{\alpha} (t) - E_{\alpha} (0) =\frac{ \tilde{\mathcal{O}} (t)}{\alpha} ,
\end{align}
which is finite for an arbitrary non-zero $\alpha$.
A similar analyses can be applied to the high-dimensional cases, where we need an additional coefficient that is proportional to the surface size of $A$.

\begin{figure}[ttt]
  \centering
  \includegraphics[width=0.7\textwidth]{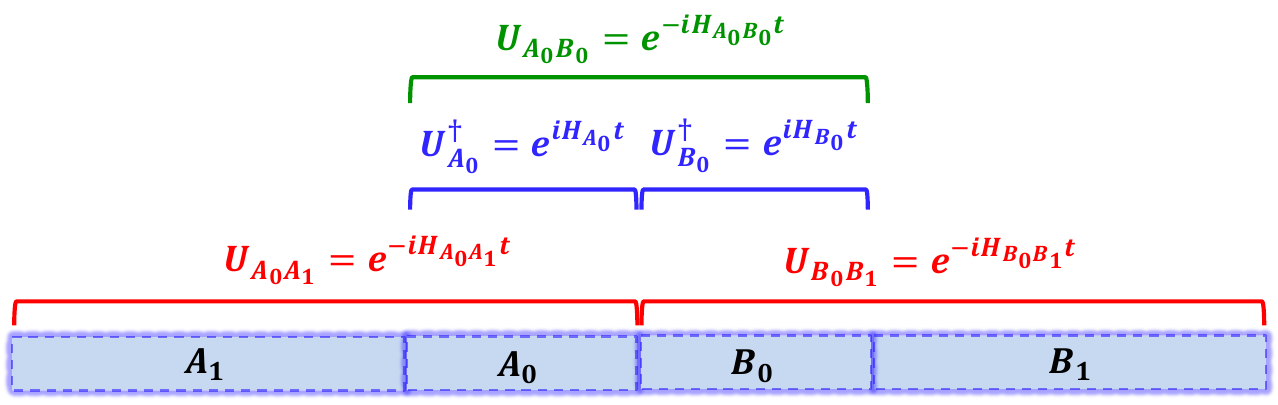}
  \caption{Schematic picture of the Haah-Hastings-Kothari-Low approximation along the cut $A_0A_1 \sqcup B_0B_1$.  
  }
  \label{fig:c_HHKL}
\end{figure}

\textit{Proof sketch.} 
For the reader’s convenience, we sketch the proof in the spirit of Ref.~\cite{PhysRevX.11.011047}.
The key ideas are: (i) decomposing the time evolution using the Haah–Hastings–Kothari–Low (HHKL) approximation~\cite{HHKL2021}; and (ii) applying the Schmidt-rank bounds of Refs.~\cite{PhysRevB.85.195145,arad2013area} to the Taylor expansion of the time-evolution operator restricted to the central block.

We first decompose the total system into $\Lambda= A_1 \sqcup A_0 \sqcup B_0 \sqcup B_1$, where $|A_0|=|B_0|=\ell$. 
Then, one can obtain the Haah-Hastings-Kothari-Low approximation of the time evolution $e^{-iH_{AB}t}$ as follows~\cite{HHKL2021} (Fig.~\ref{fig:c_HHKL}):
 \begin{align}
\norm{ e^{-iH_{AB} t}-  e^{ -i H_{A_0 B_0} t} e^{i (H_{A_0} + H_{B_0}) t}  e^{-i (H_A + H_B) t}} \le C e^{- \mu \ell + vt} , 
\end{align}
where $C,v,\mu$ is determined by the Lieb--Robinson bound. We recall that $H_{A_0 B_0}$, $H_{A_0}$ and $H_{B_0}$ are the subset Hamiltonians on $A_0\sqcup B_0$, $A_0$ and $B_0$, respectively. 
Note that if the higher dimension is considered, the above upper bound depends on the surface area of $A$.  

In order to ensure the approximation error $\epsilon$, we have to choose $\ell \propto \log(1/\epsilon)+vt$.
Now, the unitary operator $e^{i (H_{A_0} + H_{B_0}) t} e^{-i (H_A + H_B) t}$ has a product form with respect to the decomposition of $A$ and $B$, and hence we only have to consider the influence from $e^{ -i H_{A_0 B_0} t} $.
By considering the Taylor expansion as 
 \begin{align}
e^{ -i H_{A_0 B_0} t}  = \sum_{m=0}^\infty \frac{(-it)^m}{m!}H_{A_0 B_0}^m ,
\end{align}
we can truncate the expansion such that $m \le M$ with $M\propto t\norm{H_{A_0 B_0}} \propto t \log(1/\epsilon)+t^2$, where we use $\norm{H_{A_0 B_0}} \propto \ell$. 
Also, the Schmidt rank of $H_{A_0 B_0}^m$ is given by $m^{\orderof{\sqrt{m}}}$~\cite{PhysRevB.85.195145,arad2013area}. 
Therefore, we can ensure that $\sum_{m=0}^M \frac{(-it)^m}{m!}H_{A_0 B_0}^m$ has a Schmidt rank of order of $M^{\sqrt{M}} = e^{\tilde{O}\brr{ \sqrt{t \log(1/\epsilon)+t^2} }}$. 
This completes the proof. $\square$

\section{Conclusion and perspective}

In this work, we aimed to establish universal principles governing the entanglement spectrum. The key notions introduced and developed are the \emph{SE strength} (Definition~\ref{Def:Interaction_strength}) and the \emph{spectral SIE} (Theorem~\ref{thm:Renyi_SIE}). They provide a refined and quantitatively optimal framework for dynamical entanglement structure, thereby extending the celebrated SIE theorem to address richer many-body phenomena.
As an application, we have analyzed long-range interacting systems and demonstrated that, contrary to previous expectations, the quasi-polynomial complexity of their simulation can in fact be improved to polynomial complexity. These main achievements are summarized in Fig.~\ref{fig:Overview}.

At the same time, our analysis naturally brings to light a variety of fundamental open problems. These problems indicate both the limitations of the present results and promising avenues for further exploration. In the following, we list and discuss several of these open directions.

\begin{enumerate}

  \item \textbf{Unification of spectral and standard SIEs}:
%
The spectral SIE developed in this work has the distinct advantage that it provides bounds on the generation rate of R\'enyi entanglement for $\alpha<1$, a feature not available in the standard SIE. This refinement allows for a more precise control of computational complexity. However, an important limitation arises: the spectral SIE, when specialized to $\alpha = 1$, does not straightforwardly recover the standard SIE. As emphasized in Ref.~\cite{PhysRevA.109.042404}, the case $\alpha=1$ exhibits special properties that set it apart from the regime $\alpha<1$. Consequently, a major open challenge is to develop a unified theorem that fully integrates the spectral and standard SIE frameworks into a comprehensive theory of dynamical entanglement growth. At the present stage of our understanding, our results show that the bound is sharp only for $\alpha=1/2$, whereas the optimal bound for $\alpha>1/2$ remains unknown.

\item \textbf{Limits of operator approximability}:  
Proposition~\ref{No_go_theorem_approx} shows that small SE strength does not guarantee an efficient approximation of the time-evolution operator, thereby revealing a complexity separation between the full unitary operator and time-evolved states.
However, once we allow the Schmidt rank to depend on the Hilbert space dimension $\mD_{AB}$, it remains unresolved whether the required scaling is polynomial in $\mD_{AB}$ or only logarithmic in $\mD_{AB}$. 
Clarifying this point is crucial for achieving a more \emph{quantitative} understanding of complexity separation in many-body quantum dynamics.

  \item \textbf{Sufficient conditions for operator-level low-rank approximation (Conjecture~\ref{conj:operator_approx})}:  
  This problem is closely related to the problem discussed above. 
  While small SE strength is known to be insufficient to guarantee low-rank approximations at the operator level, it is conjectured that small $\alpha$-SE strength $\bar{\mJ}_\alpha$ (Def.~\ref{Def:Interaction_strength_renyi}) might suffice. Proving this conjecture and determining the exact functional dependence $g_\alpha(D)$ would close a major gap between state-level and operator-level approximability.

\item \textbf{Precise conditions for generalized area laws}:  
The generalized area law currently relies on the assumption of boundary-adiabatic paths (Assumption~\ref{assump:Boundary-adiabatic path}). 
It remains unknown whether this assumption is indeed indispensable or can be relaxed. Determining the minimal assumptions under which area laws hold is essential for both condensed matter physics and quantum information theory.

  \item \textbf{Tightening the area-law upper bound}:  
 We obtained $E_{\alpha=1}(\Omega)\lesssim [\bar J(V_{AB})/\Delta]^3$ (Theorem~\ref{thm:generalized_area_law}), but it is natural to expect a linear dependence in the boundary size, i.e., $E_{\alpha=1}(\Omega)\lesssim \bar J(V_{AB})/\Delta$. Improving the exponent in the upper bound is therefore an important quantitative challenge, with implications for the efficiency of tensor network representations of gapped ground states.


  \item \textbf{Time-efficient algorithms for ground states and Gibbs states}:  
Current results provide existence proofs of MPS (or purified MPS) approximations to gapped ground states and Gibbs states in 1D long-range interacting systems (Theorems~\ref{thnm:Gs_approx} and \ref{poly_approx:MPO_gibbs}). However, explicit and efficient algorithms to actually construct such approximations are lacking. 
Developing such algorithms would make rigorous complexity results directly applicable to numerical simulation. More importantly, it would represent a genuine complexity-theoretic breakthrough, implying that the computational complexity of classical simulation of 1D long-range interacting systems lies in the class~P.


\item \textbf{Generalization of rigorous error certification for DMRG-type algorithms}:  
The certification scheme introduced for t-DMRG (Theorem~\ref{Thm:error_efficiency_t-DMRG}) may serve as a paradigm for broader classes of variational methods, including higher-dimensional tensor networks.
 Extending the framework, improving its dependence on system size, time, and error tolerance, and making it practical for real simulations are natural directions for future work. In particular, achieving rigorous error certification for imaginary time evolution remains one of the most important open problems.

\item \textbf{Dynamical entanglement structure in short-range interacting systems}:  
For short-range interacting systems, Lemma~\ref{trivial_lem:upp_short_range} shows that the R\'enyi entanglement can be upper bounded in the form of $1/\alpha$. This $\alpha$-dependence is likely to be optimal, as suggested by analyses based on conformal field theory~\cite{Calabrese_2005,Calabrese_2009,Peschel_2009} (see also \cite{Foot1}). However, the truly optimal time-dependence of the entanglement growth remains unresolved. Clarifying this point is essential for achieving a universal understanding of the computational complexity of simulating short-range interacting systems, both on classical and quantum devices.

\end{enumerate}

Altogether, the resolution of these open problems would mark significant progress toward a unified and quantitative theory of entanglement dynamics, with far-reaching implications for the simulation of quantum many-body systems and for the broader interface between physics and computer science.

\section*{Acknowledgments}

All the authors acknowledge the Hakubi projects of RIKEN.  T. K. is supported by JST PRESTO (Grant No. JPMJPR2116), Exploratory Research for Advanced Technology (Grant No. JPMJER2302), and JSPS Grants-in-Aid for Scientific Research (No. JP23H01099, JP24H00071), Japan.
Y. K. was supported by the JSPS Grant-in-Aid for Scientific Research (No. JP24K06909), Japan.
The research of C.R. is supported by the DFG through the grant TRR 352 -- Project-ID 470903074.
This research was conducted during the internship of H.~M. and C.~R. at RIKEN, under the supervision of T.~K.

\bibliography{Spectral_SIE}
%
%
%





\end{document}